\documentclass[letterpaper,aps,pra,twocolumn,superscriptaddress,tightenlines,showpacs]{revtex4-1}

\usepackage[T1]{fontenc}
\usepackage{lmodern} 
\usepackage{graphicx,color}
\usepackage{latexsym,amsmath,amssymb,amsfonts,amsthm,mathrsfs,bm}
\usepackage{enumerate}

\usepackage{verbatim}
\usepackage{epstopdf}
\usepackage[titletoc,title]{appendix}
\usepackage{tikz}

\theoremstyle{plain}
\newtheorem{thm}{Theorem}[section]
\newtheorem{prop}[thm]{Proposition}

\newtheorem{lem}[thm]{Lemma}

\theoremstyle{definition}
\newtheorem{defn}[thm]{Definition}

\newtheorem{exmp}[thm]{Example}

\theoremstyle{remark}

\newcommand{\tr}[2]{\text{Tr}_{#1}\left(#2\right)}

\newcommand{\ket}[1]{| #1 \rangle}
\newcommand{\bra}[1]{\langle #1 |}

\newcommand{\inprod}[2]{\bra{#1}#2\rangle}
\newcommand{\ketbra}[1]{\ket{#1}\bra{#1}}

\newcommand{\identity}{\mathbb{I}}
\newcommand{\hilbert}{\mathcal{H}}
\newcommand{\Hi}{\mathcal{H}}

\newcommand{\trn}[1]{\text{Tr}_{#1}}
\newcommand{\complex}{\mathbb{C}}

\newcommand{\neigh}{\mathcal{N}}

\begin{document}

\title{Exact stabilization of entangled states in finite time by dissipative quantum circuits}

\author{Peter D. Johnson}
\affiliation{\mbox{Department of Physics and Astronomy, Dartmouth 
College, 6127 Wilder Laboratory, Hanover, NH 03755, USA}}
\affiliation{\mbox{Department of Chemistry and Chemical Biology
Harvard University, 12 Oxford Street, Cambridge, MA 02138, USA}}

\author{Francesco Ticozzi}
\affiliation{Dipartimento di Ingegneria dell'Informazione,
Universit\`a di Padova, via Gradenigo 6/B, 35131 Padova, Italy} 
\affiliation{\mbox{Department of Physics and Astronomy, Dartmouth 
College, 6127 Wilder Laboratory, Hanover, NH 03755, USA}}

\author{Lorenza Viola}
\affiliation{\mbox{Department of Physics and Astronomy, Dartmouth 
College, 6127 Wilder Laboratory, Hanover, NH 03755, USA}}

\begin{abstract}
Open quantum systems evolving according to discrete-time dynamics are capable, unlike continuous-time counterparts, to converge to a stable equilibrium in finite time with {\em zero error}. We consider dissipative quantum circuits consisting of sequences of quantum channels subject to specified quasi-locality constraints, and determine conditions under which stabilization of a pure multipartite entangled state of interest may be exactly achieved in finite time. Special emphasis is devoted to characterizing scenarios where finite-time stabilization may be achieved {\em robustly} with respect to the order of the applied quantum maps, as suitable for unsupervised control architectures. We show that if a decomposition of the physical Hilbert space into virtual subsystems is found, which is compatible with the locality constraint and relative to which the target state factorizes, then robust stabilization may be achieved by independently cooling each component. We further show that if the same condition holds for a scalable class of pure states, a continuous-time quasi-local Markov semigroup ensuring {\em rapid mixing} can be obtained. Somewhat surprisingly, we find that the commutativity of the canonical parent Hamiltonian one may associate to the target state does {\em not} directly relate to its finite-time stabilizability properties, although in all cases where we can guarantee robust stabilization, a (possibly non-canonical) commuting parent Hamiltonian may be found. Beside graph states, quantum states amenable to finite-time robust stabilization include a class of universal resource states displaying two-dimensional symmetry-protected topological order, along with tensor network states obtained by generalizing a construction due to Bravyi and Vyalyi. Extensions to representative classes of mixed graph-product and thermal states are also discussed.
\end{abstract}


\date{\today}
\maketitle 

\section{Introduction}
\label{sec:intro}

Convergence of a dynamical system to a stable equilibrium point can only arise from
irreversible, dissipative behavior.
For quantum dynamics, characterizing the stability properties of equilibrium states of both naturally 
occurring and controlled dissipative evolutions is a fundamental problem, whose significance ranges 
from mathematical aspects of open-quantum system theory \cite{davies,alicki-lendi} and non-equilibrium 
quantum statistical mechanics \cite{JaksicReview,EisertReview}, to dissipative quantum control and quantum 
engineering \cite{Zoller1996,Viola2001,Rabitz2007,TicozziTutorial}. Within quantum information processing 
(QIP) \cite{Nielsen-Chuang:10}, a main motivation for investigating {\em quantum stabilization} problems is 
provided by the key task of preparing a target quantum state from an arbitrary initial condition. 
Notably, highly entangled pure states are a resource for measurement-based quantum computation 
\cite{Raussendorf2003, Miyake2011,Wei2015} as well as quantum communication technologies \cite{gisin}; 
likewise, the preparation of both ground and thermal (Gibbs) states of physically relevant Hamiltonians 
is a prerequisite for quantum simulation algorithms \cite{Lloyd1996,Ward2009,Temme2011,Brandao2016}.  
From a condensed-matter physics standpoint, methods for preparing many-body quantum states may unlock 
new possibilities for accessing exotic phases of synthetic quantum matter in controlled laboratory 
settings \cite{Diehl2008,Diehl2013}.

Compared to the standard approach to pure-state preparation \cite{Nielsen-Chuang:10} -- namely, 
the initialization of the system in a fiducial product state via a fixed (necessarily dissipative) ``cooling'' 
mechanism, followed by a unitary quantum circuit -- the use of {\em tailored} dissipative dynamics affords 
two important practical advantages: not only is precise initialization no longer needed, but, any ``transient'' 
noise effect is effectively re-absorbed without the need for active intervention, as long as the target state is 
globally attractive. 
Crucially, the {\em invariance} requirement that the dissipative dynamics must obey for the target state to be 
not only prepared but, additionally, stabilized, allows for a further important advantage: once reached, the  
desired state may be accessed at any time afterward -- which is especially 
beneficial in scenarios where the retrieval time is not (or cannot) be precisely specified in advance. 
As a result, methods for engineering dissipation are garnering increasing attention across different experimental 
QIP platforms.  In particular, steady-state entanglement generation has been 
successfully demonstrated in systems as diverse as atomic ensembles \cite{Cirac2011}, trapped ions 
\cite{Barreiro-Nature:11,Yin2013}, superconducting qubits \cite{Devoret,Siddiqi}, and electron-nuclear spins in 
diamond \cite{Wrachtrup2016}.

It is important to appreciate that the problem of designing stabilizing dynamics is both physically relevant and 
mathematically non-trivial only in the presence of {\em constraints} on the available dynamical resources: 
if arbitrary completely-positive trace-preserving (CPTP) quantum maps \cite{kraus,Nielsen-Chuang:10} 
are able to be implemented, then any desired quantum state (pure or mixed) may be made 
invariant and attractive in a single time step \cite{remarkOnestep}.
Similarly, for continuous-time Markovian quantum dynamics described by a Lindblad master 
equation \cite{lindblad,alicki-lendi}, one may show that application of a time-independent Hamiltonian 
together with a single noise operator suffices to achieve stabilization in the generic case in principle 
\cite{ticozzi-markovian,ticozzi-schirmer}.
For multipartite quantum systems of relevance to both QIP and statistical mechanics, an important constraint 
stems from the fact that physical Hamiltonians and noise (Kraus or Lindblad) operators typically 
represent couplings that affect non-trivially a ``small'' number of subsystems at a time; 
mathematically, they are required to be {\em quasi-local} (QL) relative to the underlying tensor-product 
decomposition, in an appropriate sense. 
To date, significant theoretical effort has been devoted to investigating  
QL state stabilization problems under continuous-time Lindblad dynamics, 
both in specific physically motivated settings -- see e.g. \cite{Kraus2008,Diehl2008,Cirac2009,Sorensen2011,DallaTorre,Clerk,Reeb2016,WeimerRydberg,WeimerHubbard,Hartmann,Znidaric} 
for a partial list of contributions -- and within a general system-theoretic framework \cite{Ticozzi2012,TicozziQIC2014,switching,Johnson2016}.

In this work, we consider the problem of stabilizing a pure quantum state of interest using 
{\em time-dependent discrete-time dynamics}, 
as implemented by sequences of CPTP maps, subject to a specified QL constraint. 
Such a setting is most natural from a QIP perspective, as it embodies  
a dissipative quantum circuit picture that directly generalizes the unitary quantum circuit model 
and is ideally suited for ``digital'' open-system quantum simulation \cite{Barreiro-Nature:11,Blatt2013,Pan}; 
further to that, it is also fundamentally more general: it is well known that there exists indivisible CPTP
dynamics, which {\em cannot} be obtained from exponentiating continuous time-dependent Markovian 
dynamics \cite{WolfCirac}, as also emphasized in recent approaches to quantum channel construction \cite{Liang2016}.  Most importantly to our scope,
discrete-time dynamics support a different type of convergence to equilibrium with respect to continuous-time 
counterparts: {\em exact} convergence in {\em finite time}, as opposed to asymptotic convergence -- in which case 
the target state can be reached only {\em approximately} at any finite time and which, as we shall see, is the 
only possibility for Lindblad dynamics. 

While finite-time (or ``dead-beat'') controllers have been extensively 
analyzed and exploited in the context of classical digital control systems \cite{bellman,philips}, 
they have received far less attention in quantum engineering as yet.  
A general scheme for pure-state stabilization in finite time has been proposed in \cite{BT-TAC:10}; 
however, no QL constraint is explicitly incorporated and feedback capabilities are assumed. 
Building on our complementary analysis of asymptotic convergence properties of time-dependent 
sequences of CPTP maps in \cite{Ticozzi-alternating}, our main focus here is open-loop QL 
{finite-time stabilization} (FTS) of a target quantum state: 
What ensures that stabilization may be attained in finite time under the prescribed QL constraint?  
Further to that, what properties may enable FTS to be achieved {\em robustly}, 
in a way that is independent upon the order of implementation of the applied CPTP maps?  Clearly, the 
possibility of robust finite-time stabilization (RFTS) is especially appealing from both a control-theoretic 
and a practical perspective, as it allows for ``unsupervised'' control implementation or, equivalently, for the  
dissipative quantum circuit to be applied ``asynchronously'' -- thus recovering a desirable feature of 
time-independent stabilization schemes.

With the above questions in mind, the content of the paper and our main results may be summarized as 
follows. In Sec. \ref{sec:prelim} we introduce the necessary background and mathematical tools, with emphasis 
on spelling out the relevant stability notions.  In particular, we explicitly show (see Sec. \ref{sub:nogo}) that no 
continuous-time Lindblad master equation can converge exactly to a globally attractive equilibrium in finite time. 
In Sec. \ref{sec:fts} we develop both necessary and sufficient conditions for determining if a target state can, 
in principle, be FT-stabilized. In the case that a state is verified to be FTS, we explicitly demonstrate the 
existence of QL stabilizing dynamics, which entails the repeated application of a fixed cooling map, suitably 
interspersed with unitary dynamics.  We stress that, despite superficial similarity, FTS bears important differences
from dissipative quantum circuits implementing ``sequential generation'' \cite{Schon2005}, whereby the 
system of interest is sequentially coupled to an ancilla, and a matrix product state representation of the target 
state is used to obtain a sequence of CPTP maps as the ancilla is traced over: not only does the joint system plus 
ancilla pair require proper (pure-state) initialization, but no invariance is guaranteed in general. 
Rather, our FTS scheme may be thought of as a QL generalization of the ``splitting-subspace'' 
approach introduced in Ref. \cite{BT-TAC:10}.   

Sections \ref{sec:necessary} and \ref{sec:sufficient}, which form the core of the paper, are devoted to 
presenting several necessary and, respectively, sufficient conditions for RFTS. In particular, we show how  
RFTS requires the correlations present in the target state to be restricted in mathematically precise ways. 
Interestingly, while our necessary RFTS conditions bear similarity with criteria on clustering of correlations 
which have recently been proved to ensure efficient preparation of thermal (Gibbs) states using dissipative QL 
circuits \cite{Brandao2016disc}, a main difference is the invariance requirement on the target, which is central 
in our approach. As emphasized above and in \cite{Ticozzi-alternating}, one implication of the invariance 
property is that repeating the stabilization protocol (or even portions of it) may be used as a means to 
maintain the system in the target state over time, if so desired.  
In developing sufficient conditions for RFTS, we leverage the observation that product states 
are (trivially) RFTS to seek a description of the target state in terms of a {\em virtual subsystem decomposition}  
\cite{Knill2000,Zanardi2001}, relative to which it may factorize, in a sense that we make precise. 
We find that, counterintuitively, a pure state may be FT-stabilizable -- albeit {\em not} RFTS -- even when its 
``natural'', frustration-free QL parent Hamiltonian is {\em non-commuting}; at the same time, we also uncover 
examples of states, which {\em are} RFTS and whose natural parent Hamiltonian is non-commuting -- albeit in those 
cases a different, {\em commuting} parent Hamiltonian may also be identified. Beside providing conditions that ensure 
the RFTS task to be possible for a given target and locality constraints, our results may alternatively be 
used to construct classes of non-trivially entangled target states that are guaranteed to be RFTS for a given 
QL constraint.  In particular (see Sec. \ref{sub:nonconstructive}), we introduce a class of tensor network states 
\cite{Orus2014} that are RFTS relative to a QL structure determined by the underlying graph, upon generalizing a 
construction due to Bravyi and Vyalyi \cite{Bravyi2005} {\em beyond} the original two-body setting. While our 
primary focus throughout the present analysis is on pure target states, we isolate in Sec. \ref{sub:mixed} those 
results that are directly applicable or extend to mixed target states; in particular, we exhibit a class of RFTS Gibbs states. 

In Sec. \ref{sec:efficiency} we explore the efficiency of the proposed FTS schemes in both the non-robust and 
robust settings, by providing, in particular, an upper bound to the circuit complexity of RFTS protocols for QL 
constraints defined on a lattice.  Finally, since FT convergence is a particularly strong form of convergence, 
it is natural to explore the extent to which it may be related to ``rapidly mixing'' QL continuous-time dynamics, 
which is able to efficiently prepare an equilibrium state \cite{Reeb2012,Temme2014,Lucia2015}.  
In Sec. \ref{sec:rapid}, we indeed show that as long as the sufficient conditions for a target state to be RFTS are obeyed, there always exists a QL Liouvillian which is rapidly mixing with respect to the target. Nonetheless, it is possible for a state to admit rapidly mixing dynamics that asymptotically prepares it, while violating the necessary conditions for RFTS.  We conclude in Sec. \ref{sec:end} by highlighting open problems and directions for further investigation.  In order to progressively build insight and maintain continuity in the presentation flow, we have emphasized illustrative examples to the extent possible and deferred 
all of the technical proofs to an appendix at the end of the paper.


\section{Preliminaries}
\label{sec:prelim}

\subsection{Discrete-time quasi-local Markov dynamics }
\label{sub:dynamics}

We consider a finite-dimensional multipartite target system $S$, consisting of $N$ 
distinguishable subsystems and described on a Hilbert space $\hilbert\simeq \bigotimes_{i=1}^N 
\hilbert_i$, with dim$(\hilbert)\equiv D$ and each $\hilbert_i \simeq \complex^{d_i}$. 
We shall denote by $\mathcal{B}(\hilbert)$ the space of all linear operators on $\hilbert$. 
The state of $S$ at each time is a density operator in the space of positive-semidefinite, trace-one linear operators, 
denoted $\mathcal{D}(\hilbert) \subset \mathcal{B}(\hilbert)$. 
We assume the time evolution of $S$ to be modeled by non-homogeneous discrete-time Markov dynamics. 
Such dynamics are represented by sequences of CPTP maps $\{\mathcal{E}_t\}_{t\geq0}$ \cite{kraus}, 
whereby the evolution of the state $\rho_t$ from step $t$ to $t+1$ is given by
\begin{equation}
\rho_{t+1}=\mathcal{E}_{t}(\rho_t), \quad t =0,1,2,\ldots,
\label{eq:dyn}
\end{equation}
and we further denote the evolution map, or propagator, from $s$ to $t$ as
\begin{equation}
\mathcal{E}_{t,s} \equiv \mathcal{E}_{t-1}\circ\mathcal{E}_{t-2}\circ \ldots \circ \mathcal{E}_s, \quad t > s \geq 0. 
\label{eq:prop}
\end{equation}
In practice, a variety of constraints may restrict the available control, hence the set of possible quantum maps.
In particular, as mentioned, we require that each map acts \emph{quasi-locally}. 
Following our previous work \cite{Ticozzi2012,TicozziQIC2014,Johnson2016,Ticozzi-alternating}, the notion of 
quasi-locality we consider may be formally described by a \emph{neighborhood structure}, $\neigh$, on the multipartite Hilbert space. That is, $\neigh$ is specified by a list of subsets of indexes, $\neigh_k \subseteq \{1,\ldots,N\}$, for $k=1,\ldots,K$, encompassing a variety of physically relevant ``coupling topologies'' between subsystems (see also Fig. \ref{fig:nn}).

\begin{figure}[t]
\begin{center}
\includegraphics[width=0.85\columnwidth]{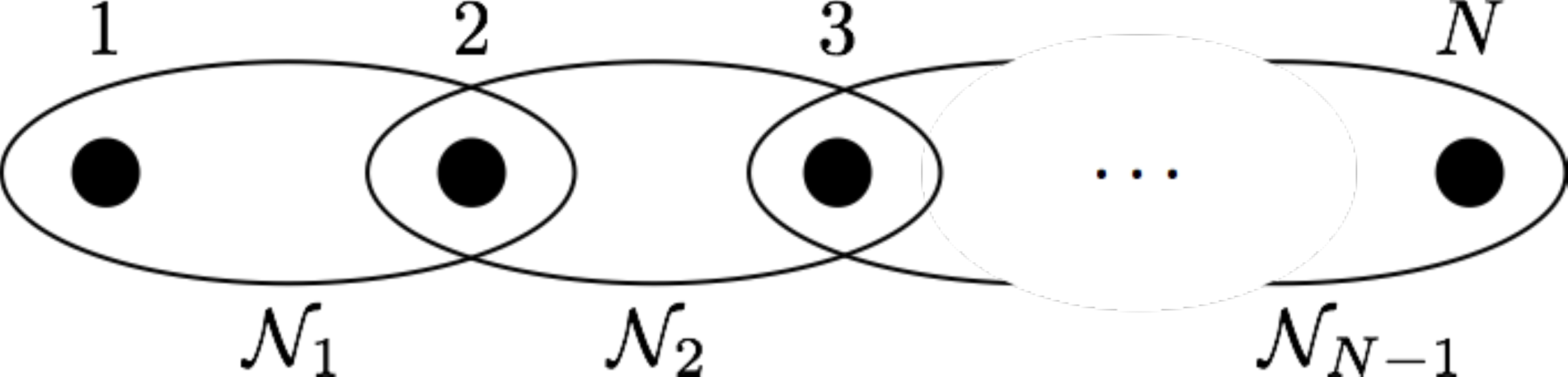}
\vspace*{-1mm}
\caption{Neighborhood structure corresponding to two-body nearest-neighbor (NN) couplings 
in one dimension (1D), ${\cal N}_j \equiv \{j, j+1\}$, $j=1,\ldots, N-1$.
Note that, except for the two boundary neighborhoods, associated to $j=1$ and $j=N-1$, a three-body neighborhood 
structure, ${\cal N}_j \equiv \{j-1,j, j+1\}$ corresponds instead to graph states on the line, 
as described in Example \ref{exmp:graphstate}.  }
\label{fig:nn} 
\end{center}
\end{figure}

\begin{defn}
\label{def:neigh-map}
A CPTP map $\mathcal{E}$ is a \emph{neighborhood map} with respect to a neighborhood  ${\neigh}_k$ if
\begin{equation}
\mathcal{E}=\mathcal{E}_{\neigh_k}\otimes {\mathcal I}_{\overline{\neigh}_k},
\end{equation}
where $\mathcal{E}_{\neigh_k}$ is the restriction of $\mathcal{E}$ to operators on the subsystems in $\neigh_k$ 
and ${\mathcal I}_{\overline{\neigh}_k}$ is the identity map for operators on $\hilbert_{\overline{\neigh}_k}$.  
The sequence $\{ \mathcal{E}\}_{t\geq 0}$ is \emph{quasi-local} with respect to a neighborhood structure $\neigh$ if, 
for each $t$, $\mathcal{E}_t$ is a neighborhood map for some $\neigh_k \in\neigh$.
\end{defn} 

A useful tool for analyzing the neighborhood-wise features of a quantum state is the ``Schmidt span''
of a linear object (vector, operator, or tensor) \cite{Johnson2016}:

\begin{defn}
Given the tensor product of two finite-dimensional inner-product spaces $W_1\otimes W_2$
and a vector $v\in W_1\otimes W_2$ with Schmidt decomposition $v=\sum_i s_i v_1^i\otimes v_2^i$, 
the {\em Schmidt span} of $v$ with respect to $W_1$ is $\Sigma_1(v)\equiv \textup{span}\{v_1^i\}$.
The corresponding \emph{extended Schmidt span} is defined as 
$\overline{\Sigma}_{1}(v)\equiv\Sigma_1(v)\otimes W_2$.
\end{defn}

\noindent  
We will mostly make use of the extended Schmidt span of the target state $\ket{\psi}$ with respect to neighborhood Hilbert spaces, namely, $\overline{\Sigma}_{\neigh_k}(\ket{\psi})=\Sigma_{\neigh_k}(\ket{\psi})\otimes\hilbert_{\overline{\neigh}_k}$.

\subsection{Convergence notions}

The task we focus on is the design of dynamics which drive $S$ towards a 
target state, subject to specified QL constraints. The following definitions provide the relevant 
stability notions in the Schr\"{o}dinger picture \cite{remarkHeisenberg}:

\begin{defn}
\label{def:gas}
A state $\rho\in \mathcal{D}(\hilbert)$ is 
\emph{globally asymptotically stable} (GAS) for the dynamics described by $\{ \mathcal{E}_t \}_{t\geq 0}$ 
if it is {\em invariant} and {\em attractive}, that is, if 
\begin{eqnarray}
&& \; \mathcal{E}_t (\rho) = \rho, \quad \forall t \geq 0, 
\label{inv} \\
&& \lim_{t\rightarrow \infty}|\mathcal{E}_{t,s}(\sigma)-\rho |=0,\quad \forall \sigma \in  \mathcal{D}(\hilbert), \forall s\geq 0. 
\label{attract}
\end{eqnarray}
\end{defn}

Following \cite{Ticozzi-alternating}, we define a notion of GAS with respect to the QL discrete-time dynamics 
given in Eqs. (\ref{eq:dyn})-(\ref{eq:prop}):

\begin{defn}
A target state $\rho$ is discrete-time \emph{quasi-locally stabilizable} (QLS) with respect to a neighborhood 
structure $\neigh$ if there exists a sequence $\{\mathcal{E}_t\}_{t\geq 0}$ of neighborhood maps rendering $\rho$ GAS.
\end{defn}

A main result in \cite{Ticozzi-alternating} (Theorem 8) establishes the following necessary and sufficient condition 
for determining whether a target {\em pure} state is QLS. Adapting the notation to the present context, we have:

\begin{thm} [\cite{Ticozzi-alternating}]
\label{thm:asymptotic}
A target pure state $\rho=\ketbra{\psi}$ is discrete-time QLS if and only if
\begin{equation}
\label{eq:qls}
\textup{span}(\ket{\psi})=\bigcap_{k}\overline{\Sigma}_{\neigh_k}(\ket{\psi}).
\end{equation}
\end{thm}

While the above characterizes {\em asymptotic} convergence, our aim in this work is to determine further 
conditions on the target state which enable \emph{finite-time} QL stabilization, in a sense 
made precise in the following:

\begin{defn}
A target state $\rho$ is quasi-locally \emph{finite-time stabilizable} (FTS) in $T$ steps 
with respect to a neighborhood structure $\neigh$ if there exists a finite sequence 
$\{\mathcal{E}_t\}_{t=1}^{T}$ of neighborhood maps satisfying
\begin{eqnarray}
&&\hspace*{-4mm} \mathcal{E}_t (\rho) = \rho, \quad  t =1, \ldots, T, 
\label{invFT} \\
&& \hspace*{-5mm} \mathcal{E}_{T,1}(\sigma)  
=(\mathcal{E}_{T}\circ \mathcal{E}_{T-1} \circ \ldots \circ \mathcal{E}_1) (\sigma)
= \rho, \; \forall \sigma \in  \mathcal{D}(\hilbert),  
\label{attractFT}
\end{eqnarray}
where $T\geq 0$ is the smallest integer for which attractivity holds. 
Furthermore, $\rho$ is \emph{robustly finite-time stabilizable} (RFTS) if $(\ref{attractFT})$ 
holds for any permutation of the $T$ maps.
\end{defn}

{\subsection{No-go for exact finite-time convergence with Lindblad dynamics}
\label{sub:nogo}

In continuous time, the counterpart to the discrete-time non-homogeneous Markovian dynamics 
defined in Eqs. (\ref{eq:dyn})-(\ref{eq:prop}) may be expressed as
\begin{equation}
\dot{\rho}_t = {\cal L}_t (\rho_t), \quad t \geq 0, 
\label{eq:ct}
\end{equation}
with formal solution given by the time-ordered propagator
${\cal E}_{t,s} \equiv {\mathcal T}\exp\{ \int_s^t ds\, {\cal L}_s\}, $
and where the Liouvillian generator ${\cal L}_t$ has the canonical Gorini-Kossakowskii-Sudarshan-Lindblad 
form \cite{gks,lindblad,alicki-lendi} ($\hbar=1$): 
\begin{equation}
{\cal L}_t = -i [H(t), \cdot ]  +
\sum_k \Big( L_k(t) \cdot L_k(t)^{\dagger}-\frac{1}{2}\{L_k^{\dagger}(t)L_k(t),\cdot\} \Big). 
\label{Lind}
\end{equation} 
Here, $H(t)$ and $\{L_k (t)\}$ represent an Hermitian (effective) Hamiltonian operator and arbitrary 
noise (Lindblad) operators, respectively, that are allowed to be time-dependent in general.  

Given a target state $\rho$, the property of GAS 
may be defined in analogy to Definition \ref{def:gas}, by noting that the invariance condition in (\ref{inv}) 
may be equivalently restated as ${\cal E}_{t,s} (\rho) = \rho$, for all $t > s\geq 0$, or also as a kernel condition, 
${\cal L}_t (\rho)= 0$, for all $t$.  Following \cite{Johnson2016}, quasi-locality constraints 
may be imposed by requiring that the Liouvillian $\mathcal{L}_t$ be 
expressible at any time in the form $\mathcal{L}_t \equiv \sum_k \mathcal{L}_{t, {\neigh_k}}
\otimes\mathcal{I}_{\overline{\neigh}_k}$. 
Previous work has extensively explored asymptotic QL stabilization in the case of 
{\em homogeneous} (time-invariant) continuous-time dynamics 
\cite{Kraus2008,Ticozzi2012,TicozziQIC2014,Johnson2016}, in which case each neighborhood generator 
${\cal L}_{\neigh_k}$ is time-independent and the propagator 
simplifies to a one-parameter semigroup of CPTP maps $\{\mathcal{E}_t=e^{{\cal L}t}\}_{t\geq 0}$.  
In particular, for a pure target state $\rho=|\psi\rangle\langle \psi |$, the necessary and sufficient conditions 
for asymptotic QL stabilization with discrete-time dynamics, Eq. (\ref{eq:qls}), are {\em formally} identical to those
characterizing asymptotic QL stabilization by purely dissipative Lindblad dynamics, namely, one where 
the task may be achieved by a generator with $H \equiv 0$. 

While for a time-independent Lindblad master equation the impossibility of exact FTS may be expected 
from the fact that the propagator $e^{t\cal L}$ converges exponentially to its steady state, a stronger no-go result
holds for arbitrary Markovian master equations, as in Eqs. (\ref{eq:ct})-(\ref{Lind}) -- and in fact, more generally, 
for non-Markovian {\em time-local} master equations \cite{kossakowski}.
This follows from a general result on {\em linear} time-varying dynamical systems:  

\begin{prop}
\label{prop:no-go}
Consider a dynamics driven by a (time-varying) linear equation on a linear space ${\cal X}$:
\[ \dot X_t = {\cal L}_t (X_t), \quad X_0=x_0. \]
Assume that ${\cal S}\subset {\cal X}$ is an invariant and attractive subspace for ${\cal L}_t$,
and that ${\cal L}_t$ is modulus-integrable, that is, $\int_0^t | {\cal L}_s|\,ds < \infty$, for all finite $t$.  
Then if $X_0$ does not belong to ${\cal S}$, $X_t$ will not be in ${\cal S}$ for all finite $t,$ 
namely, there cannot be exact convergence in finite time. 
\end{prop}

In the case at hand, the above Proposition may be applied with 
${\cal S} \equiv \{\lambda  \rho,\,\lambda\in {\mathbb C} \},$ 
the one-dimensional subspace associated to the target state $\rho$.  A crucial element entering the 
proof is the structure of dynamics on the orthogonal complement ${\cal S}^\perp$, that stems from the 
invariance requirement \cite{ticozzi-DID}.   Thus, no FTS of $\rho$ is possible 
with continuous time-local dynamics in general.

\subsection{Canonical parent Hamiltonian for asymptotically stabilizable pure states}

For pure target states obeying the conditions for asymptotic stability under either discrete-time or 
continuous-time QL Markovian dynamics (Theorem \ref{thm:asymptotic}),
physical insight can be gained by picturing the dissipative process as effectively cooling the system 
into the ground state of an appropriate Hamiltonian \cite{Kraus2008,Ticozzi2012,Ticozzi-alternating}.

Recall that a Hamiltonian is QL if it may be expressed as 
a sum of neighborhood-acting terms, 
$H\equiv\sum_k H_k=\sum_k H_{\neigh_k}\otimes\identity_{\overline{\neigh}_k},$ and it is 
\emph{frustration-free} (FF) if its ground space is contained in the ground state space of each 
such term $H_k$; that is, if $\ket{\psi}$ has minimal energy with respect to $H$, it has 
minimal energy with respect to each $H_k$.
In particular, a corollary in \cite{Ticozzi-alternating} shows that  $\ket{\psi}$ is 
discrete-time QLS with respect to $\neigh$ if and only if it is the {\em unique} ground state of some 
FF QL ``parent'' Hamiltonian $H$.  Accordingly, the QL stabilizing dynamics may be thought of as 
each neighborhood map ``locally cooling'' $S$ with respect to $H_k$: these 
local coolings collectively achieve global cooling to $\ket{\psi}$ by virtue of the FF property.

Among QL FF parent Hamiltonians that a given pure state may admit, 
one can be constructed in a {\em canonical way} from the state itself as follows:
\begin{defn}
\label{def:canham}
Given a neighborhood structure $\neigh =\{ {\cal N}_k \}$, the \emph{canonical FF parent Hamiltonian} 
associated to $|\psi\rangle$ is defined as
\begin{equation}
H_{\ket{\psi}} \equiv \sum_{k} (\identity - \Pi_{\neigh_k}\otimes\identity_{\overline{\neigh}_k}) 
\equiv \sum_k (\identity- \Pi_k),
\label{canonicalH}
\end{equation}
in terms of the 
projectors $\Pi_{\neigh_k}$ and $\Pi_k$ associated to 
the Schmidt span $\Sigma_{\neigh_k}(\ket{\psi})$ and the extended Schmidt span 
$\overline{\Sigma}_{\neigh_k}(\ket{\psi})$, respectively.
\end{defn}

\noindent 
This canonical Hamiltonian satisfies the following ``universal'' property: if there exists a QL FF Hamiltonian 
with $\ket{\psi}$ as its unique ground state, then $\ket{\psi}$ is the unique ground state of $H_{\ket{\psi}}$. 
Thus, $\ket{\psi}$ is QLS if and only if it is the unique ground state of its canonical FF parent Hamiltonian. 
A QL Hamiltonian such as $H_{\ket{\psi}}$ is referred to as {\em commuting} if the projectors $\Pi_k$ are mutually 
commuting.  While asymptotic stabilization is known to be possible independent of whether $H_{\ket{\psi}}$ is 
commuting or not \cite{Kraus2008,Ticozzi2012,Johnson2016}, for continuous-time dynamics, the existence of 
a commuting structure is also known to play a key role in influencing the speed of convergence to the 
steady state \cite{Reeb2012,Temme2014,Brandao2016} (cf. Sec. \ref{subsec:rapidL}).
It is thus natural to explore what implications commutativity of $H_{\ket{\psi}}$ may have in the context of FTS, 
and RFTS in particular. 


\section{Finite-time stabilization}
\label{sec:fts}

\subsection{Necessary conditions} 
\label{subsec:ftscond}

We begin the analysis of FT stabilization by providing a necessary condition for a 
pure target state to be FTS under specified QL constraints. 

\begin{thm}[\emph{\bf Small Schmidt-span condition}]
\label{thm:ftsnec}
A pure state $\ket{\psi}$ is FTS with respect to $\neigh$ only if 
it is QLS [Eq. (\ref{eq:qls})] 
and there exists at least one neighborhood $\neigh_k\in\neigh$ for which 
\begin{equation}
\label{eq:sss}
2\,\textup{dim}(\Sigma_{\neigh_k}(\ket{\psi})) \leq  \textup{dim}(\hilbert_{\neigh_k}).
\end{equation}
\end{thm}
\noindent
Intuitively, and as formalized in the proof, the necessity of a small Schmidt span 
may be understood from the fact that, in order for the sequence  
$\mathcal{E}_{T}\circ \mathcal{E}_{T-1} \circ \ldots \circ \mathcal{E}_1$ to stabilize $\ket{\psi}$,  
there must exist a neighborhood map $\mathcal{E}_{k}$ able to  
take a state $\sigma \ne \ket{\psi}\bra{\psi}$ into the target, while leaving the latter invariant. In
terms of quantum error correction, this action can be viewed as correcting a neighborhood-acting error on
$\ket{\psi}$. If $\sigma_{\neigh_k}(\ket{\psi})$ is too large, however, no neighborhood-acting errors can map
$\ket{\psi}$ to a non-trivial correctable state.
The existence of states which are stabilizable in infinite time but violate the 
small Schmidt span condition of Eq. (\ref{eq:sss}) is explicitly demonstrated in the following example. Thus, 
FTS states are a strict subset of QLS states, as one may intuitively expect. 

\begin{exmp}[\textbf{Spin-3/2 AKLT state}]
\label{ex:AKLTnotSSS}
The spin-3/2 Affleck-Kennedy-Lieb-Tasaki (AKLT) state $\ket{\text{AKLT}^N_{3/2}}$ \cite{AKLT} is typically defined in the thermodynamic limit on a system of spins arranged on a two-dimensional (2D) honeycomb lattice. More generally, given any degree-three graph with a spin-3/2 particle on each vertex, the corresponding AKLT state may be defined as the unique ground state of the two-body Hamiltonian $H= \sum_{\langle i, j\rangle} P^{(J=3)}_{ij}$, where $P^{(J=3)}_{ij}$ projects into the spin-3 subspace of particles $i$ and $j$, and the summation is carried out over each pair of adjacent vertices. With respect to the two-body neighborhood structure defined by $H$, the corresponding spin-3/2 AKLT state $\ket{\text{AKLT}^N_{3/2}}$ satisfies Eq. (\ref{eq:qls}) (which also follows from analysis in \cite{kraus}), and is QLS for every 
$N$. Consider the specific case of the $N=6$ spin-3/2 AKLT state defined with respect to the bipartite cubic graph (Fig. \ref{fig:AKLTnotSSS}). As verified numerically in MATLAB, this state violates the small Schmidt span condition and therefore is not FTS.  

\begin{figure}[t]
\includegraphics[width=0.6\columnwidth]{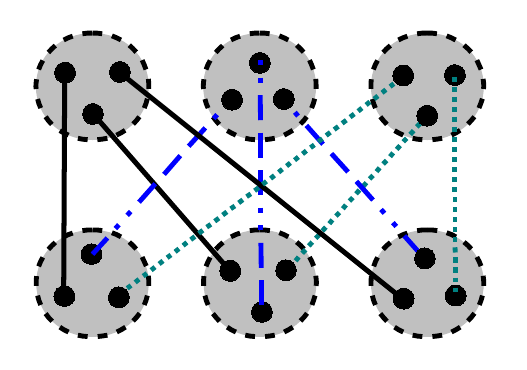}
\vspace*{-3mm}
\caption{(Color online) Example of a QLS but non-FTS state: the spin-$3/2$ AKLT state on a bipartite cubic graph. The pairs of nodes connected by an edge are virtual spin-1/2 particles in a singlet state. The dashed circles contain the systems which are projected into the spin-$3/2$ subspace. As verified numerically, this AKLT state violates the small Schmidt span property since for each top-bottom pair of systems (i.e., each neighborhood), the Schmidt span dimension ($=9$) exceeds half the Hilbert space dimension ($=16/2$).}
\label{fig:AKLTnotSSS}
\end{figure}
\end{exmp}

\subsection{Sufficient conditions}\label{sec:ftssuff}

Next, we construct FTS dynamics for any target state satisfying a particular condition. A crucial component of the scheme that we present is the use of neighborhood-acting \emph{unitary} maps, interspersed with dissipative maps. 

Let $\mathcal{U}(\hilbert)$ denote the unitary group of ($D \times D$) matrices on $\hilbert$, and $\mathfrak{u}(\hilbert)$ the corresponding Lie algebra. It is then useful to introduce the following target-dependent subgroups 
of $\mathcal{U}(\hilbert)$:

\begin{defn}
The \emph{unitary stabilizer group} of a vector $\ket{\psi}\in\hilbert$ is defined as $$\mathcal{U}_{\ket{\psi}}\equiv\{U\in\mathcal{U}(\hilbert)\,|\, U\ketbra{\psi}U^{\dagger}=\ketbra{\psi}\} \subset \mathcal{U}(\hilbert), $$ 
with the associated Lie algebra being denoted by $\mathfrak{u}_{\ket{\psi}}$.
The \emph{neighborhood unitary stabilizer group} of a vector $\ket{\psi}\in\hilbert$ with respect to $\neigh_k$ is defined as $$\mathcal{U}_{\neigh_k,\ket{\psi}}\equiv\{U\in\,\mathcal{U}(\hilbert_{\neigh_k})\otimes\identity_{\overline{\neigh}}\,|\, U\ketbra{\psi}U^{\dagger}=\ketbra{\psi}\},$$ 
with the associated Lie algebra being denoted by $\mathfrak{u}_{\neigh_k,\ket{\psi}}$.
\end{defn}

A crucial step in building our FTS scheme is the decomposition of elements of the global stabilizer group 
$\mathcal{U}_{\ket{\psi}}$ into a {\em finite} product of elements from the neighborhood stabilizer groups 
$\mathcal{U}_{\neigh_k,\ket{\psi}}$. The following proposition describes a condition for determining whether 
such a decomposition is possible:

\begin{prop}[\textbf{Unitary generation property}]
\label{thm:unitarygeneration}
Given a state $\ket{\psi}$ and a neighborhood structure $\neigh$, any element in $\mathcal{U}_{\ket{\psi}}$ may 
be decomposed into a finite product of elements in $\mathcal{U}_{\neigh_k,\ket{\psi}}$ if and only if
\begin{align}
\label{eq:ugenprop}
\langle \mathfrak{u}_{\neigh_k,\ket{\psi}} \rangle_{k}=\mathfrak{u}_{\ket{\psi}},
\end{align}
where $\langle \cdot \rangle_{k}$ denotes the smallest Lie algebra which contains all Lie algebras from the set indexed by $k$.
\end{prop}

\noindent 
Importantly, the linear-algebraic closure, $\langle \cdot \rangle_{k}$, may be computed numerically. Hence, for a given state, we may determine whether or not the unitary generation property holds using software such as MATLAB. 
We note that constructing an {\em explicit} decomposition may still be difficult in practice, and may 
be regarded as a constrained synthesis problem in geometric control, whose solution is beyond our scope here. 
The following example illustrates the essential features of the general scheme that we will use in verifying 
whether a state can be FTS. 

\begin{exmp}[\textbf{Dicke state}]
\label{ex:Dicke}
Consider a four-qubit system with a neighborhood structure 
$\neigh_1=\{1,2,3\}$ and $\neigh_2=\{2,3,4\}$. The two-excitation Dicke state, 
\begin{align*}
\ket{(0011)}\equiv\frac{1}{\sqrt{6}}(&\ket{0011}+\ket{0101}+\ket{0110}\nonumber\\
+&\ket{1001}+\ket{1010}+\ket{1100}), 
\end{align*}
is known to be QLS \cite{Ticozzi-alternating}. We now show that this state is also FTS with respect 
to the same $\neigh$. As above, we will use the notation $\ket{(X)},$ $X\in {\mathbb Z}_2^L,$ to 
denote the fully symmetric pure state $\frac{1}{\sqrt{L!}}\sum_\pi\ket{\pi(X)},$ 
where $\pi$ are the permutations of $L$ objects. 
The Schmidt span of $\ket{(0011)}$ with respect to $\neigh_1$ is 
$\Sigma_{\neigh_1}(\ket{(0011)})=\textup{span}\{\ket{(001)},\ket{(011)}\}$. Thus, the small Schmidt span condition is satisfied, as 
$ {\textup{dim}(\Sigma_{\neigh_k}(\ket{\psi}))}/ {\textup{dim}(\hilbert_{\neigh_k})} =2/8 \leq 1/2$.

Our strategy will be to use a neighborhood dissipative map, say, $\mathcal{W}$, 
which maps any density operator with support in a particular four-dimensional 
subspace into the target state, and to use neighborhood unitaries which ``rotate'' the 
range of $\mathcal{W}$ into the particular 
subspace; subsequently, a final application of $\mathcal{W}$ maps all states in this space to the pure target state.
Let $\omega\equiv e^{\frac{2\pi i}{3}}$ and let $\mathcal{W} \equiv \sum_i K_i\cdot K_i^{\dagger}$ be defined by its 
Kraus operators, acting non-trivially only on $\neigh_1$:
\begin{align*}
K_0&\equiv (\ketbra{(001)}+\ketbra{(011)})\otimes\identity, \nonumber\\
K_1&\equiv (\ket{(001)}\bra{000}+\ket{(011)}\bra{111})\otimes\identity, \nonumber\\
K_2&\equiv (\ket{(001)}\bra{(001)_{\omega}}+\ket{(011)}\bra{(011)_{\omega}})\otimes\identity, \nonumber\\
K_3&\equiv (\ket{(001)}\bra{(001)_{\omega^2}}+\ket{(011)}\bra{(011)_{\omega^2}})\otimes\identity,
\end{align*}
where $\ket{(abc)_{\nu}}\equiv \frac{1}{\sqrt{3}}(\ket{abc}+\nu\ket{bca}+\nu^2\ket{cab})$. By construction, $\mathcal{W}$ 
maps the following four orthogonal states (including the target state, itself) into 
$\ket{(0011)}$:
\begin{align*}
&\ket{\psi^0}\equiv\ket{(0011)},\\
&\ket{\psi^1}\equiv(\ket{000}\ket{1}+\ket{111}\ket{0})/\sqrt{2},\\
&\ket{\psi^2}\equiv(\ket{(001)_{\omega}}\ket{1}+\ket{(011)_{\omega}}\ket{0})/\sqrt{2},\\
&\ket{\psi^3}\equiv(\ket{(001)_{\omega^2}}\ket{1}+\ket{(011)_{\omega^2}}\ket{0})/\sqrt{2}.
\end{align*}
The range of $\mathcal{W}$ is the set of operators with support on the extended 
Schmidt span $\overline{\Sigma}_{\neigh_1}(\ket{(0011)})$. Thus, we next design a sequence of neighborhood 
unitaries $\{U_i\}$ which maps $\overline{\Sigma}_{\neigh_1}(\ket{(0011)})$ into $\textup{span}\{\ket{\psi^i}, 
\,i=0,\ldots, 3\}$:
\begin{align*}
& U=U_T\ldots U_1 = \ketbra{\psi^0}+\ket{\psi^1}\bra{(001)}\bra{0}+\nonumber\\
& \hspace*{-1mm}+ \ket{\psi^2}\bra{(011)}\bra{1} +\ket{\psi^3} [\bra{(001)}\bra{1}-\bra{(011)}\bra{0}] /\sqrt{2}+U_R,
\end{align*}
where $U_R$ is any matrix which ensures that $U$ is unitary. That $U$ can be decomposed into such a 
finite product is ensured by the fact that $\ket{(0011)}$ satisfies the Lie algebraic generation property of 
Eq. (\ref{eq:ugenprop}), as we checked using MATLAB. Finally, a simple calculation shows that 
\begin{equation*}
\mathcal{W}\circ\mathcal{U}_T\circ\ldots\circ\mathcal{U}_2\circ\mathcal{U}_1\circ\mathcal{W}(\identity/16)=\ketbra{(0011)}.
\end{equation*}
Hence, $\ket{(0011)}$ is FTS, as desired.

\vspace*{1mm}

{\em Remark:} 
While in the above example the dissipative map $\mathcal{W}$ is employed just twice, multiple uses may be required 
in the general case, with a different sequence of unitaries between each application.
Nonetheless, entropy is still removed from $S$ only by a dissipative action on a {\em single} neighborhood. This contrasts the 
QLS scheme of \cite{Ticozzi-alternating}, wherein dissipative maps alternatively act on {\em all} neighborhoods in order to asymptotically drive $S$ towards the target state. In a sense, infinite-time convergence is ensured by suitably tailoring the ``competition'' between dissipative maps, whereas a stronger form of ``cooperative'' action among CPTP maps, involving a non-trivial interplay between unitary and dissipative dynamics, is needed in our scheme for FT convergence.  
It is worth to anticipate that the Dicke state $\ket{(0011)}$ is provably {\em not} RFTS, as it violates a necessary condition 
we establish in Proposition \ref{thm:complementcommutecondition}.
This demonstrates that the RFTS property is strictly stronger than FTS, as expected.
\end{exmp}

We now state our general sufficient condition for FTS:
\begin{thm}
\label{thm:fts}
A state $\ket{\psi}$ is FTS relative to a connected neighborhood structure $\neigh$ if there exists at least one neighborhood $\neigh_k\in \neigh$ satisfying the small Schmidt span condition, 
$2\,\textup{dim}(\Sigma_{\neigh_k}(\ket{\psi})) \leq  \textup{dim}(\hilbert_{\neigh_k})$, 
and the unitary generation property holds, 
$\langle \mathfrak{u}_{\neigh_\ell,\ket{\psi}}\rangle_\ell=\mathfrak{u}_{\ket{\psi}}$.
\end{thm}

Notice that in the above theorem we request the neighborhood structure to be {\em connected}. To illustrate why this 
is important, consider a neighborhood structure comprised of two disjoint sets of neighborhoods (i.e., no neighborhood 
from the first set and from the second set have non-trivial intersection), giving, a ``left-right'' factorization 
$\hilbert\simeq\hilbert_L\otimes\hilbert_R$. 
The condition for asymptotic QLS, Eq. (\ref{eq:qls}),  
can only be satisfied if the target state is itself factorized, $\ket{\psi} \equiv \ket{\psi}_L\otimes\ket{\psi}_R$.
But then the neighborhood unitary stabilizers can, at most, generate $\mathcal{U}_{\ket{\psi}_{_L}}\otimes
\mathcal{U}_{\ket{\psi}_{_R}},$ which is {\em strictly smaller} than $\mathcal{U}_{\ket{\psi}}$. 
Disconnected neighborhood structures will never allow the unitary generation property to hold. Trivially, 
any product state is FTS with respect to a disconnected neighborhood structure. What is needed, then, is 
that the unitary generation property holds for each connected component of $\neigh$.
This motivates restricting to neighborhood structures which are connected, 
as a disconnected $\neigh$ precludes the possibility of stabilizing entangled target states. 

We now outline our general strategy for FTS. Assume that $\ket{\psi}$ and $\neigh$ obey the 
conditions of Theorem \ref{thm:fts} and, for ease of notation, let $\Sigma^0\equiv\Sigma_{\neigh_k}(\ket{\psi})$. 
Decompose $\hilbert_{\neigh_k}\simeq\bigoplus_{i=0}^{r-1}\Sigma^i\oplus\mathcal{R}$, where $\Sigma^i$ are 
orthogonal isomorphic copies of $\Sigma^0$ and $\mathcal{R}$ is the remainder space of minimal dimension. 
The small Schmidt span condition ensures that $r\geq 2$.  
For simplicity, our general proof is given (in Sec. \ref{sec:proofs}) for $r=2$ which, from a control standpoint, 
may be seen as a QL generalization of the splitting-subspace scheme for FTS introduced in 
\cite{baggio-splitting}.  However, the construction may be easily modified to improve the 
efficiency of the cooling action implemented by ${\mathcal W}$. 
If $S$ consists of $N$ qudits, with $D=d^N$, let 
$s_\ell \equiv  \text{dim}(\Sigma_{\neigh_\ell} (\ket{\psi})$, 
$r_\ell\equiv \lfloor d^{|\neigh_\ell |}/s_\ell \rfloor $, and 
$r\equiv \max_{\ell} r_\ell.$ 
Physically, we may think of $\log_d r_\ell$ as the ``cooling rate'' of $\neigh_\ell$-neighborhood maps 
with respect to $\ket{\psi}$, and of $\log_d r$ as the maximum cooling rate across $\neigh$. A larger cooling 
rate affords $\mathcal{W}$ to more greatly reduce the rank of the input density matrix.
Associating a tensor factor to the index $i$, and further identifying $\Sigma^i \simeq \Sigma$, for all $i$, 
the global Hilbert space $\hilbert= \hilbert_{\neigh_k}\otimes 
\hilbert_{\overline{\neigh}_k}$ decomposes as 
$$\bigg (\bigoplus_{i=0}^{r-1}\Sigma^i\otimes\hilbert_{\overline{\neigh}_k}\bigg)\oplus\mathcal{R}
\otimes\hilbert_{\overline{\neigh}_k}\simeq\complex^r\otimes(\Sigma \otimes\hilbert_{\overline{\neigh}_k})
\, \oplus \,\mathcal{R}\otimes\hilbert_{\overline{\neigh}_k}. $$
Then we can let $\mathcal{W}$ map {\em each} of the $r$ isomorphic copies $\Sigma^i$ onto $\Sigma^0$ as
\( \mathcal{W}\equiv (\ketbra{0}\trn{})\otimes\mathcal{I} \oplus \mathcal{I}. \)
The  unitary CPTP maps $\mathcal{U}_i$ are constructed so as to maximize the 
rank-reduction achieved by each $\mathcal{W}$. This is accomplished using the following algorithm:

\begin{enumerate}
\item Choose an orthonormal basis $\{\ket{\psi^0_{\alpha}}\}$, $\alpha=0, \ldots, \delta$, 
for $\Sigma^0\otimes\hilbert_{\overline{\neigh}_k}$, 
with $\ket{\psi^0_0}\equiv\ket{\psi}$, $\delta = s_k |\overline{\neigh}_k| -1$. 
This determines isomorphic orthonormal bases $\{\ket{\psi^i_{\alpha}}\}$ for 
the copies $\Sigma^i\otimes\hilbert_{\overline{\neigh}_k}$.

\item Choose an orthonormal basis $\{\ket{\lambda_{\beta}}\}$ for $\mathcal{R}\otimes\hilbert_{\overline{\neigh}_k}$, 
$\beta= 0, \ldots, \overline{\delta}$.

\item Order the basis vectors as
\begin{align}
\label{eq:basisorder}
&\ket{\psi^0_0},\ket{\psi^1_0},\ldots,\ket{\psi^{r-1}_0},\ket{\psi^0_1},\ket{\psi^1_1},\ldots,\ket{\psi^{r-1}_1},\nonumber\\
&\vdots\nonumber\\
&\ket{\psi^0_\delta },\ket{\psi^1_\delta},\ldots,\ket{\psi^{r-1}_\delta},\ket{\lambda_{0}},\ldots,\ket{\lambda_{\overline{\delta}}}.
\end{align}

\item The choice of each $\mathcal{U}_i$ depends recursively on the input density matrix $\rho_i=\mathcal{W}(\mathcal{U}_{i-1}(\rho_{i-1}))$, beginning with $\rho_1 = \mathcal{W}(\identity)$. 

\item In each step, $\mathcal{U}_{i}$ is a permutation of the basis vectors, chosen so that the target state is fixed and, iteratively, each basis vector in the support of $\rho_i$ is mapped to the first basis vector 
according to the ordering of Eq. (\ref{eq:basisorder}). Since $\ket{\psi}$ satisfies the unitary generation property with 
respect to $\neigh$, each global stabilizer $\mathcal{U}_i$ can be decomposed into a finite number of neighborhood stabilizers.\end{enumerate}

The sequence of CPTP maps terminates after a finite number of steps because the rank of the 
input (fully mixed) density matrix is necessarily reduced in each step. 
In contrast to the simpler implementation in the proof, this strategy allows $\mathcal{W}$ to simultaneously map 
{\em multiple states} to the target subspace. A concrete implementation of the algorithm is described in the example below.

\begin{figure*}[t]
\includegraphics[width=1.8\columnwidth]{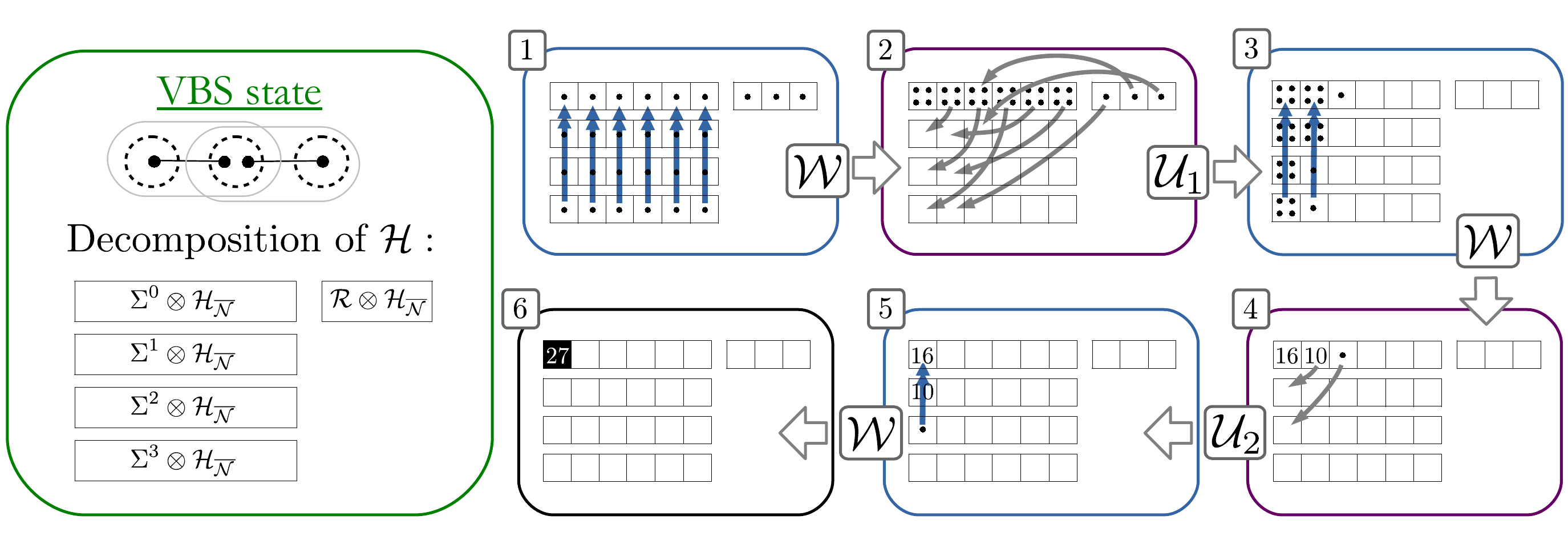}
\vspace*{-4mm}
\caption{(Color online) 
FTS scheme for the $N=3$ spin-$1$ VBS state on the line [Eq. (\ref{aklt1})], under 
NN constraints. In each numbered panel, each square represents one of the $D=27$ dimensions in the $\complex^3\otimes\complex^3\otimes\complex^3$ state space, while the dots represent the probabilistic weight of each basis vector for the current state. The task is to move all the probabilistic weight from the initial flat distribution (completely mixed state) into the box in the upper left-hand corner, corresponding to the target state. The Schmidt span on the first two systems is 
2-dimensional, leading to the 6-dimensional extended Schmidt span $\Sigma^0$, as represented by the first row of boxes. The remaining rows, labeled $\Sigma^i$, are isometric copies of this subspace, leaving the 3-dimensional remainder space 
$\mathcal{R}\otimes\hilbert_{\overline{\neigh}}$. The dissipative map $\mathcal{W}$ acts only on the first two qutrits, cooling each $\Sigma^i$ to $\Sigma^0$. The unitaries $\mathcal{U}_1$ and $\mathcal{U}_2$ leave the target state invariant while preparing probabilistic weight to be cooled by $\mathcal{W}$. Since $\ket{\textup{VBS}_1^N}$ satisfies the unitary generation property,
each $\mathcal{U}_i$ can be decomposed into a finite sequence of neighborhood-acting invariance-satisfying maps. }
\label{fig:FTSAKLT}
\end{figure*}

\begin{exmp}[\textbf{1D VBS states}]
\label{ex:vbs}
Consider an open chain of $N$ spin-1 particles, with a two-body NN neighborhood structure.
Let $\ket{0},\ket{1}$ be spin-$1/2$ basis states, with $\ket{\psi^\pm }\equiv \frac{1}{\sqrt{2}}(\ket{01}\pm \ket{10})$.
A 1D valence-bond-solid (VBS) state $\ket{\text{VBS}^N_1}$ \cite{AKLT1} 
can then be defined as
\begin{equation}
\ket{\textup{VBS}_1^N}\equiv P_1\bigg (\prod_{i=1}^{\lceil N/2 \rceil -1} P_{2i,2i+1}\bigg)P_{N}\ket{\psi^-}^{\otimes (N-1)},
\label{aklt1}
\end{equation}
where $P_1$, $P_N$ are isometries embedding a boundary spin-$1/2$ into
a spin-$1$ via $\ket{0}\mapsto \ket{m=1}$, $\ket{1}\mapsto \ket{m=-1}$, and 
each $P_{2i,2i+1}\equiv \ket{m=1}\bra{00} + \ket{m=0}\bra{\psi^+} +\ket{m=-1}\bra{11}$
projects corresponding spins from adjacent singlet
pairs (``bonds'') into the spin-1 triplet subspace.  $\ket{\textup{VBS}_1^N}$ may be verified to obey 
Eq. (\ref{eq:qls}), hence to be QLS, for arbitrary $N$.  In fact,  $\ket{\textup{VBS}_1^N}$ is 
the unique ground state of a (non-commuting) two-body FF Hamiltonian of the form  
$H= P^{(J=3/2)}_{1,2} + \sum_{i=3}^{N-2} P^{(J=2)}_{i,i+1} + P^{(J=3/2)}_{N-1,N}$, where $P^{(J)}_{i,i+1}$  
are projectors onto the total $J$-subspace \cite{Kirillov}. In the thermodynamic limit, 
the above $H$ reduces to the well-known AKLT model and, correspondingly, $\ket{\textup{VBS}_1^N}$ defines
the (translationally invariant) spin-$1$ AKLT state \cite{AKLT1,AKLT}. 

Direct calculation shows that with respect to the boundary neighborhoods $(1,2)$ and $(N-1,N)$, the Schmidt spans have 
dimension 2, whereas with respect to the remaining, bulk neighborhoods, the Schmidt spans have dimension 4. 
The neighborhood Hilbert spaces have dimension $3\cdot3=9$, so the small Schmidt span condition is satisfied.
It remains to show that the $\ket{\text{VBS}^N_1}$ states satisfy the unitary generation property. For small values of 
$N$ this may be checked numerically, as we have done explicitly for 
$N=3$ and $N=4$. We {\em conjecture} that for all $N$, the 1D VBS state is FTS with respect to the two-body 
NN neighborhood structure. The $N=3$ case is depicted and further described in Figure \ref{fig:FTSAKLT}.
\end{exmp}

Satisfaction of the sufficient conditions for FTS in Theorem \ref{thm:fts} certainly implies satisfaction of the necessary conditions 
for FTS of Theorem \ref{thm:ftsnec}. However, it is interesting to prove a direct connection between {\em asymptotic}, yet  
not necessarily FT, stabilizability and the unitary generation property.  We have the following:

\begin{prop}
\label{thm:qlsalg}
If $\ket{\psi}$ satisfies $\langle \mathfrak{u}_{\neigh_k,\ket{\psi}} \rangle_{k}=\mathfrak{u}_{\ket{\psi}}$ with respect to 
the neighborhood structure $\neigh$, then $\ket{\psi}$ satisfies Eq. (\ref{eq:qls}), and hence is QLS, with respect to $\neigh$.
\end{prop}
We {\em conjecture} that satisfaction of Eq. (\ref{eq:qls}) along with the small Schmidt span condition (i.e., the necessary 
conditions in Theorem \ref{thm:ftsnec}) is in fact \emph{sufficient} to ensure FTS. One avenue to proving this would be to 
establish the converse of Proposition \ref{thm:qlsalg}.


\section{Robust finite-time stabilization: necessary conditions}
\label{sec:necessary}

We begin our analysis of the robust stabilization setting by revisiting an obvious example: 

\vspace*{1mm}

\begin{exmp}[\textbf{Product states}] Given $\hilbert\simeq\bigotimes_{i=1}^N\hilbert_i$, consider a strictly 
local neighborhood structure, $\neigh = \{\neigh_i\} \equiv \{i\}$, and an arbitrary 
product state $\rho=\bigotimes_{i=1}^N\rho_i$. 
To each $i$, let us associate $\mathcal{E}_i 
\equiv(\rho_i\trn{i}{})\otimes\mathcal{I}_{\overline{i}}$. 
Then, any complete sequence of such maps gives
\begin{align*}
\mathcal{E}_N\circ\ldots\circ\mathcal{E}_2\circ\mathcal{E}_1 =\bigotimes_{i=1}^N(\rho_i\trn{i})
=\rho\trn{},
\label{eq:rftsprod}
\end{align*}
demonstrating that $\rho$ is RFTS, as expected. Since the maps commute, any ordering works. Of course, by considering 
a different $\neigh'$ with enlarged neighborhoods, relative to the strictly local one $\neigh$ above, any such factorized state 
remains RFTS. Hence, any (pure or mixed) product state is RFTS with respect to any neighborhood structure which covers all systems.
\end{exmp}

Although the above example is trivial, its structure is important: in much of our subsequent analysis, we shall seek ways to {\em represent the target state as a product state with respect to some virtual subsystems inside each neighborhood.} The next example demonstrates this idea by introducing a class of RFTS states which exhibit entanglement with respect to the physical subsystems, but can be seen as factorized with respect to virtual ones: 

\begin{exmp}[\textbf{Graph states}]
\label{exmp:graphstate}
\emph{Graph states} are a paradigmatic class of many-body entangled states which are known to be a resource for universal measurement-based quantum computation \cite{Raussendorf2003}. 
Following \cite{Zeng2015}, a graph state on $N$ qudits is defined by a graph $G=(V,E)$ with $N$ vertices and a choice of Hadamard matrix $H$. The latter must satisfy $H^{\dagger}H=d \identity$, $H=H^T$, and $|[H]_{ij}|=1$ for all $i,j$. The edge-wise action $C^H$ is defined, according to the choice of $H$, by $C^H\ket{ij} \equiv [H]_{ij}\ket{ij}$. 
The standard choice in the qubit case is that $C^{H}$ is a controlled-$Z$ transformation. Note that $C^H$ is diagonal in the computational basis and symmetric under swap of the two systems it acts on. We define the global graph unitary transformation as $U_G\equiv\prod_{(i,j)\in E} C^H_{i,j}$, with $\mathcal{U}_G(\cdot)\equiv U_G\cdot U^{\dagger}_G$. 
Then, the graph state associated to $G$ is
\begin{equation}
\ket{G}\equiv U_G\ket{+}^{\otimes N},\quad \ket{+}=H\ket{0}. 
\end{equation}
The above definition recovers the one derived from the standard (abelian) stabilizer formalism if $H$ 
coincides with the discrete Fourier transform.

A natural neighborhood structure may be associated to $G$ by defining, 
for each physical system $i$, a neighborhood $\neigh_i$ that includes system $i$ along with the graph-adjacent systems 
(i.e., the set of $j$ connected to $i$ by some edge $(i,j)\in E$). For any given $\ket{G}$, 
we may then construct a finite sequence of neighborhood maps which robustly stabilizes $\ket{G}$ relative to $\neigh$.
Let $\hat{\mathcal{E}}:\mathcal{B}(\complex^d)\rightarrow \mathcal{B}(\complex^d)$ be defined by $\hat{\mathcal{E}}
\equiv \ketbra{+} \text{Tr}$, and let $\hat{\mathcal{E}}_i$ indicate the map 
$\hat{\mathcal{E}}$ acting on system $i$ with trivial action on $\overline{i}$. To each $\neigh_i$, we then associate the 
map $\mathcal{E}_i\equiv \mathcal{U}_G\circ\hat{\mathcal{E}}_i\circ\mathcal{U}^{-1}_G$. The Kraus operators of $\hat{\mathcal{E}}_i$ are of the form $X^{\alpha}_i\otimes\identity_{\overline{i}}$. The unitary conjugation of $\hat{\mathcal{E}}_i$ transforms its Kraus operators into those of $\mathcal{E}_i$ as $\mathcal{U}_G(X^{\alpha}_i\otimes\identity_{\overline{i}})={X'}^{\alpha}_i$. Crucially, each ${X'}^{\alpha}_i$ acts non-trivially only on $\neigh_i$. This is seen as follows:
\begin{align*}
{X'}^{\alpha}_k &=U_G ( X^{\alpha}_k\otimes\identity_{\overline{k}})U^{\dagger}_G\nonumber\\
&=\Big(\prod_{j|(k,j)\in E} C^H_{k,j}\Big)(HX^{\alpha}_kH^{\dagger}\otimes\identity_{\overline{k}})\Big(\prod_{j|(k,j)\in E} C^H_{k,j} \Big)^{\dagger}\nonumber\\
&=({X'}^{\alpha}_k)_{\neigh_k}\otimes\identity_{\overline{\neigh}_k}.
\end{align*}
Hence, each $\mathcal{E}_i$ is a valid neighborhood map. 
Finally, we show that each $\mathcal{E}_i$ leaves $\ket{G}$ invariant and that the composition of any complete sequence of these maps prepares $\ket{G}$. Invariance is demonstrated by 
$\mathcal{E}_i(\ketbra{G})=\mathcal{U}_G [\hat{\mathcal{E}}_i(\mathcal{U}_G^{\dagger}(\ketbra{G}))] 
=\mathcal{U}_G(\ketbra{+}^{\otimes N})=\ketbra{G}$. Preparation is seen as follows:
\begin{align*}
\mathcal{E}_N\circ\ldots\circ\mathcal{E}_2\circ\mathcal{E}_1
&=\mathcal{U}_G\circ\hat{\mathcal{E}}_N\circ\ldots\circ\hat{\mathcal{E}}_2\circ\hat{\mathcal{E}}_1\circ\mathcal{U}^{-1}_G\nonumber\\
&=\mathcal{U}_G\circ(\ketbra{+}\trn{})^{\otimes N}\circ\mathcal{U}^{-1}_G\nonumber\\
&=\mathcal{U}_G(\ketbra{+}^{\otimes N})\trn{}=\ketbra{G}\trn{}.
\end{align*}
\end{exmp}

Graph states are a good starting point to introduce necessary conditions for RFTS.
A common feature of both product and graph states is that their canonical FF parent Hamiltonians 
are commuting. Although we will find later on that 
this commutativity is {\em not} necessary for RFTS, a weaker property \emph{is} 
necessary, nevertheless:

\begin{prop}[\textbf{Commuting projectors}]
\label{thm:complementcommutecondition}
If a target pure state $\ket{\psi}$ is RFTS with respect to neighborhood structure $\neigh$, then 
$[\Pi_k,\Pi_{\overline{k}}]=0$ for all neighborhoods $\neigh_k$, where $\Pi_k$ and $\Pi_{\overline{k}}$ are 
the orthogonal projectors onto 
$\overline{\Sigma}_{\neigh_k}(\ket{\psi})$ and $\cap_{{j\neq k}}\overline{\Sigma}_{\neigh_j}(\ket{\psi})$,
respectively.
\end{prop}
 
\noindent 
With this proposition, we can verify that neither the Dicke state on four qubits [Example \ref{ex:Dicke}], 
nor the VBS state on three qutrits [Example \ref{ex:vbs}], are RFTS on account of the lack of commutativity among 
the terms in their canonical FF Hamiltonian (note that in the tripartite setting, we may identify 
$\Pi_k \equiv \Pi_{12},$ $\Pi_{\overline{k}}\equiv \Pi_{\overline{12}}=\Pi_{23}$).  
Notwithstanding, Example \ref{exmp:graphstate} shows that there exist ``resourceful'' many-body 
entangled states which {\em are} RFTS: the key property that graph states obey is that their correlations 
are very strongly clustered -- in fact, they have finite support. 
The necessary conditions we now present show that {\em all} RFTS states must 
indeed possess ``well-behaved'' correlations, in a sense we make precise. 
For a given $\neigh$, let the \emph{neighborhood expansion} of a set of subsystems $A$ be defined as
$A^{\neigh} \equiv \bigcup_{\neigh_i\cap A\neq\emptyset} \neigh_i.$
Intuitively, $A^{\neigh}$ is the set of subsystems which are connected to $A$ by some neighborhood. We then have:

\begin{thm}
\label{thm:necessary_robust}
Let the pure state $\rho=\ketbra{\psi}$ be RFTS with respect to $\neigh$.  Then the following properties hold:

\vspace*{1mm}

\textbf{\em (i)  (Finite correlation)} For any two subsystems 
$A$ and $B$ having disjoint neighborhood expansions (i.e., 
$A^{\neigh}\cap B^{\neigh}=\emptyset$), 
arbitrary observables $X_A$ and $Y_B$  are uncorrelated, that is, 
$\tr{}{X_A Y_B \rho} =\tr{}{X_A \rho}\tr{}{Y_B \rho}$.

\vspace*{1mm}

\textbf{\em (ii) (Recoverability property)} If a map $\mathcal{M}$ acts non-trivially only on subsystem $A$, 
$\mathcal{M}\equiv\tilde{\mathcal{M}}_A\otimes\mathcal{I}_{\overline{A}}$, then there exists a sequence of CPTP neighborhood maps 
$\mathcal{E}_{j}$, each acting only on $A^{\neigh}$, such that $\rho=\mathcal{E}_l\circ\ldots\circ\mathcal{E}_1\circ\mathcal{M}(\rho)$.

\vspace*{1mm}

\textbf{\em (iii) (Zero CMI)} For any two subsets of subsystems $A$ and $B$, with $A^{\neigh}\cap B = \emptyset$, the quantum conditional mutual information (CMI), $I(A:B|C)_{\rho}\equiv S(A,C)+S(B,C)-S(A,B,C)-S(C)$, satisfies $I(A:B|C)_{\rho}=0$, where $C\equiv A^{\neigh}\backslash A$.
\end{thm}

Returning to Example \ref{ex:vbs}, since no finite length is known to exist beyond which correlations vanish in 
the AKLT spin-$1$ state \cite{Bruno}, this also precludes the possibility for the VBS states of Eq. (\ref{aklt1})
to be RFTS for large $N$. 

\vspace*{1mm}

{\em Remark:} As already noted, in \cite{Brandao2016disc} a scheme is developed to efficiently prepare 
({\em without} ensuring invariance) both Gibbs 
and ground states of certain FF QL Hamiltonians, using a sequence of QL CPTP maps.
Interestingly, their sufficient conditions for preparation are related to the above necessary 
conditions of short-ranged correlations and zero CMI.  Specifically, the scheme in \cite{Brandao2016disc}
succeeds when the target state exhibits exponentially decaying correlations  
and has a sufficiently small CMI with respect to certain regions.


\section{Robust finite-time stabilization: sufficient conditions}
\label{sec:sufficient} 

In this section, we present three distinct sets of sufficient conditions for ensuring RFTS of a pure target state. 
The first set of conditions is satisfied by {\em all} the RFTS states that we know of, and provides a general 
framework for RFTS. However, it is ``non-constructive'' in that it is not easy to operationally verify if a given state 
satisfies the required properties. The other conditions are computable, at the cost of being less general.
In particular, the second set of ``algebraic'' conditions, while being applicable to arbitrary neighborhood 
geometries and able to incorporate a number of important examples (including graph states),  
fails to detect some RFTS states we could identify.  Our third set of sufficient conditions is further specialized 
to a class of neighborhood structures whose overlaps obey suitable ``matching'' properties.

\subsection{Sufficiency criteria from virtual subsystems: basic examples}  
\label{sub:nonop}

To understand what features ensure RFTS of a general pure state, we take a closer look at the graph states 
of Example \ref{exmp:graphstate}. 
Their central property is that they factorize with respect to a decomposition of the Hilbert space into 
\emph{virtual subsystems} \cite{Knill2000, Zanardi2001}, each ``contained'' in a single neighborhood. 
This is key for allowing each map $\mathcal{E}_i$ to independently cool the corresponding virtual degree of 
freedom into the state $\ket{+}$, despite a non-trivially ``overlapping'' action of these maps at the physical level.
More formally, the fact that all the observables for a given virtual subsystem are also neighborhood operators for a 
corresponding physical neighborhood enables each map to stabilize the desired virtual-subsystem state 
while respecting the QL constraint.} As a byproduct, these maps can be chosen to {\em commute} with each other.

Graph states are associated to a virtual-subsystem description that satisfies an additional property: namely, there is a 
one-to-one correspondence between physical and virtual subsystems. This allows for a {\em unitary} mapping $U_{G}$ 
between the physical and virtual degrees of freedom, which is particularly simple to write. As we will find, such a 
strong correspondence is {\em not} necessary for RFTS. 
The ``minimal''  features that allow graph states to be RFTS can be generalized as follows. 
Let $W$ be an isomorphism between the  physical subsystem Hilbert space and a virtual subsystem Hilbert space,
\begin{equation}
\label{eq:basisfactor}
W:\bigotimes_i\hilbert_i\rightarrow\bigotimes_j\hat{\hilbert}_j, 
\end{equation}
where, in general, we need {\em not} require any pair $\hilbert_i$ and $\hat{\hilbert}_j$ to be isomorphic (e.g., 
the physical systems could be qubits, while the virtual subsystems are four-dimensional). For a more compact 
notation, we shall henceforth denote decompositions linked by an identification as in Eq. (\ref{eq:basisfactor}) simply by
$\bigotimes_i\hilbert_i\simeq\bigotimes_j\hat{\hilbert}_j$.

This ``relabeling'' of the degrees of freedom allows us to state two conditions which ensure $\ket{\psi}$ to be RFTS:

(1) The target $\ket{\psi}$ should be factorized with respect to the virtual degrees of freedom, that is, 
$\ket{\psi}\simeq\bigotimes_{j}\ket{\hat{\psi}_j};$

(2) The operators associated to any virtual subsystem should, themselves, be neighborhood operators; 
that is, for every $j$, a neighborhood $\neigh_k$ should exist such that for any virtual-subsystem operator 
$\hat{X}_j\in\mathcal{B}(\hat{\hilbert}_j)$,
$\hat{X}_j\otimes\identity_{\overline{j}} \in \mathcal{B}(\hilbert_{\neigh_k})\otimes\identity_{\overline{\neigh}_k}.$

Assume that the two conditions above hold for some $\ket{\psi}$. We can then construct a finite sequence of commuting QL 
CPTP maps which robustly stabilize $\ket{\psi}$. Define the maps $\mathcal{E}_j,$ strictly local on the virtual subsystems, as:
\begin{align*}
\mathcal{E}_j \equiv (\ketbra{\hat{\psi}_j}\trn{})_j\otimes\mathcal{I}_{\overline{j}}.
\end{align*} 
The Kraus operators of $\mathcal{E}_j$ are contained in $\mathcal{B}(\hat{\hilbert}_j)\otimes\identity_{\overline{j}}$. 
Hence, by the second property, each $\mathcal{E}_j$ is a valid neighborhood map. Each $\mathcal{E}_j$ leaves the target state invariant:
\begin{align*}
&\mathcal{E}_l(\ketbra{\psi})=\hat{\mathcal{E}}_l \Big(\bigotimes_{j=1}^M\ketbra{\hat{\psi}_j}\Big)\nonumber\\
&=\ketbra{\hat{\psi}_l}\tr{}{\ketbra{\hat{\psi}_l}}\otimes\bigotimes_{j\neq l}\ketbra{\hat{\psi}_j}=\ketbra{\psi}. 
\end{align*}
Finally, {\em any} complete sequence of these neighborhood maps prepares $\ket{\psi}$, as desired: 
\begin{align*}
\mathcal{E}_T\circ\ldots\circ\mathcal{E}_2\circ\mathcal{E}_1
&=\bigotimes_{j=1}^T\ketbra{\hat{\psi}_j}\trn{j}\nonumber\\
&\hspace*{-6mm}=\bigg(\bigotimes_{j=1}^T\ketbra{\hat{\psi}_j}\bigg)\bigg(\bigotimes_{j=1}^T\trn{j}\bigg)
=\ketbra{\psi}\trn{}.
\end{align*}

While for $d=2$ graph states are an example of stabilizer states \cite{Nielsen-Chuang:10},
we next demonstrate another class of RFTS qubit states which, while constructed in close analogy to 
graph states, are {\em not} standard stabilizer states.

\begin{exmp}[\textbf{CCZ states}]
\label{exmp:CCZ}
In \cite{Miller2015} the authors introduce 
a class of states that exhibit genuine 2D symmetry-protected topological order,
which we will refer to as \emph{controlled-controlled-Z (CCZ) states}.
While such states may be defined for any 3-uniform hypergraph (i.e., one with only 3-element edges), 
we restrict here to the triangular lattice, which allows for scaling to an arbitrary number of lattice sites, $N$. 
As with graph states, each qubit in the lattice is initialized in $\ket{+}=(\ket{0}+\ket{1})/\sqrt{2}$. Then, on 
each triangular cell a CCZ gate is applied.
Noting that all CCZ gates commute with one another, we let $U_{\Delta}\equiv\prod_{(i,j,k)\in {\cal T}} \text{CCZ}_{ijk}$, 
where ${\cal T}$ denotes the set of triangular cells on the lattice. The target CCZ state is 
\begin{equation}
\label{eq:ccz}
 \ket{\Delta}\equiv U_{\Delta}\ket{+}^{\otimes N}.
\end{equation}
To each site we associate a neighborhood defined by that qubit along with the six adjacent qubits. We verify that 
$\ket{\Delta}$ is RFTS with respect to this $\neigh$ by identifying a virtual-subsystem decomposition satisfying 
the needed 
properties. As with graph states, we can identify each physical subsystem to a virtual subsystem, 
with the unitary transformation $U_{\Delta}$ taking the physical-subsystem observables into the virtual ones. 
Then, each virtual-subsystem algebra corresponds to a neighborhood-contained algebra thanks to the 
commutativity of the CCZ gates:
\begin{align*}
\mathcal{B}(\hat{\hilbert}_i)\otimes\identity_{\overline{i}}
&=U_{\Delta}(\mathcal{B}(\hilbert_i)\otimes\identity_{\overline{i}})U^{-1}_{\Delta}\\
&=\left[U_{\neigh_i}(\mathcal{B}(\hat{\hilbert}_i)\otimes\identity_{\neigh_i\backslash i}) U^{-1}_{\neigh_i}\right]\otimes\identity_{\overline{\neigh_i}}\\
&\leq\mathcal{B}(\hilbert_{\neigh_i})\otimes\identity_{\overline{\neigh_i}},
\end{align*}
where $U_{\neigh_i}\equiv \prod_{k,l\in\neigh_i\backslash i} \text{CCZ}_{ikl}$ acts non-trivially
only on the physical systems in 
$\neigh_i$. Furthermore, by construction, $\ket{\Delta}$ is a virtual product state: considering  
Eq. (\ref{eq:ccz}), $U_{\Delta}$ maps each physical factor into a corresponding virtual-subsystem factor, 
giving $\ket{\Delta}\simeq \ket{\hat{+}}^{\otimes N}$ with respect to $\bigotimes_{i=1}^N \hat{\hilbert}_i$. 
As the neighborhood containment property of the virtual subsystems and the factorization of $\ket{\Delta}$ 
are satisfied, the CCZ state is verified to be RFTS. 
\end{exmp}

\subsection{Non-constructive general sufficient conditions}
\label{sub:nonconstructive}

The need for introducing a more general type of virtual-subsystem decomposition is illustrated 
by the following target state, which does {\em not} admit a simple neighborhood 
factorization as considered above, yet {\em is} RFTS: 

\begin{figure}[t]
\includegraphics[width=0.7\columnwidth]{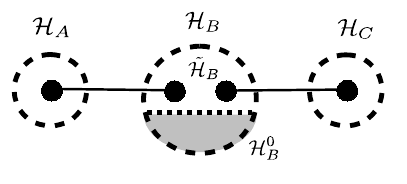}
\vspace*{-5mm}
\caption{Example of the virtual-subsystem decomposition used in constructing 
a state that cannot be obtained as a simple virtual product state yet is RFTS.}
\label{fig:nonfacstate}
\end{figure}

\begin{exmp}[\textbf{Non-factorizable RFTS state}] 
\label{ex:nonfacstate}
Consider $\hilbert\simeq\hilbert_A\otimes\hilbert_B\otimes\hilbert_C\simeq \complex^2\otimes\complex^5\otimes\complex^2$, with 
$\neigh_1=\{A,B\}$, $\neigh_2=\{B,C\}$.
Let the target state be 
$$\ket{\psi}=\ket{000}+\ket{011}+\ket{120}+\ket{131}.$$ 
One may verify that $\ket{\psi}$ would not satisfy the conditions (1)-(2) proposed in the previous subsection.  
Nonetheless, we can decompose system $B$ as 
$$\hilbert_B \simeq(\complex^2\otimes\complex^2)\oplus\complex^1\simeq (\hilbert_{b}\otimes\hilbert_{b'}) \oplus\hilbert^0_B\simeq\tilde{\hilbert}_B\oplus\hilbert^0_B,$$ 
by which we may label its basis vectors as, say, 
\begin{align*}
\ket{0}&=\ket{++},\nonumber\\
\ket{1}&=\ket{+-},\nonumber\\
\ket{2}&=\ket{-+},\nonumber\\
\ket{3}&=\ket{--},\nonumber\\
\ket{4}&=\ket{e}, 
\end{align*}
where $\text{span}\{\ket{e}\}\equiv \hilbert^0_B$.  
With respect to the resulting decomposition, we can write (see Fig. \ref{fig:nonfacstate} for a schematic)
$$\ket{\psi}=[(\ket{0+}+\ket{1-})\otimes(\ket{+0}+\ket{-1})]\oplus 0.$$  
Note that $\ket{\psi}$ is orthogonal to the space $\hilbert_A\otimes\hilbert^0_B\otimes\hilbert_C$. 
We now construct maps which render $\ket{\psi}$ RFTS. Define $\mathcal{E}^0:\mathcal{B}(\hilbert_B)\rightarrow \mathcal{B}(\hilbert_B)$ to be 
$$\mathcal{E}^0(\sigma)\equiv(\identity-\ketbra{e})\sigma(\identity-\ketbra{e})+\frac{1}{4}(\identity-\ketbra{e})\bra{e}\sigma\ket{e}.$$ 
This CPTP map takes probabilistic weight from $\hilbert^0_B$
and maps it uniformly to the complement. Also, define 
$\hat{\mathcal{E}}_1:\mathcal{B}([\hilbert_A\otimes\hilbert_{b}\otimes\hilbert_{b'}]\oplus[\hilbert_{A}\otimes\hilbert^0_B])\rightarrow \mathcal{B}([\hilbert_A\otimes\hilbert_{b}\otimes\hilbert_{b'}]\oplus[\hilbert_{A}\otimes\hilbert^0_B])$ to be
$$ \hat{\mathcal{E}}_1\equiv[\ketbra{\phi^+}\trn{A,b}\otimes\mathcal{I}_{b'}]\oplus \mathcal{I}, $$ 
where $\ket{\phi^+} \equiv \frac{1}{\sqrt{2}}(\ket{0+}+\ket{1-})$. We define $\hat{\mathcal{E}}_2$ acting on $\neigh_2$ similarly. 
With these, let the two neighborhood maps 
\begin{align*}
\mathcal{E}_1&\equiv (\hat{\mathcal{E}}_1\otimes\mathcal{I}_C)\circ(\mathcal{I}_{A}\otimes\mathcal{E}^0\otimes\mathcal{I}_C),\\
\mathcal{E}_2&\equiv (\mathcal{I}_A\otimes\hat{\mathcal{E}}_2)\circ(\mathcal{I}_{A}\otimes\mathcal{E}^0\otimes\mathcal{I}_C).
\end{align*}
Crucially, since the outputs of both $\mathcal{E}_1$ and $\mathcal{E}_2$ cannot have support in $\hilbert_A\otimes\hilbert^0_B\otimes\hilbert_C$, the action of $\mathcal{E}^0$ following either map is trivial, 
$ (\mathcal{I}_{A}\otimes\mathcal{E}^0\otimes\mathcal{I}_C)\circ\mathcal{E}_1=\mathcal{E}_1, $
and similarly for $\mathcal{E}_2$. Hence, the product of either order of the maps is
\begin{align*}
 \mathcal{E}_1\circ\mathcal{E}_2&=(\hat{\mathcal{E}}_1\otimes\mathcal{I}_C)\circ(\mathcal{I}_{A}\otimes\mathcal{E}^0\otimes\mathcal{I}_C)\circ\mathcal{E}_2\nonumber\\
&=(\hat{\mathcal{E}}_1\otimes\mathcal{I}_C)\circ\mathcal{E}_2\nonumber\\
&=(\hat{\mathcal{E}}_1\otimes\mathcal{I}_C)\circ(\mathcal{I}_A\otimes\hat{\mathcal{E}}_2)\circ(\mathcal{I}_{A}\otimes\mathcal{E}^0\otimes\mathcal{I}_C)\nonumber\\
&=\ketbra{\psi}\trn{}.
\end{align*}
\end{exmp}

The key feature of the state in the above example that enables it to be RFTS is that there exists 
a neighborhood factorization 
where the algebra of each factor is contained in a corresponding neighborhood, 
{\em once a subspace of $\hilbert$ is removed ``locally''} ($\hilbert^0_B\leq\hilbert_B$). 
To cover these more general cases, two additional steps may be required before the actual identification of the 
virtual degrees of freedom is made:  {\em subsystem coarse-graining} and {\em local restriction}. 

Coarse-graining may be required as the decomposition in the physical subsystems may be more fine-grained than 
needed, relative to the specified QL constraint.
For example, consider systems $A,B,C,D$, with neighborhoods $\{A,B,C\}$ and $\{B,C,D\}$. While the physical locality 
describes four subsystems, the separation between $B$ and $C$ is ``artificial'', as far as the neighborhood structure 
is concerned. In such a scenario, it is convenient to start from a coarse-grained subsystem decomposition 
where we group subsystems $B$ and $C$, namely, $\hilbert_A\otimes\hilbert_{BC}\otimes\hilbert_D$. 
This idea may be generalized by considering the equivalence classes of the subsystems with respect to the relation ``is contained in the same set of neighborhoods as''. Explicitly, let us define the equivalence relation $\sim_{\text{cg}}$ on the subsystem's indexes 
as $i \sim_{\text{cg}} j$ whenever $i\in\neigh_k$ for some $k$ implies $j\in\neigh_k$, and vice-versa. We then have: 

\begin{defn}
Given $\hilbert\simeq\bigotimes_{i=1}^N\hilbert_i$ and a neighborhood structure $\neigh$, 
the \emph{coarse-grained subsystems} are associated to 
\(\hilbert_\ell \equiv \bigotimes_{i\in {\cal C}_\ell}\hilbert_i, \)
with ${\cal C}_\ell$ denoting equivalence classes under the relationship $\sim_{\text{cg}}$.
\end{defn}

Though usually explicitly stated, in the remainder of the paper the decomposition of the physical Hilbert space 
$\hilbert$ will be taken to refer to the {\em coarse-grained subsystems}, with 
$\neigh$ being understood accordingly. After coarse graining, in order to find a suitable factorization in virtual 
subsystems we may still need to restrict to a {\em subspace} of the coarse-grained particles:

\begin{defn}
Given $\hilbert\simeq \bigotimes_{i}\hilbert_i$ and a set of subspaces $\tilde{\hilbert}_i\leq\hilbert_i$, 
the \emph{locally restricted} Hilbert space is given by  $\tilde{\hilbert}\simeq\bigotimes_{i}\tilde{\hilbert}_i$. 
\end{defn}
\noindent
We are now ready to state the most general sufficient conditions for RFTS we can provide:

\begin{thm}[\textbf{Neighborhood factorization on local restriction}]
\label{thm:lrhilbertdecomp}
A state $\ket{\psi}$ of the coarse-grained subsystems associated to $\hilbert\simeq\bigotimes_{i=1}^N\hilbert_i$ is 
RFTS with respect to the neighborhood structure $\neigh$ if:
\begin{enumerate}
\item There exists a locally restricted space $\tilde{\hilbert}=\bigotimes_{i=1}^N\tilde{\hilbert}_i$ that admits a 
virtual-subsystem decomposition of the form $\tilde{\hilbert}=\bigotimes_{j=1}^M \hat{\hilbert}_j$, such that
\begin{equation}
\label{eq:statefacred}
\ket{\psi}=\bigotimes_{j=1}^M\ket{\hat{\psi}_j}\oplus 0 \in \Big( \bigotimes_{j=1}^M \hat{\hilbert}_j\Big) \oplus \hilbert^0 ,
\end{equation}
where $\hilbert^0\simeq\tilde{\hilbert}^{\perp}$; 
\item For each virtual subsystem $\hat{\hilbert}_j$, there exists a neighborhood $\neigh_k$ such that
\begin{equation}
\mathcal{B}(\hat{\hilbert}_j)\otimes\identity_{\overline{j}}\oplus \identity^0 \leq \mathcal{B}(\hilbert_{\neigh_k})\otimes\identity_{\overline{\neigh_k}}.
\end{equation}
\end{enumerate}
\end{thm}

\noindent 
The following two examples, inspired by the work of Bravyi and Vyalyi in \cite{Bravyi2005}, detail a construction 
of RFTS states whereby such a neighborhood factorization arises.

\begin{exmp}[\textbf{Bravyi-Vyalyi states}]
\label{ex:BV}
The focus of \cite{Bravyi2005} is the complexity of the ``common eigenspace'' problem,  which aims to determine 
whether there exists a common eigenstate $\ket{\psi}$ of some given {\em commuting} Hamiltonians $\{H_i\}$.
The important setting the authors consider is the 2-local case, whereby each $H_i$ is a two-body operator. 
We now revisit their approach and identify a class of states which, on top of being the 
unique ground state of a FF QL Hamiltonian and hence QLS, are {\em also} RFTS. 

Consider a graph $G=(V,E)$, where each vertex corresponds to a physical subsystem of 
$\hilbert =  \bigotimes_{j=1}^N\hilbert_j$, and a set of commuting two-body projectors $\{\Pi_{jk} \}$ is in one-to-one 
correspondence with the edges of $G$, that is, $(j,k)\in E$. Following Lemma 8 in \cite{Bravyi2005} 
(with slightly adapted notation), these commuting two-body projectors, which play the role of projectors onto the 
relevant eigenspace of the $H_j$, are shown to induce a decomposition of each physical subsystem space of the 
form 
\[\hilbert_j= \bigoplus_{\alpha_j}\Hi^{(\alpha_j)}_j=\bigoplus_{\alpha_j}
\bigotimes_{k} \Hi^{(\alpha_j  \alpha_k)}_{jk},\quad j=1,\ldots, N, \]
such that each projector $\{\Pi_{jk}\}$ can be represented as:
\[\Pi_{jk}=\bigg(\bigoplus_{\alpha_j}\bigoplus_{\alpha_k}\Pi_{jk}^{(\alpha_j\alpha_k)}\bigg)\otimes \identity_{\overline{jk}}.\]
\noindent 
Here, each $\Pi_{jk}^{(\alpha_j\alpha_k)}$ is an orthogonal projector acting non-trivially {\em only} on 
$\Hi^{(\alpha_j \alpha_k)}_{jk}\otimes\Hi^{(\alpha_k \alpha_j)}_{kj},$ 
and $\identity_{\overline{jk}}$ is the identity on all physical particles but $j,k$.
Intuitively, $\Hi^{(\alpha_j \alpha_k)}_{jk},$ with $j\neq k,$ is associated to virtual ``subparticles'' of particle $j$,  
that couple via $\Pi_{jk}$ to those of particle $k$, corresponding to $\Hi^{(\alpha_k \alpha_j)}_{kj}$; 
$\Hi^{(\alpha_j \alpha_j)}_{jj}$ represent local degrees of 
freedom that are left invariant by all projectors. 
If $\alpha \equiv (\alpha_1,\ldots,\alpha_N)$,
the total Hilbert space then reads \cite{Bravyi2005}
\begin{align*}
\hilbert   \simeq \bigoplus_\alpha \bigotimes_j \hilbert_j^{(\alpha_j)} \equiv \bigoplus_\alpha \hilbert^{(\alpha)} 
\simeq \bigoplus_\alpha \Big [ \bigotimes_{j \leq k}  \Hi^{(\alpha_j  \alpha_k)}_{jk} \Big].
\end{align*}
The above decomposition implies that the common $1$-eigenspace of the $\Pi_{jk}$ is spanned by 
states that, within a {\em fixed} sector $\alpha$, are simply virtual product states:
\begin{equation}
\label{eq:BVfac}
\ket{\phi} \equiv \bigotimes_{j\leq k}\ket{\phi^{(\alpha_j\alpha_k)}_{jk}} \in \hilbert^{(\alpha)}.
\end{equation}
For $j\neq k$, $\ket{\phi^{(\alpha_j\alpha_k)}_{jk}}$ belongs to the range of $\Pi_{jk}^{(\alpha_j\alpha_k)}$ 
on $\Hi^{(\alpha_j \alpha_k)}_{jk}\otimes\Hi^{(\alpha_k \alpha_j)}_{kj},$ while $\ket{\phi^{(\alpha_j\alpha_k)}_{jj}}$ 
is any state in $\Hi^{(\alpha_{j})}_{jj}.$

The states in Eq. (\ref{eq:BVfac}) can be mapped back to the physical state space $\hilbert$ by the 
isometric embeddings
\begin{eqnarray*}
V_{j}: \hilbert_j^{(\alpha_j)} \simeq 
\bigotimes_{k} \hilbert^{(\alpha_j \alpha_k)}_{jk}\rightarrow\hilbert_j,
\end{eqnarray*}
resulting in states we term \emph{Bravyi-Vyalyi (BV) states}: 
\begin{equation}
\ket{\phi_{BV}} \equiv (V_1\otimes\ldots\otimes V_N)\ket{\phi} 
\in\hilbert , 
\label{BVstates}
\end{equation}
which may also be naturally recast as tensor network states \cite{Orus2014}.
Any BV state is, by construction, RFTS with respect to the neighborhood structure determined 
by its ``interaction graph'': if, for fixed $\alpha$, we compare 
Eq. (\ref{eq:BVfac}) and Eq. (\ref{eq:statefacred}), we note that the virtual subsystems are associated to spaces 
$\hilbert_{jj}^{(\alpha_j \alpha_j)}$ and  $\hilbert_{jk}^{(\alpha_j \alpha_k)}\otimes\hilbert_{kj}^{(\alpha_k \alpha_j)}$. 
Thus, each of them is contained in $\neigh_{jk}=\{j,k\}$, as desired.  Explicitly, we have:
\begin{align*}
\mathcal{B}(\hilbert_{jj}^{(\alpha_j \alpha_j)})\otimes\identity^{(\alpha)}_{\overline{jj}}\oplus 0&
\leq\mathcal{B}(\hilbert_{\neigh_{jk}})\otimes\identity_{\overline{\neigh}_{jk}}\nonumber\\
\mathcal{B}(\hilbert_{jk}^{(\alpha_j \alpha_k)}\otimes\hilbert_{kj}^{(\alpha_k \alpha_j)})
\otimes\identity^{(\alpha)}_{\overline{jk,kj}}\oplus 0
&\leq\mathcal{B}(\hilbert_{\neigh_{jk}})\otimes\identity_{\overline{\neigh}_{jk}},
\end{align*}
where the zero operator acts on the subspace generated by all 
$\bigoplus_{\beta\neq\alpha}\bigotimes_{jk}\hilbert^{(\beta_j \beta_k) }_{jk},$ and each identity 
operator acts on all 
virtual particles associated to the string $\alpha$, except those of $jj$, and $jk,kj$, respectively. Notice that, 
unless there is a unique state in the common 1-eigenspace, the projectors used in constructing a BV state 
are {\em not} the canonical projections (Eq. (\ref{canonicalH})) associated to it.
\end{exmp}

The following example presents a generalization of BV states to QL notions beyond the original two-body
setting. More than a way to test states for RFTS, the importance of both these example lies in the fact that 
they provide non-trivially multipartite entangled states that, by design, are guaranteed to be RFTS: 

\begin{exmp}[\textbf{Generalized Bravyi-Vyalyi states}] 
\label{ex:gbv}
{Consider a neighborhood structure $\neigh$ and a virtual subsystem decomposition of each 
(coarse-grained, if necessary) physical particle in $f_i$ virtual particles, that is, 
$$\hilbert_i \equiv \hilbert^0_i\oplus \tilde{\hilbert}_i\simeq \hilbert^0_i\oplus \bigotimes_{j=1}^{f_i} \hilbert_{ij}.$$ 
Define the local restriction $\tilde{\hilbert}\equiv\bigotimes_i\tilde{\hilbert}_i$ and, 
{for each neighborhood $\neigh_k$, consider a subset $S_k$ of pairs $ij$ such that: 
(1) $\hat{\hilbert}_k\equiv\bigotimes_{ij\in S_k}\hilbert_{ij}$ is contained in $\neigh_k$; 
(2) the sets $S_k$ are disjoint and each $ij$ is contained in some $S_k$.} 
These conditions together ensure that
\begin{equation*}
\tilde{\hilbert}\simeq\bigotimes_{i=1}^N\tilde{\hilbert}_i\simeq\bigotimes_{k=1}^{|\neigh|}\hat{\hilbert}_k,
\end{equation*}
where we emphasize that the local tensor factors $\tilde\Hi_i$ are very different with respect to 
the $\hat\Hi_k$, that directly reflect the QL constraint.
{Let $V_i:\bigotimes_{j=1}^{f_i} \hilbert_{ij}\rightarrow\hilbert_i$ be isometric embeddings from the virtual to 
the physical particles, with $V\equiv V_1\otimes\ldots\otimes V_N$, and consider any virtual product state 
$\ket{\psi}=\bigotimes_k\ket{\hat{\psi}_k}\in\tilde{\hilbert}$.} In analogy to Eq. (\ref{BVstates}), 
we define a \emph{generalized BV state} to be any of the form
\begin{equation}
\label{eq:gbv}
\ket{\psi_{GBV}}\equiv V\ket{\psi}={(V_1\otimes\ldots\otimes V_N)}\bigg(\bigotimes_k\ket{\hat{\psi}_k}\bigg).
\end{equation}
The different tensor product structures and associated factorizations the construction relies upon  
are precisely what allows for $\ket{\psi_{GBV}}$ to exhibit entanglement among the physical particles. 
A concrete example of a generalized BV state  is depicted and discussed in Fig. \ref{fig:GBVU}. 
}
\end{exmp}

\begin{figure}[t]
\includegraphics[width=0.75\columnwidth]{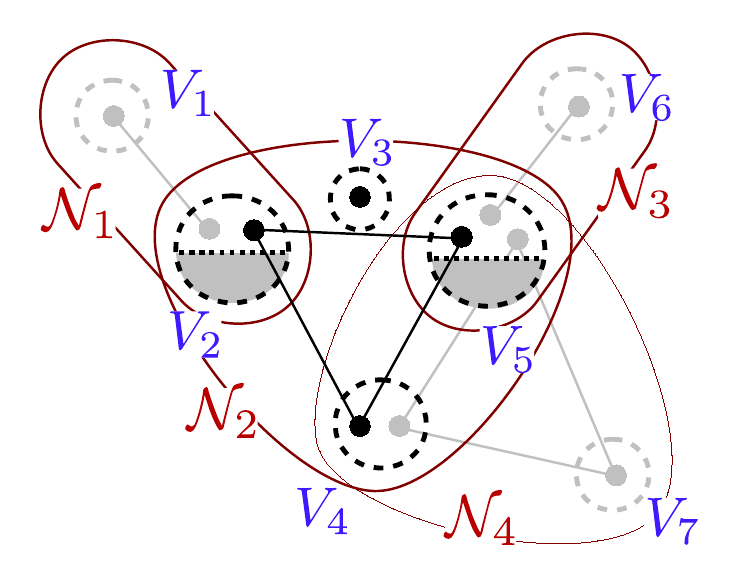}
\vspace*{-1mm}
\caption{(Color online) Example structure of a generalized BV state. Dashed circles denote physical particles, 
nodes correspond to virtual subsystems, and solid lines connect virtual subsystems that are entangled. 
The hatched semicircles indicate the subspaces $\hilbert_2^0$ and $\hilbert_5^0$ ($\hilbert_i^0\oplus\tilde{\hilbert}_i=\hilbert_i$), where the respective reduced states do not have support. Solid curves delineate the four neighborhoods. An entangled state $\ket{\hat{\psi}_k}$ is associated to each of the four groups of virtual particles that are contained in the same $\neigh_k$ of all those connected to them by solid lines. The resulting generalized BV state $\ket{\psi_{GBV}}= 
(V_1\otimes\ldots\otimes V_7) \ket{\psi}=\ket{\hat{\psi}_1}\otimes\ket{\hat{\psi}_2}\otimes\ket{\hat{\psi}_3}\otimes\ket{\hat{\psi}_4}$. 
Entanglement among the virtual subsystems is mapped into multipartite entanglement among the physical particles. 
Notwithstanding, the state is RFTS. In the RFTS scheme, the job of each neighborhood map $\mathcal{E}_k$ is to transfer probabilistic weight into $\tilde{\hilbert}_i$ and then prepare the corresponding virtual factor $\ket{\hat{\psi}_k}\in\hat{\hilbert}_k$, while acting trivially on the rest.
}
\label{fig:GBVU}
\end{figure}

\subsection{Constructive sufficiency criteria}

\subsubsection{Algebraic factorization}

In the BV and generalized BV schemes we just described, the factorization into virtual particles is induced by a
set of commuting projectors, acting on sets of particles, for which the target is the only common eigenstate.
However, given a target state, we might not know if it admits such a description. In this section, we draw
inspiration from the BV approach to investigate ways to {\em construct} a similar factorization of the Hilbert
space, amenable to RFTS, by using the projectors associated to the canonical Hamiltonian 
of a QLS state [Eq. (\ref{canonicalH})].
This construction will also include important examples that the BV schemes cannot accommodate. 
As we know from Example \ref{exmp:graphstate}, the qubit graph state on a 2D square lattice is RFTS, 
admits a virtual subsystem factorization (with each virtual particle including non-trivially degrees of freedom of 
five physical particles), and is the unique 1-eigenstate of a set of 5-body commuting projectors. 
Yet, it cannot be seen as a BV state, since coupling between more than two physical particles are involved, 
nor does it admit a generalized BV decomposition, since  
physical qubit subsystems cannot be decomposed to begin with.

We gain insight into the more general type of factorization we seek by revisiting again 
graph states, from an algebraic point of view.
Each virtual subsystem $\hat \hilbert_i$ can be associated to the operator subalgebra 
${\cal A}_i=\mathcal{B}(\hat{\hilbert}_i)\otimes\identity_{\overline{i}}$. These subalgebras 
satisfy the following properties: (1) each ${\cal A}_i$ acts non trivially only on 
$\neigh_i$; (2) ${\cal A}_i$ commutes with all Schmidt-span projectors $\Pi_k$ for $k\neq i$; (3) the ${\cal A}_i$ 
commute with each other; (4)  the union of these algebras generates the full operator algebra on the multipartite system. 
Actually, each ${\cal A}_i$ {\em can be defined} by (1) and (2), as the algebra of operators acting on $\hilbert_{\neigh_i}$ 
which commute with the remaining neighborhood projectors ($\Pi_k$, for $k\neq i$). In the following, 
we will build on this fact to {\em find} virtual particles that factorize our target. 

Recall that a set of C$^\ast$-subalgebras $\{\mathcal{A}_j\}$ of ${\cal B}(\Hi)$ is {\em commuting} if for each $X_j\in\mathcal{A}_j,X_k\in\mathcal{A}_k,$ we have $[X_j,X_k]=0.$ It is {\em complete} on $\hilbert$ if their union generates ${\cal B}(\Hi).$ A complete set of commuting subalgebras induces a factorization in virtual particles:
\begin{prop}[\textbf{Algebraically induced factorization}]
\label{thm:algfac}
If a set of algebras $\{\mathcal{A}_j\}$, $\mathcal{A}_i\leq\mathcal{B}(\hilbert)$, is complete and commuting, then each $\mathcal{A}_j$ has a trivial center and there exists a decomposition of the Hilbert space $\hilbert\simeq\bigotimes_{j=1}^T\hat{\hilbert}_j$ for which $\mathcal{A}_j\simeq\mathcal{B}(\hat{\hilbert}_j)\otimes\identity_{\overline{j}}$ for each $j$.
\end{prop}

For the sake of generality, we want to allow for factorizations on locally-restricted spaces, as in Proposition \ref{thm:lrhilbertdecomp}. Towards this, we provide a means of constructing a locally restricted space from a positive-semidefinite operator, 
such as the target state.
\begin{defn}
\label{def:subsup}
Given an operator $M\geq 0$ acting on (coarse-grained) subsystems $\bigotimes_{i=1}^N\hilbert_i$ with neighborhood structure $\neigh$, we define the \emph{subsystem support} of $M$ on $p$ as $\textup{supp}(\tr{\overline{p}}{M})$ and the \emph{subsystem kernel} of $M$ on $p$ as $\ker(\tr{\overline{p}}{M})$. The \emph{local support} of $M$ is then $\tilde{\hilbert}\equiv\bigotimes_{i=1}^N \textup{supp}(\tr{\overline{i}}{M})$, with $\hilbert=\hilbert^0\oplus\tilde{\hilbert}$.
\end{defn}

\noindent Note that, for a pure state $\ket{\psi}$, the subsystem support of $\rho=\ketbra{\psi}$ on $p$ is simply the Schmidt 
span $\Sigma_p(\ket{\psi})$. By construction, the support of each neighborhood projector $\Pi_j$ is a subspace of the local support of the target state, $\tilde{\hilbert}=\bigotimes_{i=1}^N \textup{supp}(\tr{\overline{i}}{\ketbra{\psi}})$. This allows us to define projectors $\tilde{\Pi}_k\equiv\Pi_k|_{\tilde{\hilbert}}$ restricted to the local support of $\ketbra{\psi}$. 
We denote the local support of $\ketbra{\psi}$ on a given neighborhood as $\tilde{\hilbert}_{\neigh_j}\equiv \bigotimes_{i\in\neigh_j}\textup{supp}(\tr{\overline{i}}{\ketbra{\psi}}),$ and its complement $\tilde{\hilbert}_{\overline{\neigh}_j}\equiv \bigotimes_{i\notin\neigh_j}\textup{supp}(\tr{\overline{i}}{\ketbra{\psi}})$. With these, we consider the following candidate for the algebras that are to induce a factorization of the target state: 

\begin{defn}
\label{defn:neighalg}
Given a target state $\ket{\psi}$ and a neighborhood structure $\neigh$, for each neighborhood $\neigh_j$, let
\begin{equation*}
\tilde{A}_{\neigh_j}\equiv \{X\in\mathcal{B}(\tilde{\hilbert}_{\neigh_j})| [X_{\tilde{\hilbert}_{\neigh_j}}
\otimes\identity_{\tilde{\hilbert}_{\overline{\neigh}_j}},\tilde{\Pi}_k]=0,\;\forall \,k\neq j \}.
\end{equation*}
The \emph{neighborhood algebra} is then defined as
\begin{equation}
\label{eq:neighalg}
\mathcal{A}_j\equiv (\textup{span}(\identity_{\neigh_j,0})\oplus \tilde{\mathcal{A}}_{\neigh_j})
\otimes\identity_{\overline{\neigh}_j},
\end{equation}
relative to the decomposition $\hilbert\simeq (\hilbert_{\neigh_j,0}\oplus \tilde{\hilbert}_{\neigh_j})
\otimes\hilbert_{\overline{\neigh}_j}$, where $\hilbert_{\neigh_j,0}$ is the complement of 
$\tilde{\hilbert}_{\neigh_j}$ in $\hilbert_{\neigh_j}$.
\end{defn}

\noindent 
Each neighborhood algebra $\mathcal{A}_j$ is an associative algebra, as it can be written as a commutant:
\begin{align*}
\mathcal{A}_j&=\{\Pi_k,\forall\;k\neq j;\identity_{\neigh_j}\otimes\mathcal{B}(\hilbert_{\overline{\neigh}_j})\}'.
\end{align*}
As for graph states, each $\mathcal{A}_j$ is thus the largest C$^\ast$-algebra of $\neigh_j$-neighborhood operators 
which commute with all the remaining neighborhood projectors $\Pi_k$. 

We now give the main result of this section, which states how the structure of the neighborhood algebras can 
ensure a particular factorization of the target state and, hence, that the latter be RFTS:

\begin{thm}[\textbf{Algebraic factorization RFTS}]
\label{thm:factorization} 
Let $\ket{\psi}$ on (coarse-grained) subsystems $\bigotimes_{i=1}^N\hilbert_i$ be QLS with respect to $\neigh$ and 
let the neighborhood algebras $\mathcal{A}_j$ be commuting and complete on the local support space $\tilde{\hilbert}$.
Then $\ket{\psi}$ admits a decomposition \[\ket{\psi}=0\oplus\bigotimes_j \ket{\hat{\psi}_j},\] 
with respect to the neighborhood algebra-induced factorization 
$\hilbert\simeq\hilbert^0\oplus (\bigotimes_j\hat{\hilbert}_j)$, and is thus RFTS.
\end{thm}

The key feature of this sufficient condition is that it is {\em operationally checkable}:
satisfaction of Eq. (\ref{eq:qls}) is determined by an intersection of vector spaces, and the neighborhood 
algebras and their commutativity can be computationally determined.
This sufficient condition, however, still does {\em not} incorporate all examples of RFTS that we know of. 
In some cases, the reduced states of the target state on a particular neighborhood may contain physical factors which 
are full rank: $\tr{\overline{\neigh}_k}{\ketbra{\psi}}=\rho_{\neigh_k}=\rho_{\neigh_k\backslash i}\otimes \rho_i$, with $\rho_i>0$. 
But then, invariance requires that any neighborhood map $\mathcal{E}_k$ act trivially on system $i$. Thus, if $\ket{\psi}$ 
were RFTS with respect to $\neigh=\{ \neigh_k\}$, it would be RFTS with respect to 
$\neigh' \equiv \{ \neigh_k\backslash i \}$. 
We have found cases in which the sufficient conditions of Theorem \ref{thm:factorization}, while not initially satisfied, 
become satisfied after updating the neighborhood structure as above.

\subsubsection{Matching overlap}
\label{sub:matchingoverlap}

The above algebraic condition may be simplified if the QL constraints satisfy a property that makes them similar to the 
edges of a graph, in the following sense:

\begin{defn}
A neighborhood structure  $\neigh$ satisfies the \emph{matching overlap condition} if for any set of neighborhoods 
that have a common intersection, this common intersection is also the intersection of any pair of the neighborhoods in the set.
\end{defn}

\noindent 
While two-body neighborhoods necessarily satisfy the matching overlap condition, general neighborhood structures, 
as for graph states or those in Fig. \ref{fig:GBVU}, need not (see also Fig. \ref{fig:MO}). The matching overlap condition 
basically ensures that the intersection of any two non-disjoint neighborhoods is a coarse-grained particle. This fact is used in 
establishing the following result:

\begin{figure}[t]
\includegraphics[width=0.45\columnwidth]{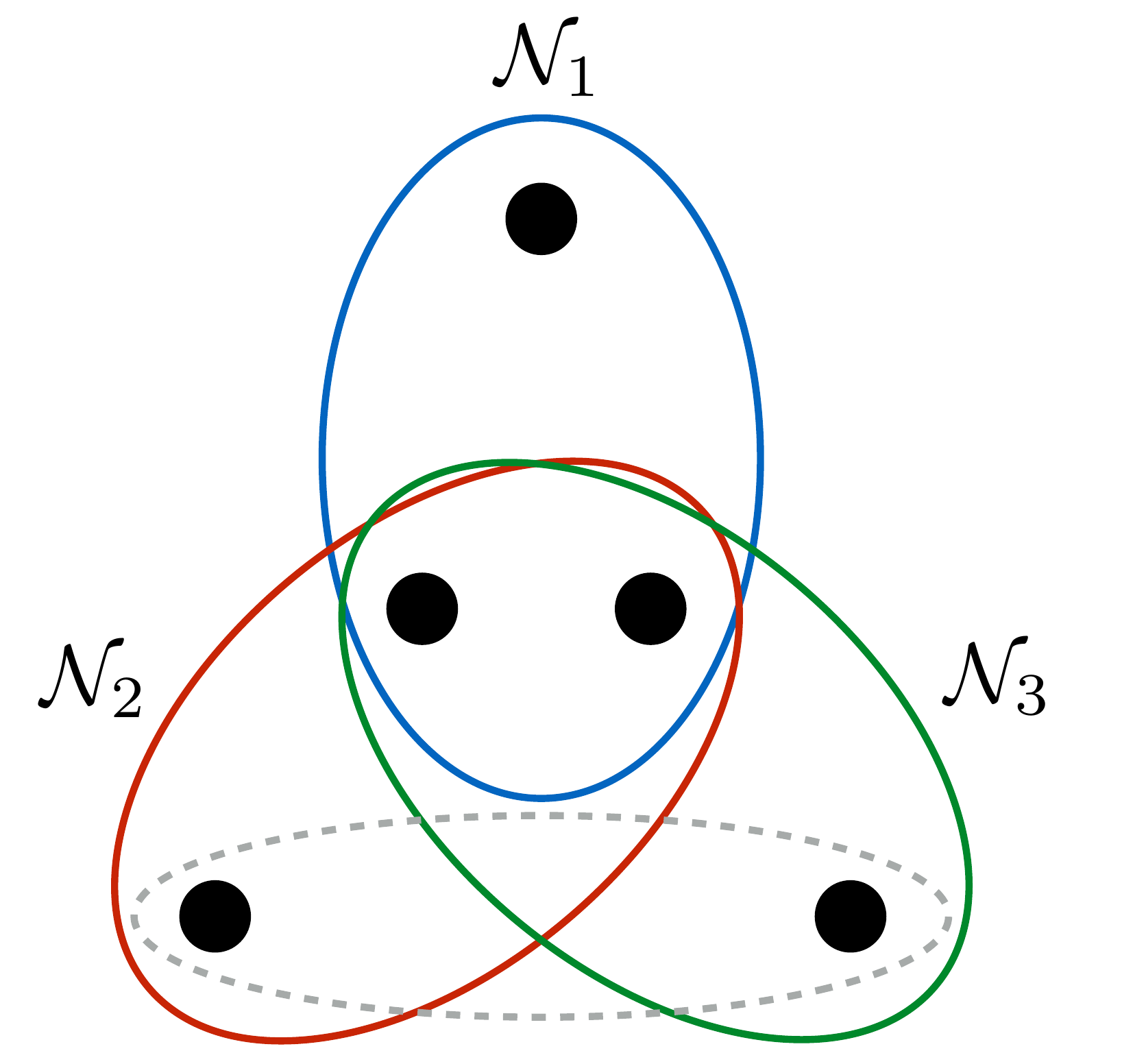}
\includegraphics[width=0.45\columnwidth]{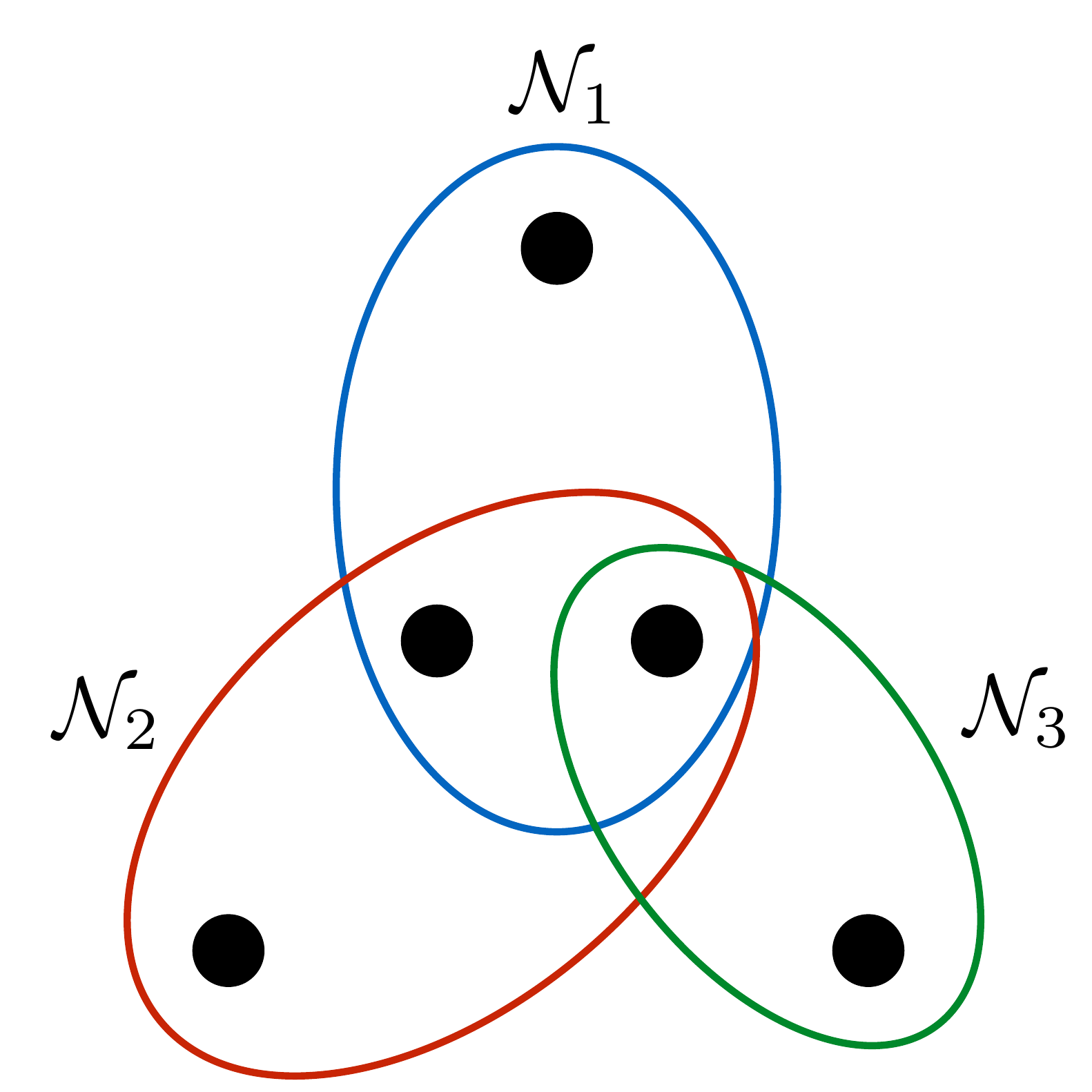}
\vspace*{-1mm}
\caption{(Color online) Illustrative example of a neighborhood structure on $N=5$ systems that does 
obey the matching-overlap property (left) versus one that does not (right). By including the dashed (grey) neighborhood, 
$\neigh$ would admit a non-trivial cycle, and lose its tree-like structure.}
\label{fig:MO}
\end{figure}

\begin{thm}[\textbf{Matching overlap RFTS}]
\label{thm:matchingoverlap}
Assume that $\ket{\psi}$ on (coarse-grained) subsystems $\bigotimes_{i=1}^N\hilbert_i$  is QLS with respect to 
$\neigh$, which satisfies the matching overlap condition. If $[\Pi_j,\Pi_k]=0$ for all pairs of neighborhood projectors, 
then $\ket{\psi}$ is RFTS.
\end{thm}

Notice how this allows us to completely by-pass the need for identifying a virtual-particle factorization to 
ascertain whether a state is RFTS. 
From a physical standpoint, the above theorem brings the commutativity properties of the canonical parent Hamiltonian 
$H= \sum_k\Pi_k$ to the fore: it is tempting to ask whether commuting neighborhood projectors may also be necessary 
for a QLS state to, further, be RFTS. The following example shows, however, that this is certainly \emph{not} true if the 
matching overlap condition is relaxed:

\begin{exmp}[\textbf{RFTS ground state of non-commuting canonical parent Hamiltonian}]
Consider nine qubits, labeled 1-9, described by the targetstate
$$ \ket{\psi}_W \equiv \ket{W}_{123}\otimes\ket{W}_{456}\otimes\ket{W}_{789}, $$ 
where $\ket{W} = \frac{1}{\sqrt{3}}(\ket{001}+\ket{010}+\ket{100})$, with the relevant $\neigh$
being depicted in Fig. \ref{fig:wstate}. 
That $\ket{\psi}_W$ is RFTS follows from the fact that it can be factorized such that each factor is contained in 
a neighborhood. The three maps which compose to stabilize $\ket{\psi}_W$ are $\mathcal{E}_{123}
 \equiv (\ketbra{W}_{123}\trn{})\otimes\mathcal{I}_{\overline{123}}$, and similary for $\mathcal{E}_{456}$ and 
 $\mathcal{E}_{789}$. To show that the neighborhood projectors $\Pi_k$ do {\em not} commute, consider 
 $\Pi_A$ and $\Pi_B$. On systems 7, 8, and 9, these, respectively, project onto $\textup{supp}(\identity_{7}\otimes 
 \tr{7}{\ketbra{W}_{789}})$ and $\textup{supp}(\tr{9}{\ketbra{W}_{789}}\otimes\identity_{9})$. A direct calculation shows 
that these two projections 
do not commute with one another. Hence, $[\Pi_A,\Pi_B]\neq 0$, and, 
by symmetry, this holds for any pair of $\Pi_k$.
Despite the fact that $H= \sum_k\Pi_k$ is thus non-commuting, we 
can still construct a different, commuting FF QL Hamiltonian for which $\ket{\psi}_W$ is the unique ground state, 
namely,
\begin{align*}
\tilde{H}_{} &= (\identity - \ketbra{W}_{123}\otimes\identity_{\overline{123}})+ 
(\identity - \ketbra{W}_{456}\otimes\identity_{\overline{456}})\nonumber\\
&+ \: (\identity - \ketbra{W}_{789}\otimes\identity_{\overline{789}}).
\end{align*}
\end{exmp}

\begin{figure}[t]
\includegraphics[width=0.7\columnwidth,viewport= 0 0 800 760,clip]{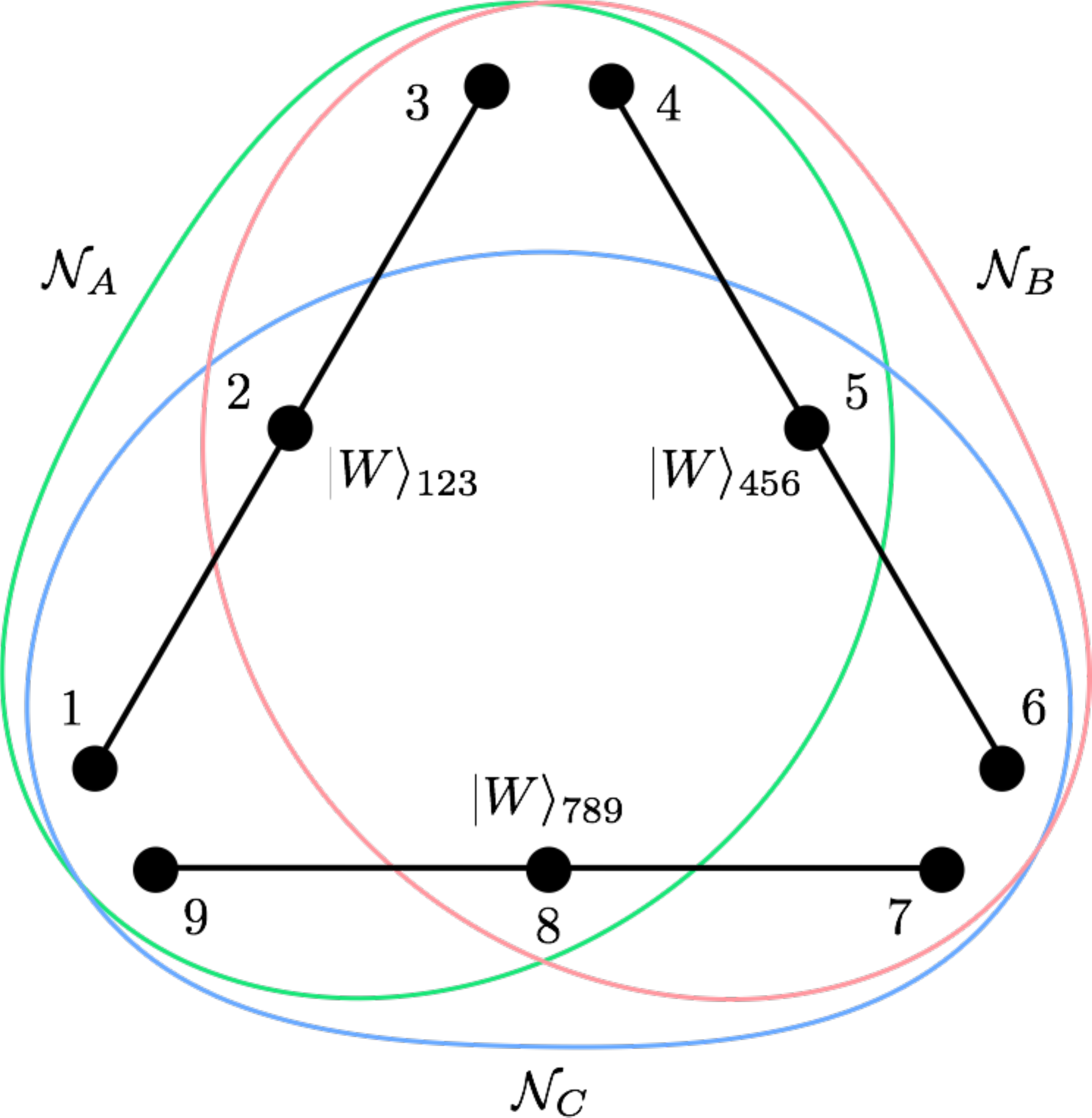}
\vspace*{-1mm}
\caption{(Color online) 
Example of a RFTS state with a non-commuting canonical parent Hamiltonian.
The three neighborhoods are enclosed by ovals and are most easily described by their respective complements. 
Letting $S\equiv \{1,\ldots,9\}$, we define $\neigh_A\equiv S\backslash \{6,7\}$, $\neigh_B\equiv S\backslash \{1,9\}$, 
and $\neigh_C\equiv S\backslash \{3,4\}$.  Note that the matching overlap condition is not obeyed. } 
\label{fig:wstate}
\end{figure}

As this example shows, the canonical Hamiltonian $H$ is {\em not}, in itself, 
useful for diagnosing whether a QLS state can be RFTS.  
In this regard, a few remarks are in order.
First, although we shall not include a formal proof here, one may show that, by 
further restricting the neighborhood structures to {\em both} obey the matching overlap property and
avoid the occurrence of loops (Fig. \ref{fig:MO}, left), commutativity of $H$ is in fact {\em necessary and sufficient} 
for a QLS pure state to be RFTS \cite{CDC2017}. These ``tree-like'' geometries 
include arbitrary 1D NN settings, though not the 2D lattice NN neighborhood structure. 
Equivalently, one can see that for any such tree-like QL constraint, $|\psi\rangle$ is RFTS if and only if 
it is a generalized BV state as in Eq. \eqref{eq:gbv}, further contributing to exact characterizations of 
ground states of commuting FF Hamiltonians \cite{Beigi2012}.
We further {\em conjecture} that, if $\ket{\psi}$ is RFTS, there always exists some FF QL
commuting parent Hamiltonian for which it is the unique ground state.  
Finding such non-canonical parent Hamiltonians remains, however, an interesting 
open problem in general.

\subsection{Extension to mixed target states}
\label{sub:mixed}

Although the focus of this paper is on target pure states, and extending the analysis of FT stabilizabity 
to general mixed states is well beyond our aim, we collect here 
those results that carry over directly to target mixed states.
In our analysis of FTS in Sec. \ref{sec:fts}, 
the purity of the target state played a crucial role. Even the necessary condition of small Schmidt span [Theorem \ref{thm:ftsnec}] 
involved criteria that only apply to target pure states. Thus, analysis of non-robust FTS for target mixed states 
remains unexplored and left to future work.
In contrast, a number of the RFTS results of Secs. \ref{sec:necessary}-\ref{sec:sufficient} are directly 
applicable to, or admit analogs for, the mixed-state case:

\vspace*{1mm}

\noindent 
$\bullet$ Theorem \ref{thm:necessary_robust} constrains the correlations of a state that is to be RFTS. 
Both the statements and the proofs of these results generalize directly to the case of an arbitrary target mixed state.

\vspace*{1mm}

\noindent 
$\bullet$ The existence of a virtual subsystem decomposition of the full Hilbert space $\hilbert$, 
as described in Sec. \ref{sec:sufficient}, still ensures that a mixed target state is RFTS. 
Here, instead of the pure state being factorized with respect to $\hilbert = \bigotimes_j \hat{\hilbert}_j$, the mixed state 
must be of the form $\rho=\bigotimes_j \hat{\rho}_j$. Accordingly, the RFTS scheme employs neighborhood 
maps which prepare the mixed-state factors among the virtual subsystems 
$\mathcal{E}_j=(\hat{\rho}_j\trn{})_j\otimes\mathcal{I}_{\overline{j}}$. 

\vspace*{1mm}
\noindent 
$\bullet$ Theorem \ref{thm:lrhilbertdecomp}, involving a virtual subsystem decomposition on top of coarse-graining
and local restriction to a proper subspace of $\hilbert$, can also be generalized. Here, the local restriction is defined by 
the mixed state's subsystem support, as in Definition \ref{def:subsup}. The construction, then, is completely analogous 
to that of the pure-state case. 

\vspace*{1mm}

As the remaining results on RFTS involve the Schmidt-span projectors derived from $\ket{\psi}$, 
and an analogous object for a mixed state is not known, they cannot be directly extended. 
Among target states for which the above tools suffice, all graph product states on qudits, whose asymptotic QL stability was 
established in \cite{Johnson2016}, are 
RFTS. Interestingly, states with a graph-product structure have been 
recently shown to play a key role toward demonstrating ``quantum supremacy'' in 2D quantum simulators \cite{Eisert2017}.
Likewise, certain thermal states are also RFTS:

\begin{exmp}[\textbf{Gibbs states of virtual-product QL Hamiltonians}]
Let $H=\sum_k H_k$ on $\hilbert\simeq\bigotimes_{i=1}^N \hilbert_i$, 
and assume that a virtual factorization $\hilbert\simeq \bigotimes_{j=1}^M\hat{\hilbert}_j$ exists, such 
that (1) for all $k$ there exists a $j\equiv j_k$ with $H_k \simeq\hat{H}^k_j\otimes\identity_{\overline{j}}$;
and (2) for each $j$ there exists a $k$ such that $\mathcal{B}(\hat{\hilbert}_j)\otimes\identity_{\overline{j}} \leq \mathcal{B}(\hilbert_{\neigh_k})\otimes\identity_{\overline{\neigh_k}}$. Then, the Gibbs state
\begin{equation*}
\rho_G(H) \equiv \exp(-\beta H)/\tr{}{\exp(-\beta H)}, \beta \geq 0, 
\end{equation*}
is RFTS. This follows from the fact that each virtual-subsystem algebra is contained in a neighborhood 
algebra and $\rho_G(H)$ is a virtual product state:  
\begin{align*}
\rho_G(H) &= \exp(-\beta  \sum_k (\hat{H}^k)_{j_k}\otimes\identity_{\overline{j_k}} )/\tr{}{\exp(-\beta H)}\nonumber\\
&= \frac{1}{\tr{}{\exp(-\beta H)}}\bigotimes_{j=1}^M\exp\Big(-\beta\hspace*{-3mm}\sum_{k\,\textup{s.t.}\, j_k = j}
\hspace*{-3mm}\hat{H}^k\Big).
\end{align*}
In particular, we can conclude that the Gibbs state associated to the 
canonical graph-state Hamiltonian is RFTS.
\end{exmp}


\section{Efficiency of finite-time stabilization}
\label{sec:efficiency}

In this section we analyze the complexity of the dissipative quantum circuits required to achieve 
FTS, by addressing how the number of CPTP neighborhood maps (circuit size) and the 
degree of parallelization (circuit depth) scale with system size.  If $S$ 
consists of $N$ qudits, with total dimension $D=d^N$, and a neighborhood structure 
$\neigh$ is given, we assume that the target state is scalable, in the sense that a family of states 
$\{ \ket{\psi^{(N)}}\}$ may be defined for any $N,$ while the size of the neighborhoods and the Schmidt-span 
dimension remain the same.

\subsection{Non-robust stabilization setting}

Recall that the design of the FTS scheme we presented in Sec. \ref{sec:ftssuff} is based on two ideas: 
(i) Choose the dissipative map $\mathcal{W}$ to maximally reduce the rank of the fully mixed state; 
(ii) Choose the unitary maps $\mathcal{U}_i$ so that the subsequent action of $\mathcal{W}$ maximally 
reduces the rank of its input. The protocol ${\cal W}\circ {\cal U}_T\circ {\cal W}\circ \ldots \circ {\cal U}_1\circ {\cal W}$
then alternates the dissipative actions with the unitary ``scrambling'' of the relevant degrees of freedom.
The maximum number of neighborhood unitaries comprising each $\mathcal{U}_i$ is 
$2(D-1)^2=2(d^N-1)^2\sim \mathcal{O}(d^{2N})$ (from Proposition \ref{thm:gencoro}), 
whereas each $\mathcal{W}$ counts as a single map. 
In turn, the total number $T$ of steps needed depends on the extent to which $\mathcal{W}$ reduces the 
rank of the input density matrix. If $r$ is the maximum cooling rate, since each $\mathcal{W}$ achieves 
a rank reduction by $d^r$, then $T\sim N/r$, whereby 
the worst-case circuit size scales as $\mathcal{O}[(N/r) d^{2N}]$  \cite{remarkExample}. 

For certain neighborhood structures, the circuit depth can be reduced by acting simultaneously on 
different neighborhoods. Suppose that $\neigh$ admits ``$L$-layering,'' namely, it 
can be partitioned into $L$ sets, such that all neighborhoods in a given set are mutually disjoint. 
If the cooling rate of all neighborhoods in a particular layer is $r$, then 
instead of defining a single neighborhood-acting dissipative map $\mathcal{W}$, 
we can define a dissipative map $\mathcal{W}_i$ for each neighborhood in the layer, 
with $\mathcal{W}\equiv \prod_i \mathcal{W}_i$. 
Since $|\neigh|/L \sim N/L$ maps can now be applied in each round of 
unitaries, a rank reduction of $(d^r)^{(N/L)}$ is achieved per round, allowing to shorten the total 
number of steps to $T\sim L/r$.  Still, the scaling of the circuit size, ${\mathcal O}[(L/r) d^{2N}]$, 
remains exponential. This unfavorable scaling is  due to the compilation the neighborhood 
stabilizer unitaries making up the global stabilizer unitaries.  While this worst-case may be drastically r
educed for particular cases in principle, we turn now attention to 
the more practically relevant case of RFTS circuits, which are entirely built out of non-unitary maps.

\subsection{Robust stabilization setting}
\label{sub:circuit}

We focus on systems and neighborhood structures defined with respect to a {\em finite $m$-dimensional lattice}. 
The importance of the lattice structure of the subsystems and neighborhoods is that is affords 
a layering, as introduced before, wherein the neighborhood maps within a given layer are mutually disjoint. 
By fixing a type of QL constraint (say, next-NN as in Fig. \ref{fig:parallel}), we will show 
how, in the RFTS setting, the resulting high degree of parallelization allows to upper-bound 
the depth of the corresponding dissipative circuit by a \emph{constant}.

To appreciate the role played by the lattice structure, consider the following example of a neighborhood structure 
which is scalable yet not amenable to support a constant-depth RFTS circuit 
Let $\neigh$ be given by the set of all pairs of subsystems, giving $|\neigh|={N \choose 2}=\frac{N(N-1)}{2}$. 
The largest number of neighborhood maps which may act in parallel is $\lfloor N/2\rfloor$. 
Hence, the best possible parallelization will still require at least $|\neigh|/\lfloor N/2\rfloor=N-1$ layers of maps.
We first describe our approach to achieving constant depth in a concrete example:

\begin{figure}[t]
\includegraphics[width=0.8\columnwidth]{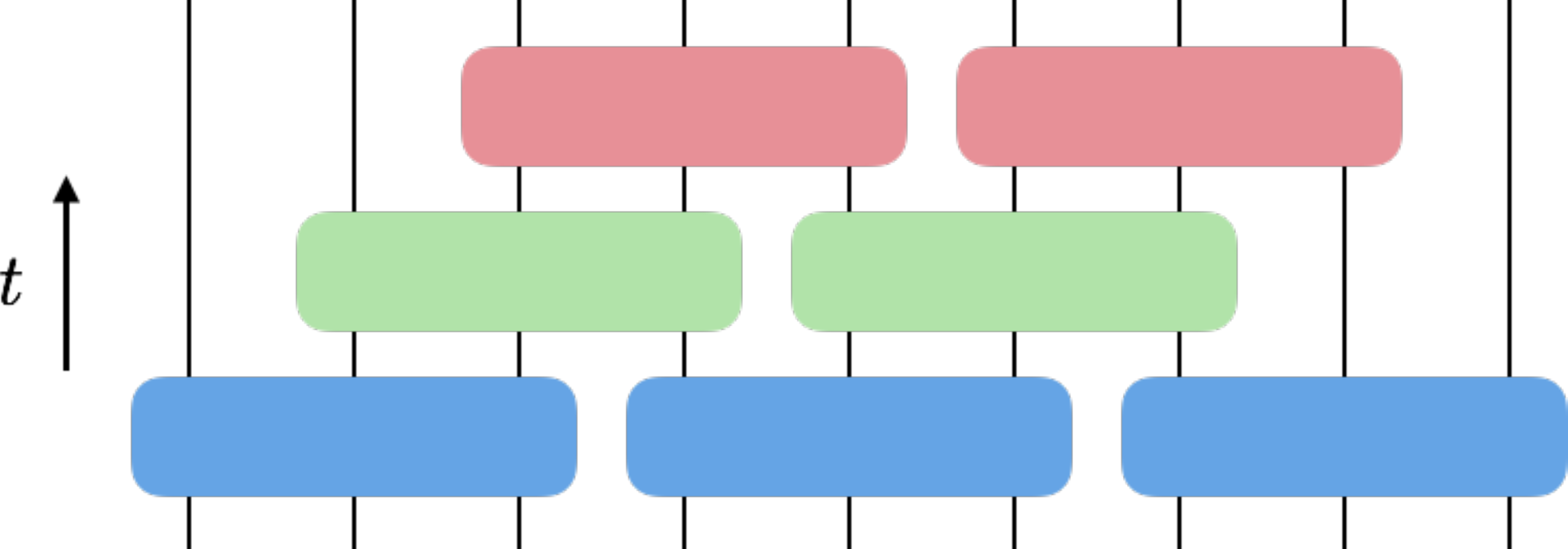}
\vspace*{-1mm}
\caption{(Color online) Given a 1D lattice of nine systems subject to a next-NN QL constraint, a state which is RFTS can be 
stabilized with a depth-$3$ QL dissipative circuit by organizing the application of maps into layers as shown. }
\label{fig:parallel}
\end{figure}

\begin{exmp}[\textbf{CCZ states on kagome lattice}]
\label{exmp:kagome}
The CCZ state we considered in Example \ref{exmp:CCZ} can be similarly defined on the kagome lattice, 
with CCZ gates acting on each triangle of systems. As depicted in Fig. \ref{fig:kagome}, to each physical system we associate 
the five-body neighborhood made of that system along with its four nearest neighbors. Similar to 
Example \ref{exmp:CCZ}, it is simple to see that the CCZ state defined on this lattice is RFTS with respect to $\neigh$.
We now show that, for arbitrary size $N$, RFTS can be achieved by a dissipative circuit of 
depth $12$. The unit cell of the kagome lattice consists of three physical systems, and, therefore, three neighborhoods 
(Fig. \ref{fig:kagome}). By translating these three physical systems and three neighborhoods by the group of lattice translations 
(generated by unit lattice vectors $\hat{e}_1$ and $\hat{e}_2$), we can obtain the set of all systems and all neighborhoods.

In a RFTS scheme, the irrelevance of the map ordering allows us to organize the neighborhood maps into layers. 
To construct a layer, consider the set of neighborhoods $\neigh^0$ in the unit cell labeled $\neigh^0_1$, $\neigh^0_2$, 
and $\neigh^0_3$  in Fig. \ref{fig:kagome}. For each direction, translate this set until it becomes disjoint with respect to 
the original set. The \emph{diameter} of the set, the maximum number of such translations needed over all directions, 
is found to be two. By translating any neighborhood in the unit cell by this diameter, the resulting neighborhood is ensured 
to be disjoint from the former. We can generate a layer of disjoint neighborhoods by repeatedly translating a unit cell neighborhood by multiples of the diameter (i.e., an even number of translations) in each direction. Three of the layers will correspond to the three neighborhoods in the unit cell. We still need to account for the neighborhoods translated by an odd number of lattice vectors in either direction. These nine remaining layers are obtained by translating each of the previous three layers by lattice translations $(0,1)$, $(1,0)$ or $(1,1)$. Thus, we have partitioned the neighborhood maps into $12$ layers. In each layer, the dissipative neighborhood maps act in parallel ensuring that, for any lattice size, the CCZ state is RFTS with respect to a depth-$12$ dissipative circuit.
\end{exmp}

\begin{figure}[t]
\includegraphics[width=0.8\columnwidth]{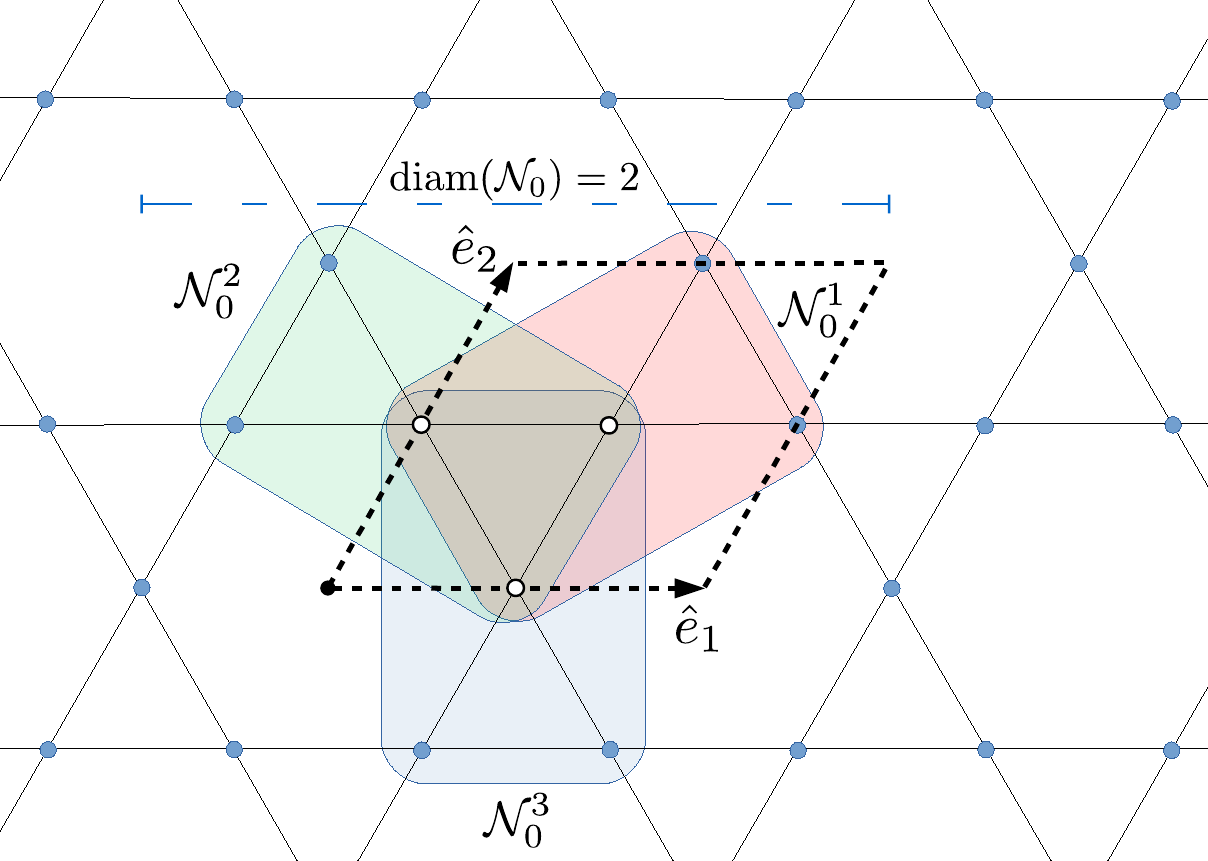}
\vspace*{-1mm}
\caption{(Color online) 
Kagome lattice with its unit cell, and neighborhood structure for the CCZ state. 
In constructing such state, the system is initialized in $\ket{+}^{\otimes N}$ and a CCZ-gate is applied to 
each triangle of adjacent systems [cf. Eq. (\ref{eq:ccz})].}
\label{fig:kagome}
\end{figure}

{The above scheme may be generalized to neighborhood structures defined on an arbitrary lattice. A lattice system is 
obtained from a unit cell containing an arrangement of ${\sf c}$ physical systems along with a discrete group of transformations, 
generated by the set of translations by $\hat{e}_1, \ldots, \hat{e}_m$. 
These can be seen as a representation of the abstract group $\mathbb{L}
\simeq \mathbb{Z}^m,$ where the $j$th component is associated to the number of (forward or backward) translations by $\hat e_j.$ 
The obtained lattice is, by construction, invariant under the action of  $\mathbb{L}$. As in Example \ref{exmp:kagome}, 
we also construct the neighborhood structure to be invariant under $\mathbb{L}$, 
by starting from a unit cell of neighborhoods, $\neigh^0,$ and thereby generating the global  $\neigh$
through translations in $\mathbb{L}$. We denote the diameter of the generating set $\textup{diam}(\neigh^0)$.} 
{In order to describe how circuit size and depth scale with $N$, we consider a sequence of finite-sized subsets of the infinite 
lattice. We take the system to be a width-$L$, $m$-dimensional hypercube of the lattice. This system contains $L^m$ 
unit cells, totaling $N = {\sf c}L^m$ subsystems and $|\neigh_0|L^m$ neighborhoods.
This induces a total $N= {\sf c} L^m$ subsystems and $|\neigh^0| L^m$ neighborhoods. For each $N$, we denote the 
corresponding neighborhood structure as $\neigh^{(N)}$. We can then bound the circuit complexity as follows:

\begin{prop}[\textbf{Lattice circuit-size scaling}]
\label{thm:latticecomplexity}
Consider an $N$-dimensional subset and neighborhood structure $\neigh^{(N)}$ on 
a $m$-dimensional lattice. If $\ket{\psi^{(N)}}$ 
is RFTS with respect to $\neigh^{(N)}$, then $\ket{\psi}$ can be stabilized by a dissipative circuit of 
size at most $|\neigh^0|(N/ {\sf c})$ and depth at most $|\neigh^0|\textup{diam}(\neigh^0)^m$.
\end{prop}}
\noindent 
As the scheme we described merely captures the essential features of the lattice toward ensuring finite depth, it is 
not guaranteed to be optimal. The following example gives a case in which a different partition of neighborhoods 
may be used to achieve improved circuit depth: 

\begin{exmp}[\textbf{Optimal depth for 2D graph states}] 
Consider graph states on the 2D square lattice. The group of lattice translations is isomorphic to 
${\mathbb L} \simeq\mathbb{Z}\times \mathbb{Z}$. Define a single neighborhood on site $(0,0)$ as that site along with the four adjacent sites, $(1,0)$, $(0,1)$, $(-1,0)$, and $(0,-1)$. 
We generate the neighborhood structure by translating this neighborhood with respect to ${\mathbb  L}$. 
Hence, there is one neighborhood per physical system and each neighborhood is labeled by an element of ${\mathbb L}$. There is one neighborhood per unit cell, and the diameter of the neighborhoods in a unit cell is $\textup{diam}(\neigh^0)=3$. Therefore, using the above scheme, we may stabilize the graph state with a circuit of depth $|\neigh^0|\,\textup{diam}(\neigh)^m=1\cdot3^2=9$. However, we can choose a different parallelization scheme which results in a depth-five circuit. By translating the neighborhood on site $(0,0)$ with just the subgroup 
${\mathbb H}\simeq \langle(1,2), (2,-1)\rangle\leq {\mathbb L}$, the generated neighborhoods are disjoint. 
The size of the coset group is $|{\mathbb L}/{\mathbb H}|=5$. Each coset 
corresponds to a layer of disjoint neighborhood maps which may act in parallel. 
The number of layers needed so that the resulting circuit includes all neighborhood maps is thus itself equal to five.
This shows that the 2D graph states on $N$ systems is RFTS with a circuit of size $N$ and depth $5$.
\end{exmp}

\section{Connection with rapid mixing}
\label{sec:rapid}

Since RFT convergence is an especially strong form of convergence, it is natural to explore 
the extent to which this may relate to the existence of continuous-time QL dynamics which 
efficiently stabilize the target state, that is, obey {\em rapid mixing} properties. After reviewing 
some relevant concepts, we show how, given a RFTS target, one may construct 
rapidly-mixing Lindblad dynamics starting from a set of stabilizing maps that commute.

\subsection{Rapidly mixing Lindblad dynamics}
\label{subsec:rapidL}

Consider a one-parameter semigroup of CPTP maps $\{\mathcal{E}_t=e^{{\cal L}t}\}_{t\geq 0},$ 
with a time-independent Lindblad generator ${\cal L}$ subject to QL constraints, ${\cal L} = \sum_j {\cal L}_{\neigh_j}\otimes 
{\cal I}_{\overline{\neigh_j}}$ (Sec. \ref{sub:nogo}). Two special CPTP maps derived from $\mathcal{L}$ are used in characterizing 
asymptotic convergence rates \cite{Wolf-lecture:12}:
\begin{itemize}
\item $\mathcal{E}_{\phi}$ is the CPTP map projecting onto the operators for which $\mathcal{L}$ has eigenvalue 
obeying $\textup{Re}(\lambda)=0$.
\item $\mathcal{E}_{\infty}$ is the CPTP map projecting onto the operators for which $\mathcal{L}$ has 
eigenvalue $\lambda=0$.
\end{itemize}
With these, the following definition provides a measure of how far the ``worst-case'' evolution is from an 
equilibrium state of the continuous-time dynamics:
{\begin{defn}
Given a CPTP map $\mathcal{E}$, its \emph{(trace-norm) contraction coefficient} is given by 
\begin{equation*}
\eta(\mathcal{E})\equiv \frac{1}{2} \sup_{\rho\geq 0, \tr{}{\rho} =1} ||(\mathcal{I}-\mathcal{E}_{\phi}) 
(\mathcal{E}(\rho))||_1.
\end{equation*}
\end{defn}
\noindent 
Since, in the stabilization settings we are interested in, $\mathcal{L}$ is engineered to have a trivial peripheral spectrum 
(no purely imaginary eigenvalues), and precisely one eigenvalue equal to zero, we can identify 
$\mathcal{E}_{\phi}=\mathcal{E}_{\infty}$ in the above.}
From the contraction coefficient, the \emph{mixing time} of the semigroup $\mathcal{E}_t=e^{{\cal L}t}$ is 
defined to be the minimum time such that $\eta(\mathcal{E}_t)=\frac{1}{2}$.
The contraction coefficient $\eta$ generated by $\mathcal{L}$ may be bounded using 
the \emph{spectral gap} $\bar{\lambda}$, namely, 
\begin{equation*}
\bar{\lambda}(\mathcal{L})\equiv \inf\{ \text{abs}(\textup{Re}(\lambda)) \, | \, \textup{Re}(\lambda)<0, \lambda \in \textup{spec}(\mathcal{L})\}.
\end{equation*}
\noindent
Note that the spectral gap of ${\cal L}$ is related to the spectral radius $\bar{\mu}_t$ of ${\cal E}_t$ 
(that is, the eigenvalue of ${\cal E}_t$ which is largest in magnitude) via $\bar{\mu}_t=e^{-\bar{\lambda} t}$ \cite{Wolf-lecture:12}. 

In the following, we will assume to have a family of states that is RFTS and scalable, in a suitable sense: 
\begin{defn} 
A {\em scalable family of RFTS pure states} $\ket{\psi^{(s)}}$, 
parametrized by $s$, is specified by:  
\begin{enumerate}
\item $\ket{\psi^{(s)}}$ on  $\Hi^{(s)}=\bigotimes_{j=1}^{N^{(s)}}\Hi_j$, where $N^{(s)}$ is monotonically 
increasing and unbounded in $s,$ and $\dim(\Hi_j)=d$ for all $j$; 
\item A neighborhood structure $\neigh^{(s)}$ such that: 
(i) the number of neighborhoods $\neigh_k^{(s)}\in\neigh^{(s)}$ scales at most {\em polynomially} in system size, that is, 
$|\neigh^{(s)}|\leq b \,N^{(s)}$, 
for some constant $b>0$; 
(ii) the  neighborhood size is {\em uniformly bounded}, that is, there exists $B>0$ 
such that $\textup{dim}(\hilbert_{\neigh_k^{(s)}}) \leq B,$ for all $k,s$.  
\end{enumerate}
\end{defn} 
\noindent 
We aim to show that a corresponding family of semigroups $\{{\cal E}^{(s)}_t\}$ satisfying 
{\em rapid mixing} relative to its unique equilibrium exists, that is, the relevant mixing time scales polynomially 
with system size:

\begin{defn}
\label{def:rm}
A family of one-parameter semigroups of CPTP maps $\{\mathcal{E}_t^{(s)}\}_{t \geq 0}$ 
satisfies \emph{rapid mixing} if there exists constants $c,\gamma,\delta > 0$ such that
\begin{equation*}
\eta(\mathcal{E}^{(s)}_t)\leq c\, [N^{(s)}]^\delta  e^{-\gamma t}, \quad \forall t\geq 0.
\end{equation*}
\end{defn} 
\noindent 
In our case, we will make use of the spectral gaps of the neighborhood Liouvillians ${\cal L}_j$, as opposed to 
those of the global Liouvillian, 
where ${\cal L}=\sum_j {\cal L}_{\neigh_j}\otimes {\cal I}_{\overline{\neigh_j}} \equiv \sum_j {\cal L}_j$.
It is easy to see that the spectral gap of a semigroup is inversely proportional to the 
operator norm of its generator. Thus, to make our results non-trivial, we impose a uniform bound on 
the norm of the neighborhood generators: $||\mathcal{L}_j|| < C$, for all $j$, for some constant $C$. 
Finally, we shall build on the following useful results: 

\begin{thm} [\cite{Reeb2012}] 
\label{thm:commutingcontraction}
Let $\mathcal{L}$ be a Liouvillian with spectral gap $\bar{\lambda}(\mathcal{L})$. 
Then, there exists $L>0$ and, for any $\nu <\bar{\lambda}(\mathcal{L})$, there exists $R>0$, such that
\begin{equation*}
 L\,e^{-\bar{\lambda}(\mathcal{L})t} \leq \eta\left(e^{\mathcal{L}t}\right)\leq R\,e^{-\nu t}.
\end{equation*}
If ${\cal L}=\sum_j{\cal L}_j$ and $\{\mathcal{L}_j\}$ commute with each other, then,
\begin{equation*}
\eta\left(e^{\mathcal{L}t}\right)=\eta\Big(e^{\sum_j\mathcal{L}_j t}\Big) \leq \sum_j 
\eta\left(e^{\mathcal{L}_jt}\right).
\end{equation*}
\end{thm}

\subsection{RFTS implies rapid mixing}
{To show that a scalable family of RFTS states 
can be associated to a rapidly-mixing semigroup, note that for all the sets of sufficient 
RFTS conditions we proposed, there exist choices of stabilizing maps that commute with each other -- 
in particular, those that stabilize each factor of the target in a QL virtual particle. Thus, without loss of 
generality, we can restrict to RFTS schemes where the neighborhood maps ${\cal E}_j$ in the sequence 
commute; if so, the corresponding neighborhood generators, ${\cal L}_j \equiv {\cal E}_j- {\cal I}$, 
also commute. We first use Theorem \ref{thm:commutingcontraction} to  upper-bound the 
contraction coefficient of sums of commuting Liouvillians, which  scales linearly in their number:

\begin{prop}[\textbf{Commuting Liouvillian contraction bound}]
\label{thm:systemsizebound}
Let $\{\mathcal{L}_j\}$ be uniformly-bounded Liouvillians, each acting on a neighborhood of uniformly-bounded size. Assume that the spectral gaps obey $\bar\lambda(\mathcal{L}_j)\geq\nu>0$, for all $j$. Then, there exists $R>0$ such that for any subset $\mathcal{S}$ of mutually commuting $\mathcal{L}_j$, we have:
\begin{equation*}
\eta(e^{\sum_{{\cal L}_j\in\mathcal{S}}\mathcal{L}_jt})\leq |\mathcal{S}|Re^{-\nu t}.
\end{equation*}
\end{prop}
\noindent 

We are  now ready for the main result of the section: commuting maps ensuring RFTS can be used 
to construct rapidly-mixing Lindblad dynamics, provided their spectral radius is bounded away from one:

\begin{thm}[\textbf{Rapid mixing for commuting RFTS}]
\label{thm:RFTSrapidmixing}
Consider a scalable family of
$ \ket{\psi^{(s)}}$ that is made RFTS  by a set of commuting neighborhood maps $\{\mathcal{E}^{(s)}_k\}$. 
Assume that there exists $\nu>0$, such that each $\lambda\in\textup{eig}(\mathcal{E}^{(s)}_k)$ satisfies either 
$\lambda=1$ or $| \lambda |< 1-\nu$. Then, there exists a family of bounded-norm QL Liouvillians ${\cal L}^{(s)}$ 
satisfying rapid mixing with respect to $ \ket{\psi^{(s)} }$.
\end{thm}

It is worth remarking that rapid mixing also ensures that the dynamics are ``stable'' with respect to local perturbations of 
the generator \cite{Lucia2015}: these stability results clearly apply to the QL Lindblad dynamics we constructed above. 
We conclude by showing that, while RFT convergence implies, in the sense we characterized, 
rapid convergence in continuous time, the converse does {\em not} hold in general: there exist 
target states which admit rapidly mixing continuous-time dynamics, yet support correlations beyond what is allowed for RFTS:

\begin{exmp}[\textbf{Non-RFTS commuting Gibbs state}] In \cite{Brandao2016} it is shown that for 
1D lattice systems, the Davies generator derived from a commuting Hamiltonian achieves rapid 
mixing with respect to the corresponding Gibbs state. Consider the 1D ferromagnetic NN Ising model, 
$H_N=-J\sum_{i=1}^{N-1} \sigma_z^i\otimes \sigma_z^{i+1}$, with $J>0$. It is well known that, in the 
thermodynamic limit, for any finite temperature, the two-point correlations of this Gibbs state are exponentially decaying 
with distance: $\tr{}{\sigma_z^{i} \otimes\sigma_z^{i+L}\rho}\sim e^{-\xi L}$, with $\xi$ finite, 
see e.g. \cite{ortiz}. Therefore, spins with disjoint neighborhood expansions are correlated, which violates the necessary 
condition for RFTS given in Theorem \ref{thm:necessary_robust}.
\end{exmp}
}

\section{Conclusions and Outlook}
\label{sec:end}

We have explored the task of {\em exactly} stabilizing a target pure state in finite time using dissipative quantum circuits 
consisting of sequences of QL CPTP maps -- a task that,  even in ideal 
conditions, may only be achieved with finite error by using continuous-time Markovian dynamics. 
In developing both necessary and sufficient conditions for FTS,
one important aim was to elucidate the role of {\em commutativity} of a frustration-free parent Hamiltonian 
one may naturally associate to the target state. 
We showed that certain cases of the well-known valence-bond-solid states {\em are} FTS relative to NN 
constraints, despite not being the ground state of any commuting NN Hamiltonian.
The remainder of the paper focused on the case where stabilization may be achieved {\em robustly}, 
i.e., independently of the order of the dissipative maps -- 
a setting which is especially attractive from both a control-theoretic and implementation standpoint.  
We developed several examples of non-trivially entangled RFTS states beyond the stabilizer formalism.  Notably, these include 
controlled-controlled-Z states, which are a universal resource for measurement-based quantum computation 
and display 2D symmetry-protected topological order \cite{Miller2015}, as well as a 
class of tensor network states obtained from a generalization of the Bravyi-Vyalyi construction \cite{Bravyi2005}, 
whereby the tensors of the construction are leveraged to build a commuting sequence of stabilizing maps.
The common feature of these examples, which ensures their RFTS, is a factorization of the target state 
with respect to a set of {\em neighborhood-compatible virtual subsystems}.
Through this decomposition, we have further clarified the role played by ``commuting structures'' in 
robustly stabilizing a state: although commutativity of the canonical parent Hamiltonian is, remarkably, 
\emph{not} necessary for RFTS, our general sufficient condition ensures the existence of some commuting 
parent Hamiltonian, although in a non-constructive way.
Finally, we showed that if a state admits such a neighborhood-respecting virtual-subsystem factorization, then there exist QL 
Lindblad dynamics which efficiently prepare this state (i.e., which satisfy rapid-mixing conditions).

We leave a number of directions for future work. 
At this point, we have only developed a non-constructive scheme for FTS in the non-robust case. Obtaining  
{\em constructive} procedures for synthesizing such stabilizing dynamics while complying with the QL constraint is an 
interesting control problem, which might be able to be tackled with methods from geometric control theory.
Likewise, in the RFTS setting, we {\em conjecture} that a virtual-subsystem factorization of the target state 
is not only sufficient, as we prove, but, in fact, needed. Were this conjecture shown to be false, such a counterexample 
would correspond to a surprising class of quantum states that are very efficient to exactly stabilize, yet exhibit entanglement 
between particles with respect to any neighborhood-respecting virtual-subsystem decomposition. 
As we mentioned, the authors of \cite{Brandao2016disc} have also characterized dissipative circuits for efficiently preparing 
(albeit {\em not} stabilizing) thermal states, under appropriate conditions. 
It would be interesting to further investigate connections between their use of the Petz recovery map \cite{Petz1988} 
and the stabilizing maps that we employ for RFTS. Related to that, 
towards extending their approach to a continuous-time scheme, it may be useful to first develop a \emph{robust} 
variation of their algorithm. The algebraic construction that we provide, along with the connections to rapid mixing, 
may shed light on such an extension.

From a more practical perspective, the assumption of error-free control dynamics we have employed throughout 
our analysis is clearly an idealization.  
In order to assess the viability of the proposed stabilization schemes, analyzing their performance against 
different kinds of implementation errors will be crucial, both in terms of obtaining rigorous, system-independent bounds for 
{\em approximate FTS} and quantitative results for specific dissipative-engineering platforms.  In this respect, 
we emphasize that implementation of discrete-time dynamics and entangled-state ``pumping'' 
via engineered dissipation has already been experimentally demonstrated in trapped ions 
\cite{Barreiro-Nature:11,Blatt2013}; 
thus trapped ions could further be natural candidates for exploring RFTS  protocols.  
Very recently, a proposal for discrete-time dissipative control in circuit QED systems has also been 
put forward \cite{Liang2016}, based on the idea of using the Fock states of a microwave-cavity mode to 
encode a $d$-dimensional target system, dispersively coupled to a transmon-qubit ancilla.  
Despite important differences (notably, the infinite-dimensional nature of the oscillator mode)
it would be very interesting to explore protocols for FT state stabilization or, more 
generally, FT encoding and quantum error correction -- by respecting 
the constraints native to such a system.

\begin{acknowledgments}

We are especially indebted to Giacomo Baggio and Jacob Miller for useful input on the  
problem of finite-time stabilization and on CCZ states, respectively.  
It is a pleasure to also thank Pierre Clare, Ramis Movassagh, and 
Salini Karuvade for valuable discussions and exchange during this project.
Work at Dartmouth was partially supported by the NSF under grants No.
PHY-1104403 and PHY-1620541, 
and the Constance and Walter Burke {\em Special Projects Fund in Quantum 
Information Science}.   P.D.J. also gratefully acknowledges support 
from a {\em Gordon F. Hull} Dartmouth graduate fellowship. F.T. has been 
partially funded by the QUINTET, QFUTURE and QCOS 
projects of the Universit\`a degli Studi di Padova.

\end{acknowledgments}


\begin{appendix}

\section{Technical proofs}
\label{sec:proofs}

We present here complete proofs of all the technical results stated in the main text.

\vspace*{2mm}

\noindent
$\bullet$ \textbf{No FT-convergence with Lindblad dynamics:}

\vspace*{1mm}

\noindent 
{\bf Proposition \ref{prop:no-go}}
{\em Consider a dynamics driven by a (time-varying) linear equation on a linear space ${\cal X}$:
\[ \dot X_t = {\cal L}_t (X_t), \quad X_0=x_0. \]
Assume that ${\cal S} \leq
{\cal X}$ is an invariant and attractive subspace for ${\cal L}_t$,
and that ${\cal L}_t$ is modulus-integrable, that is, $\int_0^t | {\cal L}_s|\,ds < \infty$, for all finite $t$.  
Then if $X_0$ does not belong to ${\cal S}$, $X_t$ will not be in ${\cal S}$ for all finite $t,$ 
namely, there cannot be exact convergence in finite time. }

\begin{proof}
Consider the orthogonal projectors $\Pi_S,\Pi_S^\perp$ on ${\cal S}$ and its orthogonal complement. 
${\cal S}$ being invariant means that
\( \Pi_S^\perp {\cal L}_t \Pi_S = 0.\)
Hence, the dynamics on the orthogonal complement $X^\perp_t =\Pi_S^\perp X_t$ is just
\[\dot X^\perp_t = \Pi_S^\perp {\cal L}_t X_t  =\Pi_S^\perp {\cal L}_t \Pi_S^\perp X_t 
\equiv  {\cal L}^\perp_t X^\perp_t .\]
By a classical result (see e.g. Ref. \cite{bellman}, Chap. 10, Section 10), 
we have that $X^\perp_t \equiv \Phi^\perp_{t,0}X^\perp_0,$ 
where the propagator $\Phi^\perp_{t,0}$ is {\em invertible} at all times (that is, it has no zero eigenvalues), 
as long as ${\cal L}^\perp_t$ is modulus-integrable.
If $\Phi^\perp_{t,0}$ is invertible at all times, then it follows that $X^\perp_t \neq 0$ for all $t$, 
for all non-zero initial conditions.  
\end{proof}

\noindent
$\bullet$ \textbf{Necessary conditions for FTS:}

\vspace*{1mm}

\noindent 
{\bf Theorem \ref{thm:ftsnec}}
{\em A pure state $\ket{\psi}$ is FTS with respect to $\neigh$ only if 
it is QLS [Eq. (\ref{eq:qls})] 
and there exists at least one neighborhood $\neigh_k\in\neigh$ for which}
\[ 2\,\textup{dim}(\Sigma_{\neigh_k}(\ket{\psi})) \leq  \textup{dim}(\hilbert_{\neigh_k}). \]

\begin{proof}
{We first prove that $\ket{\psi}$ being FTS implies that $\ket{\psi}$ is QLS.
By negation, assume that $\ket{\psi}$ does not satisfy Eq. (\ref{eq:qls}). Then there exists some $\ket{\phi}\notin\textup{span}(\ket{\psi})$, for which $\ket{\phi}\in\bigcap_{k}\overline{\Sigma}_{\neigh_k}(\ket{\psi})$. Any $\ket{\psi}$-preserving neighborhood map $\mathcal{E}_{k}$ must fix all states in $\overline{\Sigma}_{\neigh_k}(\ket{\psi})$. Any sequence of 
such maps fixes $\ketbra{\phi}$ and, hence, it cannot map $\ketbra{\phi}$ into $\ketbra{\psi}$, as is required for FTS. 

We continue by showing that the remaining small Schmidt span condition is also necessary.
Assuming that $\ket{\psi}$ is FTS, let $\mathcal{E}_T\ldots\mathcal{E}_1(\cdot)=\ketbra{\psi}\tr{}{\cdot}$ be a sequence of CPTP stabilizing maps. 
Then, there must exist some map $\mathcal{E}_k$ for which $\mathcal{E}_k(\sigma)=\ketbra{\psi}$ 
for some $\sigma\neq \ketbra{\psi}$. Using the locality and $\ket{\psi}$-invariance of 
$\mathcal{E}_k$, we show that the condition $\mathcal{E}_k(\sigma)=\ketbra{\psi}$ places an upper bound 
on the Schmidt rank of $\ket{\psi}$ with respect to the $\neigh_k|\overline{\neigh}_k$ bipartition. 
The analysis is made easier considering a purification of the equation $\mathcal{E}_k(\sigma)=\ketbra{\psi}$, the purified equation being linear in $\ket{\psi}$. In purifying, ancilla systems must be introduced for both $\sigma$ and $\mathcal{E}_k$. 
Letting $\hilbert_A$ be the state space of the ancilla purifying $\sigma$, we have
\begin{align*}
\sigma \rightarrow\,\, &\ket{\phi}\in\hilbert_A\otimes\hilbert_{\overline{\neigh}_k}\otimes\hilbert_{\neigh_k},
\quad \tr{A}{\ketbra{\phi}}=\sigma.
\end{align*}
Letting $\hilbert_B$ be the ancilla purifying $\mathcal{E}_k$, we obtain an isometry representation,
\begin{align*}
 \mathcal{E}_k \rightarrow\,\, & V: \hilbert_{\overline{\neigh}_k}\otimes\hilbert_{\neigh_k}\rightarrow 
 \hilbert_{\overline{\neigh}_k}\otimes\hilbert_{\neigh_k}\otimes\hilbert_B,\nonumber\\
& 
\hspace*{-10mm} V = \identity_{\overline{\neigh}_k}\otimes \tilde{V}_{\neigh_k\rightarrow\neigh_k B},
\quad \tr{B}{V\cdot V^{\dagger}}=\mathcal{E}_k(\cdot).
\end{align*}
Accordingly, $\mathcal{E}_k(\sigma)=\ketbra{\psi}$ becomes
\begin{align*}
\ketbra{\psi}_{\neigh_k\overline{\neigh}_k} & =\tr{B}{V \sigma V^{\dagger}} 
\nonumber\\
&
=\tr{AB}{(\identity_A\otimes V)\ketbra{\phi}(\identity_A\otimes V)^{\dagger}}.
\end{align*}
Hence, $(\identity_A\otimes V)\ket{\phi}$ is some pure state, which, upon tracing out $AB$, 
leaves the pure state $\ket{\psi}$. Therefore,
\begin{equation}
\label{eq:finalstep}
 (\identity_A\otimes V)\ket{\phi}=\ket{\lambda}_{AB}\ket{\psi}_{\neigh_k\overline{\neigh}_k},
\end{equation}
where $\ket{\lambda}_{AB}$ is some pure state on $\hilbert_A\otimes\hilbert_B$.

The invariance condition, $\mathcal{E}_k(\ketbra{\psi})=\ketbra{\psi}$, constrains the form of the isometry 
$V$. Invariance requires
\begin{align*}
&&\tr{B}{(\identity_{\overline{\neigh}_k}\otimes \tilde{V}_{\neigh_k\rightarrow\neigh_k B})\ketbra{\psi}(\identity_{\overline{\neigh}_k}\otimes\tilde{V}_{\neigh_k\rightarrow\neigh_k B})^{\dagger}} 
\\ 
&&=\ketbra{\psi}. 
\end{align*}
Hence, $(\identity_{\overline{\neigh}_k}\otimes \tilde{V}_{\neigh_k\rightarrow\neigh_k B})\ket{\psi}$ is some pure state, which, upon tracing out $B$, leaves the pure state $\ket{\psi}$. Therefore, $(\identity_{\overline{\neigh}_k}\otimes \tilde{V}_{\neigh_k\rightarrow\neigh_k B})\ket{\psi}=\ket{0}_B\otimes\ket{\psi}$, where $\ket{0}_B$ is some pure state on $B$. This ensures that $\tilde{V}_{\neigh_k\rightarrow\neigh_k B}$ acts trivially on $\Sigma_{\neigh_k}(\ket{\psi})$, outputting $\ket{0}$ on $B$.
The action of $\tilde{V}_{\neigh_k \rightarrow\neigh_k B}$ on $\Sigma_{\neigh_k}(\ket{\psi})^{\perp}$, 
which we denote $\tilde{V}^{\perp}_{\neigh_k\rightarrow\neigh_k B}$, is unconstrained as of yet. 
In summary, invariance ensures that $\tilde{V}_{\neigh_k\rightarrow\neigh_k B}$ acts trivially on
$\Sigma_{\neigh_k}(\ket{\psi})$, giving 
$\tilde{V}_{\neigh_k \rightarrow\neigh_k B}=\Pi_{\neigh_k}\otimes\ket{0}_B + 
\tilde{V}^{\perp}_{\neigh_k \rightarrow\neigh_k B},$
where $\Pi_{\neigh_k}$ is the projector onto $\Sigma_{\neigh_k}(\ket{\psi})$, and 
$\tilde{V}^{\perp}_{\neigh_k\rightarrow\neigh_k B}$ satisfies $\tilde{V}^{\perp}_{\neigh_k\rightarrow\neigh_k B}\Pi_{\neigh_k}=0$. 
Trace-preservation of $\mathcal{E}_k$ constrains $\tilde{V}^{\perp}_{\neigh_k\rightarrow\neigh_k B}$. In terms of $V$, 
the latter requires $\identity=V^{\dagger}V(=\mathcal{E}_k^{\dagger}(\identity))$. Evaluating this in terms of the above decomposition yields 
\begin{align*}
\identity_{\neigh_k \overline{\neigh_k}}&=[\Pi_{\neigh_k}+(\Pi_{\neigh_k}\otimes\bra{0}_B)
\tilde{V}^{\perp}_{\neigh_k\rightarrow\neigh_k B}+\text{H.c.} \nonumber\\
&+(\tilde{V}^{\perp}_{\neigh_k \rightarrow\neigh_k B})^{\dagger}\tilde{V}^{\perp}_{\neigh_k\rightarrow\neigh_k B}]\otimes\identity_{\overline{\neigh}_k}.
\end{align*}
The non-trivial part of this equation is on system $\neigh_k$, where the equation may be block-decomposed as
\begin{equation*}
\left[
\begin{array}{c|c}
\identity & 0 \\
\hline
0 & \identity
\end{array}
\right]=
\left[
\begin{array}{c|c}
\Pi_{\neigh_k} & (\Pi_{\neigh_k}\otimes\bra{0}_B)\tilde{V}^{\perp} \\
\hline
((\Pi_{\neigh_k}\otimes\bra{0}_B)\tilde{V}^{\perp})^\dagger & (\tilde{V}^{\perp})^\dagger \tilde{V}^{\perp}
\end{array}
\right],
\end{equation*}
showing that $(\Pi_{\neigh_k}\otimes\bra{0}_B) \tilde{V}_{\neigh_k\rightarrow\neigh_k B}^{\perp}=0$. With these conditions on $V$, we return to Eq. (\ref{eq:finalstep}), which becomes 
\begin{align}
\label{eq:assumedeq}
\ket{\lambda}_{AB}\ket{\psi}_{\overline{\neigh}_k\neigh_k}&=\identity_A\otimes V\ket{\phi}\nonumber\\ 
&=(\identity_A\otimes\identity_{\overline{\neigh}_k}\otimes\Pi_{\neigh_k})\ket{\phi}\otimes\ket{0}_B \nonumber\\
&+(\identity_A\otimes\identity_{\overline{\neigh}_k}\otimes \tilde{V}^{\perp}_{\neigh_k\rightarrow\neigh_k B})\ket{\phi}.
\end{align}
Decompose this equation 
into three parts according to:
\begin{align*}
\hilbert\simeq&\, \textup{supp}(\identity_{A\overline{\neigh}_k}\otimes\Pi_{\neigh_k}\otimes\ketbra{0}_B)\nonumber\\
&\oplus\textup{supp}(\identity_{A\overline{\neigh}_k}\otimes\Pi_{\neigh_k}\otimes(\identity-\ketbra{0}_B))\nonumber\\
&\oplus\textup{supp}(\identity_{A\overline{\neigh}_k}\otimes(\identity-\Pi_{\neigh_k})\otimes\identity_B).
\end{align*}
The vector $\ket{\lambda}_{AB}\ket{\psi}_{\overline{\neigh}_k\neigh_k}$ lies entirely in the first two blocks. Later, we will also need to use the fact that $\identity_A\otimes(\identity-\ketbra{0})_B\ket{\lambda}_{AB}\neq 0$. This follows from $\ket{\lambda}_{AB}\ket{\psi}_{\overline{\neigh}_k \neigh_k}$ having a non-trivial part in the second block, which we now show to follows from the assumption $\sigma \notin\textup{span}(\ketbra{\psi})$.

Assume, by contradiction, that the norm-1 vector $\ket{\lambda}_{AB}\ket{\psi}_{\overline{\neigh}_k\neigh_k}$ 
lies completely in the first block, and let 
$\ket{\lambda^0}_A\equiv (\identity_A\otimes\bra{0}_B)\ket{\lambda}_{AB}$.
Then, $\ket{\lambda}_{AB}\ket{\psi}_{\overline{\neigh}_k\neigh_k}=(\identity_A\otimes\ketbra{0}_B\otimes\identity_{\overline{\neigh}_k}\otimes\Pi_{\neigh_k})\ket{\lambda}_{AB}\ket{\psi}_{\overline{\neigh}_k\neigh_k}=\ket{\lambda^0}_A\ket{0}_B\ket{\psi}_{\overline{\neigh}_k \neigh_k}$, where $\|\ket{\lambda^0}\|=1$.
Projecting the right-hand side of Eq. (\ref{eq:assumedeq}) into the first block yields 
\begin{align*}
&\ket{\lambda^0}_A\ket{\psi}_{\overline{\neigh}_k\neigh_k}\ket{0}_B=(\identity_{A\overline{\neigh}_k}\otimes\Pi_{\neigh_k})\ket{\phi}_{A\overline{\neigh_k}\neigh_k}\ket{0}_B\nonumber\\
&+\identity_{A\overline{\neigh}_k}\otimes[(\Pi_{\neigh_k}\otimes\bra{0}_B) \tilde{V}^{\perp}_{\neigh_k \rightarrow \neigh_k B}]\ket{\phi}_{A\overline{\neigh}_k\neigh_k}\ket{0}_B.
\end{align*}
The last term is zero, as $(\Pi_{\neigh_k}\otimes\bra{0}_B)\tilde{V}_{\neigh_k \rightarrow\neigh_k B}^{\perp}=0$. Then, removing the common factor of $\ket{0}_B$ from the remaining terms, we have
$ \ket{\lambda^0}_A\ket{\psi}_{\overline{\neigh}_k \neigh_k}=(\identity_{A\overline{\neigh}_k}\otimes\Pi_{\neigh_k})\ket{\phi}.$
The vector $\ket{\lambda^0}_A\ket{\psi}_{\overline{\neigh}_k \neigh_k}$ is assumed to be norm-1, and hence $(\identity_{A\overline{\neigh}_k}\otimes\Pi_{\neigh_k})\ket{\phi}$ is as well. Since $\Pi_{\neigh_k}$ is a projector, $\|(\identity_{A\overline{\neigh}_k}\otimes\Pi_{\neigh_k})\ket{\phi}\|=\|\ket{\phi}\|$ only if $(\identity_{A\overline{\neigh}_k}\otimes\Pi_{\neigh_k})\ket{\phi}=\ket{\phi}=\ket{\lambda^0}_A\ket{\psi}_{\overline{\neigh}_k\neigh_k}$. Tracing out system $A$, this last equation becomes $\ketbra{\psi}=\tr{A}{\ketbra{\phi}}=\sigma$. 
{This contradicts $\sigma \not \in \textup{span}(\ketbra{\psi})$, so we conclude that 
$\ket{\lambda}_{AB}\ket{\psi}_{\overline{\neigh}_k\neigh_k}$ lies, at least partly, in the second block.} From this, it follows that 
\begin{align*}
&& 0\neq (\identity_{A \overline{\neigh}_k \neigh_k}\otimes(\identity-\ketbra{0})_B)\ket{\lambda}_{AB}
\ket{\psi}_{\overline{\neigh}_k \neigh_k} \\
&& \equiv \ket{\lambda^{\perp}}_{AB}\ket{\psi}_{\overline{\neigh}_k\neigh_k}.
\end{align*}
Defining the matrix 
$\hat{V}\equiv\Pi_{\neigh_k}\otimes(\identity_B-\ketbra{0}_B)\tilde{V}^{\perp}_{\neigh_k\rightarrow\neigh_k B}$, 
the second block equation reads
\begin{equation}
\label{eq:veceq}
\ket{\lambda^{\perp}}_{AB}\ket{\psi}_{\overline{\neigh}_k\neigh_k}=(\identity_A\otimes\hat{V}\otimes \identity_{\overline{\neigh}_k} )\ket{\phi}.
\end{equation}

Towards bounding the Schmidt rank of $\ket{\psi}$, it is useful to transform the above vector equation into a matrix equation by applying partial-transpose to the composite Hilbert space $\hilbert_A\otimes\hilbert_{\overline{\neigh}_k}$. Hence, the vectors $\ket{\lambda^{\perp}}_{AB}$, $\ket{\psi}_{\overline{\neigh}_k\neigh_k}$, and $\ket{\phi}_{A\overline{\neigh}_k \neigh_k}$ are transformed into matrices:
$$\lambda^{\perp}:\hilbert_A\rightarrow\hilbert_B \;
\psi:\hilbert_{\overline{\neigh}_k}\rightarrow\hilbert_{\neigh_k} ,\;
\phi:\hilbert_{A\overline{\neigh}_k}\rightarrow\hilbert_{\neigh_k} .$$
Note that $\textup{rank}(\psi)=\textup{dim}(\Sigma_{\neigh_k}(\ket{\psi}))$. 
Eq. (\ref{eq:veceq}) is transformed into the matrix equation
$\lambda^{\perp}\otimes\psi=\hat{V}\phi.$
It follows that $\textup{rank}(\lambda^{\perp}\otimes\psi)=\textup{rank}(\hat{V}\phi)$. On the left hand side, 
\[ \textup{dim}(\Sigma_{\neigh_k}(\ket{\psi}))=\textup{rank}(\psi)\leq \textup{rank}(\lambda^{\perp}\otimes\psi), 
\]
using the fact that $\ket{\lambda^{\perp}}_{AB}\neq 0$, as shown earlier. On the right hand side, 
\[ \textup{rank}(\hat{V}\phi)\leq\textup{rank}(\hat{V})\leq \textup{dim}(\Sigma_{\neigh_k}(\ket{\psi})^{\perp}),
\]
where the last inequality follows from $\textup{ker}(\hat{V})\geq \textup{ker}(\tilde{V}^{\perp}_{\neigh_k \rightarrow\neigh_k B})\geq\Sigma_{\neigh_k}(\ket{\psi})$. The above two inequalities together imply that $\textup{dim}(\Sigma_{\neigh_k}(\ket{\psi}))\leq\textup{dim}(\Sigma_{\neigh_k}(\ket{\psi})^{\perp})$. With $\hilbert_{\neigh_k} \simeq\Sigma_{\neigh_k}(\ket{\psi})\oplus\Sigma_{\neigh_k}(\ket{\psi})^{\perp}$, we obtain
$ \textup{dim}(\hilbert_{\neigh_k}) \geq 2 \,{\textup{dim}(\Sigma_{\neigh_k}(\ket{\psi}))}$, as claimed.  
}
\end{proof}

\vspace*{1mm}

\noindent
$\bullet$ \textbf{Unitary generation property:}

\vspace*{1mm}

We first develop a few results which build up to a proof of Proposition \ref{thm:unitarygeneration}. The following is a repurposing of Proposition 1.2.2 in A. Borel and H. Bass, {\em Linear Algebraic Groups} (W. A. Benjamin, 1969), to our setting:

\begin{prop}
\label{thm:gencoro}
Consider a Hilbert space, $\hilbert\simeq\bigotimes_i \hilbert_i$, of finite dimension $D=\prod_i d_i$, a neighborhood structure $\neigh$, and a pure state $\ket{\psi}\in\hilbert$. Let $\mathcal{U}_{\neigh_k,\ket{\psi}}$ be the neighborhood stabilizer groups of $\ket{\psi}$. Then, for any element $U\in\langle \mathcal{U}_{\neigh_k,\ket{\psi}}\rangle_k$, there exists a sequence of at most $2(D-1)^2$ elements $U_j$, drawn from the $\mathcal{U}_{\neigh_k,\ket{\psi}}$, such that $U=U_1\ldots U_{2(D-1)^2}$. Further, 
the group $\langle \mathcal{U}_{\neigh_k,\ket{\psi}}\rangle_k$ generated by the neighborhood stabilizer groups is connected.
\end{prop}

\begin{lem} 
\label{thm:alggen}
Let a Lie subgroup $H$ of a Lie group $G$ be generated by connected Lie subgroups $H_{\alpha}$, $\alpha\in A$ for some set $A$. Then the Lie algebra $\mathfrak{h}$ of $H$ is generated by the Lie algebras $\mathfrak{h}_{\alpha}$ of the corresponding subgroups $H_{\alpha}$.
\end{lem}
\begin{proof}
Let $\hat{\mathfrak{h}}\equiv \langle \mathfrak{h}_{\alpha}\rangle_{\alpha}\subseteq\mathfrak{h}$, and $\hat{H}\subseteq H$ the corresponding connected Lie subgroup. To show that two Lie algebras are equal, $\hat{\mathfrak{h}}=\mathfrak{h}$, it suffices to show that their Lie groups are equal, $\hat{H}=H$. Thus, it remains to show that $\hat{H}\supseteq H$. Since each $H_{\alpha}$ is connected, each is the exponential of its Lie algebra $\mathfrak{h}_{\alpha}$. Hence, $\hat{\mathfrak{h}}\supseteq \mathfrak{h}_{\alpha}$ for all $\alpha$ implies $\hat{H}\supseteq H_{\alpha}$ for all $\alpha$. $H$ is the smallest Lie subgroup of $G$ containing all $H_{\alpha}$. Thus, $\hat{H}$ being a group and $\hat{H}\supseteq H_{\alpha}$ for all $\alpha$ implies that $\hat{H}\supseteq H$. Finally, since $\hat{H}=H$, we have $\hat{\mathfrak{h}}=\mathfrak{h}$, as desired.
\end{proof}

\begin{lem} 
\label{thm:groupalgequiv}
$\langle \mathfrak{u}_{\neigh_k,\ket{\psi}}\rangle_k=\mathfrak{u}_{\ket{\psi}}$ if and only if $\langle \mathcal{U}_{\neigh_k,\ket{\psi}}\rangle_k=\mathcal{U}_{\ket{\psi}}$.
\end{lem}
\begin{proof}
($\Leftarrow $) Equal Lie groups have equal Lie algebras.
($\Rightarrow $) Assume that $\langle \mathfrak{u}_{\neigh_k,\ket{\psi}}\rangle_k=\mathfrak{u}_{\ket{\psi}}$, and let $\tilde{\mathfrak{u}}$ be the Lie algebra of $\langle \mathcal{U}_{\neigh_k,\ket{\psi}}\rangle_k$. $\mathcal{U}_{\ket{\psi}}$ and all of the $\mathcal{U}_{\neigh_k,\ket{\psi}}$ are connected Lie subgroups of $\mathcal{U}(\hilbert)$. Therefore, they are equal to the exponential of their respective Lie algebras. Furthermore, Proposition \ref{thm:gencoro} ensures that $\langle \mathcal{U}_{\neigh_k, \ket{\psi}}\rangle_k$ is connected due to the connectedness of the $\mathcal{U}_{\neigh_k,\ket{\psi}}$. Hence, $\langle\mathcal{U}_{\neigh_k,\ket{\psi}}\rangle_k$ and $\mathcal{U}_{\ket{\psi}}$ being connected implies that if $\langle\mathcal{U}_{\neigh_k,\ket{\psi}}\rangle_k=\mathcal{U}_{\ket{\psi}}$, then $\tilde{\mathfrak{u}}=\mathfrak{u}_{\ket{\psi}}$. The $\mathcal{U}_{\neigh_k,\ket{\psi}}$ being connected ensures, by Lemma \ref{thm:alggen}, that the Lie algebra $\tilde{\mathfrak{u}}$ is generated by the Lie algebras $\mathfrak{u}_{\neigh_k,\ket{\psi}}$. Thus, it follows that 
$\langle \mathcal{U}_{\neigh_k,\ket{\psi}}\rangle_k=\mathcal{U}_{\ket{\psi}}$. 
\end{proof}

\noindent 
{\bf Proposition \ref{thm:unitarygeneration} (\textbf{Unitary generation property})}
{\em Given a state $\ket{\psi}$ and a neighborhood structure $\neigh$, any element in $\mathcal{U}_{\ket{\psi}}$ can be decomposed into a finite product of elements in $\mathcal{U}_{\neigh_k,\ket{\psi}}$ if and only if}
\begin{align*}
\langle \mathfrak{u}_{\neigh_k,\ket{\psi}} \rangle_{k}=\mathfrak{u}_{\ket{\psi}},
\end{align*}
{\em where}
$\langle \cdot \rangle_{k}$ {\em denotes the smallest Lie algebra which contains all Lie algebras from the set indexed by $k$.}
\begin{proof}

($\Leftarrow$) 
By Lemma \ref{thm:groupalgequiv}, the Lie algebra generation implies the Lie group generation. By Proposition \ref{thm:gencoro}, then, it follows that elements of $\mathcal{U}_{\ket{\psi}}$ decompose into finite products of elements of the $\mathcal{U}_{\neigh_k,\ket{\psi}}$. 

($\Rightarrow$) $\langle \mathfrak{u}_{\neigh_k,\ket{\psi}} \rangle_{k}\leq\mathfrak{u}_{\ket{\psi}}$ is true by construction. We show that $\langle \mathfrak{u}_{\neigh_k,\ket{\psi}} \rangle_{k}\geq\mathfrak{u}_{\ket{\psi}}$ under the decomposition assumption. Consider an arbitrary element $X\in \mathfrak{u}_{\ket{\psi}}$. Then, by assumption, $\textup{exp}(X)=U\in\mathcal{U}_{\ket{\psi}}$ admits a decomposition, $U=U_{T}\ldots U_1$, with each $U_i$ in a $\mathcal{U}_{\neigh_k,\ket{\psi}}$. As the $\mathcal{U}_{\neigh_k,\ket{\psi}}$ are connected, each element in $\mathcal{U}_{\neigh_k,\ket{\psi}}$ is the exponentiation of an element in $\mathfrak{u}_{\neigh_k,\ket{\psi}}$: $U_i=\textup{exp}(X_i)$. Hence, $U=\textup{exp}(X_T)\ldots \textup{exp}(X_1)$. Iterating the Baker-Campbell-Hausdorff formula, we can write $U=\textup{exp}(Y)$ for some element $Y\in \langle \mathfrak{u}_{\neigh_k,\ket{\psi}} \rangle_{k}$. Then, since $U=\textup{exp}(X)=\textup{exp}(Y)$, $X$ and $Y$ are proportional to one another, ensuring $\lambda Y=X\in \langle \mathfrak{u}_{\neigh_k,\ket{\psi}} \rangle_{k}$, for $\lambda\in\mathbb{R}$. 
\end{proof}

\vspace*{1mm}

\noindent
$\bullet$ \textbf{Sufficient conditions for FTS:}

\vspace*{1mm}

\noindent 
{\bf Theorem \ref{thm:fts}.}
{\em A state $\ket{\psi}$ is FTS relative to a connected neighborhood structure $\neigh$ if there exists 
at least one neighborhood 
$\neigh_k\in \neigh$ satisfying the small Schmidt span condition, 
$2\,\textup{dim}(\Sigma_{\neigh_k}(\ket{\psi})) \leq  \textup{dim}(\hilbert_{\neigh_k})$, 
and the unitary generation property holds, 
$\langle \mathfrak{u}_{\neigh_\ell,\ket{\psi}}\rangle_\ell=\mathfrak{u}_{\ket{\psi}}$.}

\begin{proof}
We construct a finite sequence of CPTP maps which is guaranteed to stabilize $\ket{\psi}$
under the given assumptions. Let dim$(\Sigma_{\neigh_k}(\ket{\psi}))\equiv s_k$.  
The Schmidt span dimension condition ensures that
 $s_k \leq \textup{dim}(\Sigma_{\neigh_k}(\ket{\psi})^{\perp})$. 
Then, we can choose any subspace $\Sigma_{\neigh_k}^1\leq\textup{dim}(\Sigma_{\neigh_k}(\ket{\psi})^{\perp})$, such that 
$\textup{dim}(\Sigma_{\neigh_k}^1)=s_k$. 
We think of $\Sigma_{\neigh_k}^1$ as a ``copy'' of $\Sigma_{\neigh_k}(\ket{\psi})$ lying inside 
$\Sigma_{\neigh_k}(\ket{\psi})^{\perp}$. For convenience, 
{let $\Sigma_{\neigh_k}(\ket{\psi}) \equiv \Sigma_{\neigh_k}^0$ (further abbreviated to just $\Sigma^0$ in the main text). Then, $\hilbert_{\neigh_k}=\Sigma_{\neigh_k}^0 \oplus\Sigma_{\neigh_k}^1\oplus {\cal R}$, where ${\cal R}$ is the remaining subspace of $\hilbert_{\neigh_k}$. 
Choosing an identification between $\Sigma^0_{\neigh_k}$ and $\Sigma_{\neigh_k}^1$, we can write
$\Sigma^0_{\neigh_k}\oplus\Sigma_{\neigh_k}^1\simeq \complex^2\otimes\Sigma_{\neigh_k},$ 
such that
\begin{align*}
\Sigma_{\neigh_k}^0\simeq \ket{0}\otimes\Sigma_{\neigh_k},\quad
\Sigma_{\neigh_k}^1\simeq \ket{1}\otimes\Sigma_{\neigh_k}.
\end{align*}
The first CPTP map in our FTS sequence is 
$ \mathcal{E}_{\neigh_k}\equiv(\ketbra{0}\trn{}\otimes\mathcal{I}_{\Sigma_{\neigh_k}} ) \oplus\mathcal{I}_{\cal R},$
with respect to $\hilbert_{\neigh_k}=(\complex^2\otimes\Sigma_{\neigh_k})\oplus {\cal R}$. 
The corresponding global map is $\mathcal{W}\equiv\mathcal{E}_{\neigh_k}\otimes\mathcal{I}_{\overline{\neigh}_k}$. 
The global Hilbert space decomposes as
\begin{equation} 
\hspace*{-2mm}
\hilbert\simeq\hilbert_{\neigh_k}\otimes\hilbert_{\overline{\neigh}_k}=(\complex^2\otimes\Sigma_{\neigh_k}\otimes\hilbert_{\overline{\neigh}_k})\oplus (\mathcal{R}\otimes\hilbert_{\overline{\neigh}_k}),
\label{Adec}
\end{equation}
whereby the target state can be written as $\ket{\psi}=(\ket{0}\otimes\ket{\tilde{\psi}})\oplus 0$. From this 
and the definition of $\mathcal{W}$, we can see that $\mathcal{W}(\ketbra{\psi})=\ketbra{\psi}$, as desired.
Furthermore, the only state orthogonal to $\ket{\psi}$ whose density operator is mapped to $\ketbra{\psi}$ is $\ket{\psi'}\equiv(\ket{1}\otimes\ket{\tilde{\psi}})\oplus 0$. Hence, we may interpret $\mathcal{W}$ as correcting an arbitrary error $U$ acting on the qubit subsystem associated to $\mathbb{C}^2$ in Eq. (\ref{Adec}).
The strategy we employ towards stabilizing $\ket{\psi}$ is  to iterate the following procedure: (1) apply a sequence of invariance-satisfying, neighborhood unitaries to map a state $\ket{\alpha}$ to $\ket{\psi'}$; (2) apply $\mathcal{W}$ to map $\ket{\psi'}$ to $\ket{\psi}$.}

For each $\ket{\alpha}\in\ker(\bra{\psi})$, let the unitary transformation to be used in step (1) be defined by 
$ U_{\alpha}\equiv\ketbra{\psi}\oplus(\ket{\psi'}\bra{\alpha}+\ket{\alpha}\bra{\psi'})\oplus\identity.$ 
This unitary has a non-trivial action only on $\textup{span}\{\ket{\alpha},\ket{\psi'}\}$, acting as $U_{\alpha}\ket{\alpha}=\ket{\psi'}$. Thus, the composition of $U_{\alpha}$ and $\mathcal{W}$ gives a map which takes $\ket{\alpha}$ to $\ket{\psi}$, as needed. We label the corresponding CPTP map with $\mathcal{U}_{\alpha}(\cdot)\equiv U_{\alpha}\cdot U_{\alpha}^{\dagger}$.

From Proposition \ref{thm:unitarygeneration}, we know that the assumed property $\langle \mathfrak{u}_{\neigh_\ell,\ket{\psi}}\rangle_{\ell}=\mathfrak{u}_{\ket{\psi}}$ ensures that any $U\in\mathcal{U}_{\ket{\psi}}$ can be decomposed into a finite product of invariance-satisfying, neighborhood unitaries. Since any $U_{\alpha}$, with $\inprod{\psi}{\alpha}=0$ is in $\mathcal{U}_{\ket{\psi}}$, such a $U_{\alpha}$ may be composed from a finite sequence of $\ket{\psi}$-preserving neighborhood maps. 

Finally, we construct the sequence of CPTP maps which renders $\ket{\psi}$ FTS. Let $\{\ket{\alpha}\}$ label an orthonormal basis set for $\textup{ker}(\bra{\psi})$ with the following ordering:
\begin{align*}
&\ket{0'}\equiv \ket{\psi'} ,\nonumber\\
&\textup{span}\{\ket{1'},\ldots,\ket{d_k'}\} \equiv (\ket{0}\otimes\Sigma_{\neigh_k})\otimes\hilbert_{\overline{\neigh}_k})\ominus\textup{span}(\ket{\psi}), \nonumber\\
&\textup{span}\{\ket{(d_k+1)'},\ldots,\ket{T'}\}\equiv {\cal R}\otimes\hilbert_{\overline{\neigh}_k}.
\end{align*}
Then, we define the FTS sequence of CPTP as
$ \mathcal{E}\equiv\mathcal{W}\circ\mathcal{U}_{T'}\circ\ldots\circ\mathcal{W}\circ\mathcal{U}_{1'}\circ\mathcal{W}\circ\mathcal{U}_{0'}.$
The individual maps manifestly satisfy invariance. It remains to show that $\mathcal{E}=\ketbra{\psi}\trn{}$. It suffices to show that $\mathcal{E}(\identity)=D\ketbra{\psi}$, where $D=\text{dim}(\hilbert)$. In the first step, 
\begin{align*}
\mathcal{W}\mathcal{U}_{0'}(\identity)&=\mathcal{W}(\identity)\nonumber\\
&=\mathcal{W} \{[(\ketbra{0}+\ketbra{1})\otimes\identity]\oplus 0 +(0\otimes 0)\oplus \identity \}\nonumber\\
&=2(\ketbra{0}\otimes\identity)\oplus 0 +(0\otimes 0)\oplus \identity\nonumber\\
&=2\Pi_{\neigh_k}+\Pi_{\cal R},
\end{align*}
where the decomposition used in the second and third lines is given in Eq. (\ref{Adec})
and $\Pi_{\cal R}$ is the projector onto ${\cal R}\otimes\hilbert_{\overline{\neigh}_k}$. 
In the next step,
\begin{align*}
\mathcal{W}\mathcal{U}_{1'}\mathcal{W}\mathcal{U}_{0'}(\identity)&=
\mathcal{W} [ 2\,\mathcal{U}_{1'}(\Pi_{\neigh_k}) +\mathcal{U}_{1'}(\Pi_{\cal R}) ] \nonumber\\
 &=\mathcal{W} \{2 \,[\mathcal{U}_{1'}(\Pi_{\neigh_k}-\ketbra{1'})\nonumber\\
 & \hspace*{4mm} +\mathcal{U}_{1'}(\ketbra{1'})]+\Pi_{\cal R}\} \nonumber\\
 &=\mathcal{W} [ 2( \Pi_{\neigh_k}-\ketbra{1'})+ {2 \ketbra{\psi'}}   +\Pi_{\cal R} ]\nonumber\\
 &=2\, \mathcal{W}(\Pi_{\neigh_k}-\ketbra{1'})  \nonumber\\
 & \hspace*{4mm} +2\, \mathcal{W}(\ketbra{\psi'})+\mathcal{W}(\Pi_{\cal R})\nonumber\\
 &=2\, (\Pi_{\neigh_k}-\ketbra{1'})+2\ketbra{\psi}+\Pi_{\cal R},
\end{align*} 
\noindent
where we used the fact that all operators have support in $(\ket{0}\otimes\Sigma_{\neigh_k}(\ket{\psi})\oplus {\cal R})\otimes\hilbert_{\overline{\neigh}_k}$, on which $\mathcal{W}$ acts trivially. Similarly, in the next step we have,
\begin{align*}
 \mathcal{W}\mathcal{U}_{2'}\mathcal{W}\mathcal{U}_{1'}\mathcal{W}\mathcal{U}_{0'}(\identity)&=2 \,
 (\Pi_{\neigh_k}-\ketbra{1'}-\ketbra{2'})\nonumber\\
&+4\, \ketbra{\psi}+\Pi_{\cal R}.
\end{align*}
Continuing until $\ket{d_k'}$, we obtain
\[ \mathcal{W}\circ\mathcal{U}_{d_k'}\circ\ldots\circ\mathcal{W}\circ\mathcal{U}_{1'}\circ\mathcal{W}\circ\mathcal{U}_{0'}(\identity)=2 s_k\ketbra{\psi}+\Pi_{\cal R},\]
where recall that $s_k=\text{dim}(\Sigma_{\neigh_k}(\ket{\psi}))$. At this point, we continue the sequence with the unitaries 
transferring vectors from ${\cal R}\otimes\hilbert_{\overline{\neigh}_k}$ to $\ket{\psi'}$, namely,
\begin{align*}
\mathcal{W}\mathcal{U}_{(d_k+1)'}(2 s_k\ketbra{\psi}+\Pi_{\cal R})&=[s_k+(s_k+1)]\ketbra{\psi}\nonumber\\
&\hspace*{-8mm}+\Pi_{\cal R} -\ketbra{(d_k+1)'}.
\end{align*}
Continuing in this way, the sequence terminates at
\begin{align*}
 \mathcal{W}\mathcal{U}_{T'}\mathcal{W}\mathcal{U}_{(d_k+1)'}(2 s_k\ketbra{\psi}+\Pi_R)
 & =(s_k+T)\ketbra{\psi} \\
 & = D \ketbra{\psi},
\end{align*}
which establishes the desired result. 
\end{proof}

\vspace*{1mm}

\noindent
{\bf Proposition \ref{thm:qlsalg}.}
{\em If $\ket{\psi}$ satisfies $\langle \mathfrak{u}_{\neigh_k,\ket{\psi}} \rangle_{k}=\mathfrak{u}_{\ket{\psi}}$ 
with respect to the neighborhood structure $\neigh$, then $\ket{\psi}$ satisfies Eq. (\ref{eq:qls}), and hence is QLS, 
with respect to $\neigh$.}

\begin{proof}
{By contradiction, assume that $\ket{\psi}$ does not satisfy Eq. (\ref{eq:qls}). Then $\bigcap_k \overline{\Sigma}_{\neigh_k}(\ket{\psi})=\mathcal{S}>\textup{span}(\ket{\psi})$. We show that $\langle \mathfrak{u}_{\neigh_k,\ket{\psi}} \rangle_{k} \leq\mathfrak{u}_{\mathcal{S}}$, 
where $\mathfrak{u}_{\mathcal{S}}$ is the Lie algebra associated to the Lie group $\mathcal{U}_{\mathcal{S}}$ that stabilizes 
$\mathcal{S}$. }The defining property of $\mathcal{U}_{\mathcal{S}}$ is that for all $U\in\mathcal{U}_{\mathcal{S}}$, $U\ket{s}\bra{s'}U^{\dagger}=\ket{s}\bra{s'}$ for all $\ket{s},\ket{s'}\in\mathcal{S}$. Then, the defining property of the corresponding Lie algebra is that for all $X\in\mathfrak{u}_{\mathcal{S}}$,  $[X,\ket{s}\bra{s'}]=0$ for all $\ket{s},\ket{s'}\in\mathcal{S}$. Consider an arbitrary neighborhood $\neigh_k$ and an element $Y\in\mathfrak{u}_{\neigh_k,\ket{\psi}}$. By definition, $Y$ satisfies $[Y,\ket{r}\bra{r'}]=0$ for all $\ket{r},\ket{r'}\in\overline{\Sigma}_{\neigh_k}(\ket{\psi})$. Since $\mathcal{S}\leq\overline{\Sigma}_{\neigh_k}(\ket{\psi})$, we have $[Y,\ket{s}\bra{s'}]=0$ for all $\ket{s},\ket{s'}\in\mathcal{S}$. Thus, $Y\in\mathfrak{u}_{\mathcal{S}}$. As this inclusion holds for all elements in $\mathfrak{u}_{\neigh_k,\ket{\psi}}$ for any $k$, and $\mathfrak{u}_{\mathcal{S}}$ is closed with respect to linear combination and Lie product, then any element in $\langle \mathfrak{u}_{\neigh_k,\ket{\psi}} \rangle_{k}$ is contained in $\mathfrak{u}_{\mathcal{S}}$. Since $\mathcal{S}>\textup{span}(\ket{\psi})$, we have $\mathfrak{u}_{\mathcal{S}}<\mathfrak{u}_{\ket{\psi}}$. 
{Finally, since $\langle \mathfrak{u}_{\neigh_k,\ket{\psi}} \rangle_{k}\leq\mathfrak{u}_{\mathcal{S}}$,
it follows that $\langle \mathfrak{u}_{\neigh_k,\ket{\psi}} \rangle_{k}<\mathfrak{u}_{\ket{\psi}}$, which contradicts the assumption.
}
\end{proof}

\vspace*{1mm}

\noindent
$\bullet$ \textbf{Necessary conditions for RFTS:}

\vspace*{1mm}

\noindent In order to prove Proposition \ref{thm:complementcommutecondition}, we need some preliminary results. 
We start with a lemma which constrains the form of stabilizing CPTP maps:

\begin{lem}
\label{thm:invarianceoutputlem}
If a CPTP map $\mathcal{E}_k$ acting on neighborhood $\neigh_k$ preserves 
$\ket{\psi}$, then, for arbitrary $\rho$, $\mathcal{E}_k$ satisfies
\[ \Pi_k \mathcal{E}_k(\rho) \Pi_k = \Pi_k \rho \Pi_k + \Pi_k\sigma \Pi_k,\]
where $\Pi_k$ is the orthogonal projector onto $\overline{\Sigma}_{\neigh_k}(\ket{\psi})$ and $\sigma\equiv 
\mathcal{E}_k(\Pi^{\perp}_k \rho \Pi^{\perp}_k)\geq 0$.
\end{lem}

\noindent
{\em Proof.}
If $\mathcal{E}_k$ is to preserve $\ket{\psi}$, the Kraus operators of $\mathcal{E}_k$ must act trivially on $\textup{supp}[\trn{\overline{\neigh_k}}(\ket{\psi}\bra{\psi})]$. This requires the form 
{$K_i \equiv \lambda_i \Pi_k + R_i\Pi_k^{\perp},$ $\lambda_i \in {\mathbb C}$, for some operator $R_i$ \cite{BT-TAC:10}.}
In turn, trace preservation of $\mathcal{E}_k$ requires that
\[ \identity = \sum_i|\lambda_i|^2\Pi_k+\lambda_i^*\Pi_kR_i\Pi_k^{\perp}+\lambda_i\Pi_k^{\perp}R_i^{\dagger}\Pi_k+\Pi_k^{\perp}R_i^{\dagger}R_i\Pi_k^{\perp}.
\]
From this, it follows that $\sum_i|\lambda_i|^2=1$, $\Pi_k^{\perp}\sum_iR_i^{\dagger}R_i\Pi_k^{\perp}=\Pi_k^{\perp}$, and, most importantly,
\( \Pi_k(\sum_i\lambda_i^*R_i)\Pi_k^{\perp}=0.\)
Finally, applying these conditions to $\Pi_k \mathcal{E}_k(\rho) \Pi_k$, we find
\begin{align*}
\hspace*{1cm}
\Pi_k \mathcal{E}_k(\rho) \Pi_k&=\Pi_k(\sum_i |\lambda_i|^2\Pi_k\rho\Pi_k+ \lambda_i\Pi_k\rho\Pi_k^{\perp}R_i^{\dagger}\\ 
& + \text{H.c.}+ R_i\Pi_k^{\perp}\rho\Pi_k^{\perp}R_i^{\dagger})\Pi_k\\
&=\Pi_k\rho\Pi_k+ \Pi_k\rho(\Pi_k^{\perp}\sum_i\lambda_i R_i^{\dagger}\Pi_k)\\
& + \text{H.c.}+ \sum_i \Pi_kR_i\Pi_k^{\perp}\rho\Pi_k^{\perp}R_i^{\dagger}\Pi_k\\
&=\Pi_k\rho\Pi_k+\Pi_k\sigma\Pi_k.\hspace*{2.5cm}\Box
\end{align*}

\noindent 
We will also make use of the following trace inequality:
\begin{lem} 
\label{lem:commpen}
 Let $\Pi_1$ and $\Pi_2$ be orthogonal projectors onto subspaces ${\cal V}_1,{\cal V}_2,$ respectively, and $\Pi_{1\cap 2}$ the projector onto ${\cal V}_1\cap{\cal V}_2,$. Then
\[ \rm{Tr}\,({\Pi_1\Pi_2}) \geq \rm{Tr}\,{\Pi_{1\cap 2}}+\frac{1}{2}\rm{Tr}\,({|[\Pi_1,\Pi_2]|^2}).\]
\end{lem}

\begin{proof}
First, note that
$\frac{1}{2}|[\Pi_1,\Pi_2]|^2=\frac{1}{2}(\Pi_1\Pi_2\Pi_1+\Pi_2\Pi_1\Pi_2-(\Pi_1\Pi_2)^2-(\Pi_2\Pi_1)^2).$
Taking the trace of both sides and rearranging terms, 
\[  \text{Tr}\,({\Pi_1\Pi_2})=\text{Tr}({(\Pi_1\Pi_2)^2})+\frac{1}{2}\text{Tr}\,({|[\Pi_1,\Pi_2]|^2}).\]
Observe that, under conjugation, $\Pi_{1\cap 2}(\Pi_1\Pi_2)^2\Pi_{1\cap 2}=\Pi_{1\cap 2}$. Using that trace is non-increasing under conjugation with respect to a projector, we obtain $\tr{}{(\Pi_1\Pi_2)^2}\geq\tr{}{\Pi_{1\cap 2}(\Pi_1\Pi_2)^2\Pi_{1\cap 2}}=\tr{}{\Pi_{1\cap 2}}$. Making this replacement, the result follows.
\end{proof}

\vspace*{1mm}

\noindent 
{\bf Proposition \ref{thm:complementcommutecondition} (\textbf{Commuting FF Hamiltonian})}
{\em If a target pure state $\ket{\psi}$ is RFTS with respect to neighborhood structure $\neigh$, then 
$[\Pi_k,\Pi_{\overline{k}}]=0$ for all neighborhoods $\neigh_k$, where $\Pi_k$ and $\Pi_{\overline{k}}$ are 
the orthogonal projectors onto 
$\overline{\Sigma}_{\neigh_k}(\ket{\psi})$ and $\cap_{{j\neq k}}\overline{\Sigma}_{\neigh_j}(\ket{\psi})$,
respectively.}

\begin{proof}
Assume that $\ket{\psi}$ is RFTS under $\mathcal{E}_T\circ\ldots\mathcal{E}_1$.
Let $\mathcal{E}_k$ be the neighborhood map on $\neigh_k$ and $\mathcal{E}_{\overline{k}}$
be the composition of the remaining neighborhood maps.
Robust stabilizability implies $\mathcal{E}_{k}\circ\mathcal{E}_{\overline{k}}=\ket{\psi}\bra{\psi}\rm{Tr}$, whereas invariance
requires $\mathcal{E}_j(X)=X$ for any $X\in\overline{\Sigma}_{\neigh_j}(\ketbra{\psi})$. Since $\Pi_{\overline{k}}\in \overline{\Sigma}_{\neigh_j}(\ketbra{\psi})$ for all $j\neq k$, each $\mathcal{E}_j$ with $j\neq k$ must fix $\Pi_{\overline{k}}$. Hence, we have
\[ \mathcal{E}_{\overline{k}}(\Pi_{\overline{k}})=\Big(\prod_{j\neq k}\mathcal{E}_j\Big)(\Pi_{\overline{k}})=\Pi_{\overline{k}}.\]
Thus, applying the full sequence of CPTP maps to $\Pi_{\overline{k}}$, we also have
$\ket{\psi}\bra{\psi}\trn{}(\Pi_{\overline{k}})=\mathcal{E}_{k}\circ\mathcal{E}_{\overline{k}}(\Pi_{\overline{k}})
=\mathcal{E}_{k}(\Pi_{\overline{k}}).$
{Conjugating the equation 
$\ket{\psi}\bra{\psi}\trn{}(\Pi_{\overline{k}})=\mathcal{E}_{k}(\Pi_{\overline{k}})$
with respect to $\Pi_k$, we can apply 
Lemma \ref{thm:invarianceoutputlem} to obtain}
\begin{align}
\label{eq:P}
\ket{\psi}\bra{\psi}\trn{}(\Pi_{\overline{k}})&=\Pi_k\Pi_{\overline{k}}\Pi_k+\Pi_k\sigma\Pi_k,
\end{align}
where $\sigma$ is a positive-semidefinite operator. Next, conjugating both sides of Eq. (\ref{eq:P}) 
with respect to the projector $\tilde{\Pi}_k\equiv\Pi_k-\ket{\psi}\bra{\psi}$ kills the left hand-side, while 
leaving the sum of two positive semidefinite operators on the right hand-side,
$0=\tilde{\Pi}_k\Pi_k\Pi_{\overline{k}}\Pi_k\tilde{\Pi}_k+\tilde{\Pi}_k\Pi_k\sigma\Pi_k\tilde{\Pi}_k.$
Since such a sum is zero only if both matrices are zero, taking the trace of the first zero matrix gives
\begin{align*}
0&=\trn{}(\tilde{\Pi}_k\Pi_k\Pi_{\overline{k}})
=\trn{}((\Pi_k-\ket{\psi}\bra{\psi})\Pi_{\overline{k}})\\
&=\trn{}(\Pi_k\Pi_{\overline{k}})-\trn{}(\Pi_{k\cap{\overline{k}}})
\geq\trn{}(|[\Pi_k,\Pi_{\overline{k}}]|^2), 
\end{align*}
where the last inequality follows from Lemma \ref{lem:commpen}. This can only be true if $[\Pi_k,\Pi_{\overline{k}}]=0$. As
the above arguments are made for an arbitrary $\neigh_k$, they must hold for all neighborhoods, whereby 
the conclusion follows.
\end{proof}

\vspace*{1mm}

\noindent
{\bf Theorem \ref{thm:necessary_robust}.}
{\em Let the pure state $\rho=\ketbra{\psi}$ be RFTS with respect to $\neigh$.  Then the following properties hold:

\vspace*{1mm}

\textbf{\em (i)  (Finite correlation)} For any two subsystems 
$A$ and $B$ having disjoint neighborhood expansions (i.e., 
$\mathcal{A}^{\neigh}\cap\mathcal{B}^{\neigh}=\emptyset$), 
arbitrary observables $X_A$ and $Y_B$  are uncorrelated, that is, 
$\tr{}{X_A Y_B \rho} =\tr{}{X_A \rho}\tr{}{Y_B \rho}$.

\vspace*{1mm}

\textbf{\em (ii) (Recoverability property)} If a map $\mathcal{M}$ acts non-trivially only on subsystem $A$, $\mathcal{M}\equiv\tilde{\mathcal{M}}_A\otimes\mathcal{I}_{\overline{A}}$, then there exists a sequence of CPTP neighborhood maps 
$\mathcal{E}_{j}$, each acting only on $A^{\neigh}$, such that $\rho=\mathcal{E}_l\circ\ldots\circ\mathcal{E}_1\circ\mathcal{M}(\rho)$.

\vspace*{1mm}

\textbf{\em (iii) (Zero CMI)} For any two subsets of subsystems $A$ and $B$, with $A^{\neigh}\cap B = \emptyset$, the quantum conditional mutual information (CMI), $I(A:B|C)_{\rho}\equiv S(A,C)+S(B,C)-S(A,B,C)-S(C)$, satisfies $I(A:B|C)_{\rho}=0$, where $C\equiv A^{\neigh}\backslash A$.}

\vspace*{1mm}

\noindent 
{\em Proof.} {\bf (i)}
Since $\rho$ is RFTS with respect to $\neigh$, there exists a sequence of neighborhood maps such that $\rho\rm{Tr}=\mathcal{E}_T\ldots\mathcal{E}_1$. Let $\mathcal{E}_{A^{\neigh}}$ be the composition of all such maps which act non-trivially on $A$, and similarly for $\mathcal{E}_{B^{\neigh}}$ with $B$. By assumption, $\mathcal{E}_{A^{\neigh}}$ and $\mathcal{E}_{B^{\neigh}}$ act disjointly. Let $\mathcal{E}_{\textup{rest}}$ be the composition of the remaining maps. By the robustness assumption, we may reorder the maps to write $\rho=\mathcal{E}_{\textup{rest}}\mathcal{E}_{A^{\neigh}}\mathcal{E}_{B^{\neigh}}(\sigma_{A^{\neigh}}\otimes\tau_{B^{\neigh}}\otimes\omega_{\overline{A^{\neigh}B^{\neigh}}})$, for arbitrary input density operators. Let $X_A$ and $Y_B$ be arbitrary observables acting on $A$ and $B$. We have $\tr{}{X_A Y_B \rho}=\tr{}{X_A Y_B \mathcal{E}_{\textup{rest}}\mathcal{E}_{A^{\neigh}}\mathcal{E}_{B^{\neigh}}(\sigma_{A^{\neigh}}\otimes\tau_{B^{\neigh}}\otimes\omega_{\overline{A^{\neigh}B^{\neigh}}})}$. Since $\mathcal{E}_{\textup{rest}}^{\dagger}$ is unital and both $X_A$ and $Y_B$ are trivial where the map acts, the latter expression simplifies to $\tr{}{X_A Y_B \rho}=\tr{}{X_A \overline{\mathcal{E}}_{A^{\neigh}}(\sigma_{A^{\neigh}})\otimes Y_B\overline{\mathcal{E}}_{B^{\neigh}}(\tau_{B^\neigh})}$, where $\overline{\mathcal{E}}_{A^{\neigh}}$, $\overline{\mathcal{E}}_{B^{\neigh}}$ are defined to act on their respective systems and we have traced out $\omega_{\overline{A^{\neigh}B^{\neigh}}}$. The trace can be separated as $\tr{}{X_A Y_B \rho}=\tr{}{X_A \overline{\mathcal{E}}_{A^{\neigh}}(\sigma_{A^{\neigh}})}\tr{}{Y_B\overline{\mathcal{E}}_{B^{\neigh}}(\tau_{B^\neigh})}$. 
For the remaining steps, we first note that 
\begin{align*}
& \tr{}{X_A \overline{\mathcal{E}}_{A^{\neigh}}(\sigma_{A^{\neigh}})}  \nonumber\\
& \hspace*{4mm}=\tr{}{X_A \overline{\mathcal{E}}_{A^{\neigh}}(\sigma_{A^{\neigh}})\otimes\overline{\mathcal{E}}_{B^{\neigh}}(\tau_{B^\neigh})\otimes\omega_{\overline{A^{\neigh}B^{\neigh}}}},\\
&\tr{}{Y_B\overline{\mathcal{E}}_{B^{\neigh}}(\tau_{B^\neigh})}\nonumber\\
& \hspace*{4mm}=\tr{}{\overline{\mathcal{E}}_{A^{\neigh}}(\sigma_{A^{\neigh}})\otimes Y_B\overline{\mathcal{E}}_{B^{\neigh}}(\tau_{B^\neigh})\otimes\omega_{\overline{A^{\neigh}B^{\neigh}}}}.
\end{align*}
Finally, with $\mathcal{E}_{\textup{rest}}$ being trace-preserving, we may re-insert it into the trace to obtain
\begin{align*}
&\tr{}{X_A Y_B \rho}\nonumber =\tr{}{X_A \overline{\mathcal{E}}_{A^{\neigh}}(\sigma_{A^{\neigh}})}\tr{}{Y_B\overline{\mathcal{E}}_{B^{\neigh}}(\tau_{B^\neigh})}\nonumber\\
&=\tr{}{\mathcal{E}_{\textup{rest}}[X_A \overline{\mathcal{E}}_{A^{\neigh}}(\sigma_{A^{\neigh}})\otimes\overline{\mathcal{E}}_{B^{\neigh}}(\tau_{B^\neigh})\otimes\omega_{\overline{A^{\neigh}B^{\neigh}}}]}\nonumber\\
&\times \tr{}{\mathcal{E}_{\textup{rest}}[\overline{\mathcal{E}}_{A^{\neigh}}(\sigma_{A^{\neigh}})\otimes Y_B\overline{\mathcal{E}}_{B^{\neigh}}(\tau_{B^\neigh})\otimes\omega_{\overline{A^{\neigh}B^{\neigh}}}]}\nonumber\\
&=\tr{}{X_A \mathcal{E}_{\textup{rest}}\mathcal{E}_{A^{\neigh}}\mathcal{E}_{B^{\neigh}}(\sigma_{A^{\neigh}}\otimes\tau_{B^{\neigh}}\otimes\omega_{\overline{A^{\neigh}B^{\neigh}}})}\nonumber\\
&\times \tr{}{Y_B \mathcal{E}_{\textup{rest}}\mathcal{E}_{A^{\neigh}}\mathcal{E}_{B^{\neigh}}(\sigma_{A^{\neigh}}\otimes\tau_{B^{\neigh}}\otimes\omega_{\overline{A^{\neigh}B^{\neigh}}})}\nonumber\\
&=\tr{}{X_A\rho}\tr{}{Y_B\rho}, 
\end{align*}
where in the second-to-last step we have used the fact that $\mathcal{E}_{\textup{rest}}$ acts trivially on $X_A$ and $Y_B$.

\vspace*{1mm}

\noindent 
{\bf (ii)}  
Let $\mathcal{E}'_k$ be the sequence of neighborhood maps which renders $\rho$ RFTS. Define the subsequence of maps 
$\mathcal{E}_{A^\neigh}\equiv\prod_{\neigh_k\cap A\neq \emptyset}\mathcal{E}'_k$, and let $\mathcal{E}_{R}$ be the product of 
the remaining $\mathcal{E}'_k$. We then have that $\mathcal{E}_{R}\circ\mathcal{E}_{A^\neigh}(\sigma)=\rho$ for any density operator $\sigma$. We show that $\mathcal{E}_{A^\neigh}$, acting on the transformed target state, $\mathcal{M}(\rho)$, recovers $\rho$:
\begin{align*}
\mathcal{E}_{A^\neigh}\circ\mathcal{M}(\rho) &= \mathcal{E}_{A^\neigh}\circ\mathcal{M}\circ\mathcal{E}_{R}\circ\mathcal{E}_{A^\neigh}(\sigma)\nonumber\\
&= \mathcal{E}_{A^\neigh}\circ\mathcal{E}_{R}\circ\mathcal{M}\circ\mathcal{E}_{A^\neigh}(\sigma)\nonumber\\
&= \mathcal{E}_{A^\neigh}\circ\mathcal{E}_{R}(\rho')\nonumber = \rho,
\end{align*}
where $\sigma$ is any density operator and $\rho' \equiv \mathcal{M}\circ\mathcal{E}_{A^\neigh}(\sigma)$.

\vspace*{1mm}

\noindent 
{\bf (iii)} {By specializing property (ii) to the case where $\mathcal{M}=(\tau_A\trn{A})\otimes\mathcal{I}_{\overline{A}}$, with $\tau_A$ being the completely mixed state on $A$, we have $\mathcal{E}^{A^\neigh}\mathcal{M}(\rho)=(\tau_A\trn{A})\otimes\mathcal{I}_{\overline{A}} (\rho)
=\mathcal{E}^{A^\neigh}(\tau_A\otimes\rho_{\overline{A}})$.} Then, using the fact that $A^{\neigh}\cap B = \emptyset$, we trace out all but $A^{\neigh}$ and $B$ (i.e., all but systems $ABC$) to obtain $\rho_{ABC}=\tr{\overline{ABC}}{\rho}=\tr{\overline{ABC}}{\mathcal{E}^{A^\neigh}(\tau_A\otimes\rho_{\overline{A}})}=\mathcal{E}^{A^\neigh}(\tau_A\otimes\rho_{BC})$.
Since $\rho$ is written as a short quantum Markov chain, we have $I(A:B|C)_{\rho}=0$. \hfill$\Box$

\vspace*{3mm}

\noindent 
$\bullet$ \hspace*{-0mm}\textbf{Non-constructive RFTS sufficient conditions:} 

\vspace*{1mm}

\noindent 
{\bf Theorem \ref{thm:lrhilbertdecomp} (\textbf{Neighborhood factorization on local restriction})}
{\em A state $\ket{\psi}$ of the coarse-grained subsystems associated to $\hilbert\simeq\bigotimes_{i=1}^N\hilbert_i$ is 
RFTS with respect to the neighborhood structure $\neigh$ if:}
\begin{enumerate}
\item 
{\em There exists a locally restricted space $\tilde{\hilbert}=\bigotimes_{i=1}^N\tilde{\hilbert}_i$ that admits a 
virtual-subsystem decomposition of the form $\tilde{\hilbert}=\bigotimes_{j=1}^M \hat{\hilbert}_j$, such that}
\begin{equation*}
\ket{\psi}=\bigotimes_{j=1}^M\ket{\hat{\psi}_j}\oplus 0 \in \Big( \bigotimes_{j=1}^M \hat{\hilbert}_j\Big) \oplus \hilbert^0 ,
\end{equation*}
{\em where $\hilbert^0\simeq\tilde{\hilbert}^{\perp}$;} 
\item {\em For each virtual subsystem $\hat{\hilbert}_j$, there exists a neighborhood $\neigh_k$ such that} 
\begin{equation*}
\mathcal{B}(\hat{\hilbert}_j)\otimes\identity_{\overline{j}}\oplus \identity^0 \leq \mathcal{B}(\hilbert_{\neigh_k})\otimes\identity_{\overline{\neigh_k}}.
\end{equation*}
\end{enumerate}

\noindent
{\em Proof.}
Assume that the above conditions hold. We construct a finite sequence of commuting QL CPTP maps which robustly stabilize $\rho$. First, we construct the maps which prepare the locally restricted space. Define the map $\mathcal{E}^0_i:\mathcal{B}(\hilbert_i)\rightarrow \mathcal{B}(\hilbert_i)$ to be $$\mathcal{E}^0_i(\cdot)\equiv P_i\cdot P_i^{\dagger}+\frac{P_i}{\tr{}{P_i}}\tr{}{(\identity-P_i)\cdot},$$ 
where $P_i$ is the projector onto $\tilde{\hilbert}_i$. For each $\neigh_k$, we construct a map $\mathcal{E}^0_{k}\equiv\bigotimes_{i\in\neigh_k}\mathcal{E}^0_i$, which prepares support on the locally restricted space of all coarse-grained subsystems contained in $\neigh_k$.

On the virtual systems, let $\hat{\mathcal{E}}_j:\mathcal{B}(\hilbert^0\oplus\hat{\hilbert}_j\otimes\hat{\hilbert}_{\overline{j}})\rightarrow \mathcal{B}(\hilbert^0\oplus\hat{\hilbert}_j\otimes\hat{\hilbert}_{\overline{j}})$ as 
\( \hat{\mathcal{E}}_j(\cdot)=\mathcal{I}^0\oplus(\hat{\rho}_j\trn{})_j\otimes\mathcal{I}_{\overline{j}}.\)
Each virtual subsystem labeled $j$ is associated to a neighborhood $\neigh_k$ on which its operators act non-trivially. Correspondingly, each neighborhood-acting map $\mathcal{E}_j$ is constructed from $\hat{\mathcal{E}}_j$ by pre-composing it with $\mathcal{E}^0_{k}$,
$ \mathcal{E}_j\equiv \hat{\mathcal{E}}_j\circ\mathcal{E}^0_{k}(\cdot).$
The Kraus operators of $\hat{\mathcal{E}}_j$ are contained in $\identity^0\oplus\mathcal{B}(\hat{\hilbert}_j)\otimes\identity_{\overline{j}}$. Hence, by $\identity^0\oplus\mathcal{B}(\hat{\hilbert}_j)\otimes\identity_{\overline{j}} \leq \mathcal{B}(\hilbert_{\neigh_k})\otimes\identity_{\overline{\neigh_k}}$, we have that the Kraus operators of $\mathcal{E}_j$ act non-trivially only on  $\neigh_k$. Thus, each map $\mathcal{E}_j$ is a valid neighborhood map. Finally, we must show that an arbitrary sequence of these maps prepares $\rho$ while leaving it invariant. For invariance, we have $\mathcal{E}_j(\rho)=\hat{\mathcal{E}}_j\mathcal{E}^0_{k}(\rho)=\hat{\mathcal{E}}_j(0\oplus\bigotimes_{j=1}^M\hat{\rho}_j)=0\oplus\bigotimes_{j=1}^M\hat{\rho}_j=\rho.$ To prove preparation of $\rho$, we use the fact that $\mathcal{E}^0_i\mathcal{E}_j=\mathcal{E}_j\mathcal{E}^0_i$. 
Consider an arbitrary complete sequence of the neighborhood maps,
\begin{align*}
\mathcal{E}_M\circ\ldots\circ\mathcal{E}_1&=(\hat{\mathcal{E}}_M\circ\mathcal{E}^0_{M})\circ\mathcal{E}_{M-1}\circ\ldots\circ\mathcal{E}_2\circ\mathcal{E}_1\nonumber\\
&=\hat{\mathcal{E}}_M\circ\mathcal{E}_{M-1}\circ\ldots\circ\mathcal{E}_2\circ\mathcal{E}_1\circ\mathcal{E}^0_{M} .
\end{align*}
We continue in this way, using the commutativity of the support projections with the $\mathcal{E}_j$ to move all of the support projections to act first. Since every coarse-grained particle will have been accounted for, we may combine the action of all of these projections $\mathcal{E}^0_{k}$ into a single projection $\mathcal{E}^{\ell}$ which has the effect of projecting onto $\tilde{\hilbert}$,
\begin{align*}
\mathcal{E}_M\circ\ldots\circ\mathcal{E}_1&=(\hat{\mathcal{E}}_M\circ\ldots\circ\hat{\mathcal{E}}_1)\circ(\mathcal{E}^0_M\circ\ldots\circ\mathcal{E}^0_1)\nonumber\\
&=\hat{\mathcal{E}}_M\circ\ldots\circ\hat{\mathcal{E}}_1\circ\mathcal{E}^{\ell}.
\end{align*}
Finally, we see that the composition of these maps constitutes a preparation of the target state,
\begin{align*}
\hspace*{1cm}
\mathcal{E}_M\circ\ldots\circ\mathcal{E}_1&=\mathcal{I}^0\oplus\bigotimes_{j=1}^M(\hat{\rho}_j\trn{})\circ\mathcal{E}^{\ell}\nonumber\\
&=(0\oplus\bigotimes_{j=1}^M\hat{\rho}_j)\trn{} =\rho\trn{}.\hspace*{1.7cm}\Box
\end{align*}


\noindent 
$\bullet$ \textbf{Algebraic sufficient conditions for RFTS:}

\vspace*{1mm}

\noindent 
{\bf Proposition \ref{thm:algfac} (\textbf{Algebraically induced factorization})}
{\em If a set of algebras $\{\mathcal{A}_j\}$, $\mathcal{A}_i\in\mathcal{B}(\hilbert)$, is complete and commuting, then each $\mathcal{A}_j$ has a trivial center and there exists a decomposition of the Hilbert space $\hilbert\simeq\bigotimes_{j=1}^T\hat{\hilbert}_j$ for which $\mathcal{A}_j\simeq\mathcal{B}(\hat{\hilbert}_j)\otimes\identity_{\overline{j}}$ for each $j$.}

\begin{proof}
First, by contradiction, 
assume that a neighborhood algebra $\mathcal{A}_j$ were reducible, so that there exists some $X\in\mathcal{A}_j$ where $X\in\mathcal{A}_j'$, but $X\neq c\cdot \identity$. As $X\in\mathcal{A}_j$, it commutes with all elements of $\mathcal{A}_k$ for $k\neq j$, and hence the algebra generated by all the neighborhood algebras has a non-trivial commutant, which violates completeness. We obtain the Hilbert space factorization as follows. For any algebra $\mathcal{A}_j$ with trivial center acting on $\hilbert$, there exists a decomposition $\hilbert\simeq\hat{\hilbert}_j\otimes\hilbert_{\overline{j}}$ for which $\mathcal{A}_j=\mathcal{B}(\hat{\hilbert}_j)\otimes\identity_{\overline{j}}$. Starting with $\mathcal{A}_1$, we have $\hilbert\simeq\hat{\hilbert}_1\otimes\hilbert_{\overline{1}}$. From this, $\mathcal{A}_1'=\identity_{1}\otimes\mathcal{B}(\hat{\hilbert}_{\overline{1}})$. As the algebras are all commuting, $\mathcal{A}_2\leq\mathcal{A}_1'=\identity_{1}\otimes\mathcal{B}(\hat{\hilbert}_{\overline{1}})$. Hence, $\mathcal{A}_2$ carries a natural action on $\hilbert_{\overline{1}}$, and $\mathcal{A}_2$ having a trivial center implies that there is a decomposition $\hilbert_{\overline{1}}=\hat{\hilbert}_2\otimes\hilbert_{\overline{1,2}}$, for which $\mathcal{A}_2=\identity_1\otimes\mathcal{B}(\hilbert_2)\otimes\identity_{\overline{1,2}}$. So far we have $\hilbert\simeq\hat{\hilbert}_1\otimes\hat{\hilbert}_2\otimes\hilbert_{\overline{1,2}}$. With the introduction of each additional algebra, we obtain another factor in $\hilbert$. Continuing in this way, completeness of the set of $\mathcal{A}_j$ ensures that once all $\mathcal{A}_j$ have been included, $\hilbert$ will have been decomposed 
as $\hilbert\simeq\bigotimes_j\hat{\hilbert}_j$.
\end{proof}

\vspace*{1mm}

\noindent The following two lemmas will be used to formulate the decomposition of Proposition \ref{thm:lsalgdecomp}. The latter will then be used for proving Theorem \ref{thm:factorization}.

\begin{lem}
\label{thm:stateprojker}
Consider a Hilbert space $\hilbert\simeq\bigotimes_i \hilbert_i$, a neighborhood $\neigh_k$ containing $\hilbert_p$, and 
$\ket{\psi}\in\hilbert$. Then, the subsystem kernel of $\ketbra{\psi}$ on $p$ coincides with the subsystem kernel of the neighborhood projector $\Pi_k$ on $p$.
\end{lem}

\noindent
{\em Proof.}
With $p\in\neigh_k$, $ \ker(\tr{\overline{p}}{\ketbra{\psi}})=\ker(\tr{\overline{p}}{\rho_{\neigh_k}}).$
Using the spectral decomposition $\rho_{\neigh_k}=\sum_j\lambda_j \ketbra{j}$ along with properties of the kernel function, we have: 
\begin{align*}
\hspace*{1cm}\ker(\tr{\overline{p}}{\ketbra{\psi}})&=\ker(\sum_j\lambda_j\tr{\overline{p}}{\ketbra{j}})\nonumber\\
&=\ker(\sum_j\tr{\overline{p}}{\ketbra{j}})\nonumber\\
&=\ker( \tr{ \overline{p}}{ \tr{ \overline{\neigh}_k}{\tilde{\Pi}_k\otimes\identity_{\overline{\neigh}_k}}}) \nonumber\\
&=\ker(\tr{\overline{p}}{\Pi_k}).\hspace*{2.2cm}\Box
\end{align*}

\begin{lem}
\label{thm:ptracelem}
Given a positive-semidefinite operator $P$ acting on $\hilbert_A\otimes\hilbert_B$, let $P_A \equiv \tr{B}{P}$. Then, $\ker(P_A)=\ker(\Sigma_A(P))$.
\end{lem}
\begin{proof}
The direction $\ker(P_A)\supseteq\ker(\Sigma_A(P))$ is trivial since $P_A\in\Sigma_A(P)$. For $\ker(P_A)\subseteq\ker(\Sigma_A(P))$, assume $P_A\ket{v}=0$. Since $P_A\geq 0$, this is equivalent to $\tr{}{\ketbra{v}P_A}=0$. In terms of $P$ then, we have $\tr{}{\ketbra{v}\otimes\identity P}=0$. Let $\{E_i\}_{i=1}^{d_B^2}$ constitute an informationally complete POVM on $\hilbert_B$ (i.e., $\textup{span}\{E_i\}=\mathcal{B}(\hilbert_B)$). Then $\sum_i\tr{}{\ketbra{v}\otimes E_i P}=0$. Since each term must be non-negative, we have $\tr{}{\ketbra{v}\otimes E_i P}=0$ for all $i$. We may rewrite this as $\bra{v}\tr{B}{(\identity\otimes\sqrt{E_i})P(\identity\otimes\sqrt{E_i})}\ket{v}=0$, which implies $\tr{B}{(\identity\otimes\sqrt{E_i})P(\identity\otimes\sqrt{E_i})}\ket{v}=0$ for all $i$. Since the POVM is informationally complete, $$\textup{span}\{\tr{B}{(\identity\otimes E_i)P}|i=1,\ldots,d_B^2\}=\Sigma_A(P).$$ Thus, $\ket{v}\in\ker(\Sigma_A(P))$.
\end{proof}

\begin{prop}
\label{thm:lsalgdecomp}
Let $\bigotimes_{i=1}^N\hilbert_i$ be a Hilbert space with a neighborhood structure $\neigh$, and let $\ket{\psi}$ be any state in $\hilbert$. For any neighborhood $\neigh_k$ containing a system $\hilbert_p$, consider the reduced state $\rho_p=\tr{}{\ketbra{\psi}}$ and the decomposition $\hilbert_p\simeq\textup{supp}(\rho_p)\oplus\ker(\rho_p)$. Then there exists a decomposition $\hilbert_p\simeq(\bigoplus_l \hilbert_l\otimes\hilbert'_l)\oplus \ker(\rho_p)$ such that
\begin{equation*}
\textup{alg}\{\Sigma_p(\Pi_k)\}=\Big( \bigoplus_{l}\mathcal{B}(\hilbert_l^p)\otimes\identity_{{\hilbert'}_l^p}\Big)
\oplus\textup{span}\{\identity\}.
\end{equation*}
\end{prop}

\begin{proof}
The above decomposition is ensured as long as $\textup{alg}\{\Sigma_p(\Pi_k)\}$ commutes with all of $\identity_{\textup{supp}(\rho_p)}\oplus\mathcal{B}(\ker(\rho_p))$. We show that an arbitrary basis element in $\identity_{\textup{supp}(\rho_p)}\oplus\mathcal{B}(\ker(\rho_p))$ commutes with all elements in $\Sigma_p(\Pi_k)$. Consider the non-orthonormal basis $\{\identity,\ket{\alpha}\bra{\beta}\}$, where $\ket{\alpha},\ket{\beta}$ are basis elements of $\ker(\rho_p)$. We need only verify that elements $\ket{\alpha}\bra{\beta}$ commute with $\Sigma_p(\Pi_k)$, as $\identity$ does trivially. Since $p\in\neigh_k$, we may apply Lemma \ref{thm:stateprojker} to obtain that $\ker(\rho_p)=\ker(\tr{\overline{p}}{\Pi_k})$. From Lemma \ref{thm:ptracelem} we have $\ker(\tr{\overline{p}}{\Pi_k})=\ker(\Sigma_p(\Pi_k))$. Thus, $\ket{\alpha},\ket{\beta}\in \ker(\Sigma_p(\Pi_k))$, ensuring that $\ket{\alpha}\bra{\beta}\in \Sigma_p(\Pi_k)'$.
\end{proof}

\vspace*{1mm}

\noindent 
{\bf Theorem \ref{thm:factorization} (Algebraic factorization RFTS).}
{\em Let $\ket{\psi}$ on (coarse-grained) subsystems $\bigotimes_{i=1}^N\hilbert_i$ be QLS with respect to $\neigh$ and 
let the neighborhood algebras $\mathcal{A}_j$ be commuting and complete on the local support space $\tilde{\hilbert}$.
Then $\ket{\psi}$ admits a decomposition \[\ket{\psi}=0\oplus\bigotimes_j \ket{\hat{\psi}_j},\] 
with respect to the neighborhood algebra-induced factorization 
$\hilbert\simeq\hilbert^0\oplus (\bigotimes_j\hat{\hilbert}_j)$, and is thus RFTS.}

\begin{proof}
Completeness and commutativity of the $\mathcal{A}_j$ induce the decomposition 
$\tilde{\hilbert}\simeq\bigotimes_j\hat{\hilbert}_j$.
The latter ensures that each $\mathcal{A}_j$ is of the form $\identity_0\oplus\mathcal{B}(\hat{\hilbert}_j)\otimes\identity_{\overline{j}}$. Each $\Pi_k$ commutes with all elements in $\mathcal{A}_j$ for $j\neq k$. This can only be the case if $\Pi_k$ acts as identity on each factor $\hat{\hilbert}_j$ with $j\neq k$,
\( \Pi_k=0\oplus\hat{\Pi}_k\otimes\bigotimes_{j\neq k}\identity, \)
where we have used the fact that the $\Pi_j$ do not have support on the local kernel space $\hilbert_0$.
Thus, the $\Pi_k$ are mutually commuting with one another. This commutativity along with 
{asymptotic QLS} [Eq. (\ref{eq:qls})] 
ensures that
\( \Pi_1\Pi_2\ldots\Pi_T=0\oplus\bigotimes_j \hat{\Pi}_j = \ketbra{\psi}. \)
The trace of the left hand side is the product of ranks of projectors $\hat{\Pi}_j$ and is equal to the trace of $\ketbra{\psi}$, which is 1. Hence, each projector satisfies $\hat{\Pi}_j=\ketbra{\hat{\psi}_j}$. Thus,
\( \ket{\psi}=0\oplus\bigotimes_j \ket{\hat{\psi}_j}.\)
With this factorization of $\ket{\psi}$, as well as the fact that $\identity_0\oplus\mathcal{B}(\hat{\hilbert}_j)\otimes\identity_{\overline{j}}\leq\mathcal{B}(\hilbert_{\neigh_j})\otimes\identity_{\overline{\neigh}_j}$ for each $j$ (by construction), Thm. \ref{thm:lrhilbertdecomp} ensures the $\ket{\psi}$ is RFTS with respect to $\neigh$, as desired.
\end{proof}


\noindent 
$\bullet$ \textbf{Matching overlap condition for RFTS:}

\vspace*{1mm}

\noindent 
{\bf Theorem \ref{thm:matchingoverlap} (\textbf{Matching overlap RFTS}).}
{\em Assume $\ket{\psi}$ on (coarse-grained) subsystems $\bigotimes_{i=1}^N\hilbert_i$  is QLS with respect to $\neigh$, which satisfies the matching overlap condition. If $[\Pi_j,\Pi_k]=0$ for all pairs of neighborhood projectors, then $\ket{\psi}$ is RFTS.}

\begin{proof}
We obtain a decomposition of each $\hilbert_p$ that constitutes a global change of basis leading to a neighborhood factorization as in Prop. \ref{thm:lrhilbertdecomp} which implies RFTS. Consider an arbitrary coarse-grained particle $p$ with Hilbert space $\hilbert_p$. The decomposition of $\hilbert_p$ is induced by the algebra $\textup{alg}\{\Sigma_p(\Pi_k)\}_{\neigh_k\ni p}$. By Prop. \ref{thm:lsalgdecomp}, each $\textup{alg}\{\Sigma_p(\Pi_k)\}$ is contained in $\mathcal{B}(\textup{supp}(\rho_p))\oplus\textup{span}\{\identity_{\ker(\rho_p)}\}$. We show that, furthermore, $\textup{alg}\{\Sigma_p(\Pi_k)\}_{\neigh_k\ni p}=\mathcal{B}(\textup{supp}(\rho_p))\oplus\textup{span}\{\identity_{\ker(\rho_p)}\}$, by establishing that its center is equal to $\textup{span}\{\identity,\identity_{\textup{supp}(\rho_p)}\oplus 0\}.$

Assuming otherwise, there exists an $X=\tilde{X}\oplus 0 \notin \textup{span}\{\identity,\identity_{\textup{supp}(\rho_p)}\oplus 0\}$ such that $X\in\Sigma_p(\Pi_k)'$ for each $\neigh_k\ni p$. Then, $[\identity_{\overline{p}}\otimes X_p,\Pi_k]=0$ for all $\neigh_k$ (including $\neigh_k$ \reflectbox{$\notin$}$p$). Since $X$ acts non-trivially on $\textup{supp}(\rho_p)$, we have $\identity_{\overline{p}}\otimes X_p \ket{\psi} = \ket{\tau}\notin\textup{span}(\ket{\psi})$. Since $\ket{\psi}$ 
satisfies Eq. (\ref{eq:qls}), it is the only vector for which $\Pi_k\ket{\psi}=\ket{\psi}$ for all neighborhoods $\neigh_k$. However, for $\ket{\tau}$, $ \Pi_k\ket{\tau}=
\Pi_k(\identity_{\overline{p}}\otimes X_p)\ket{\psi}=(\identity_{\overline{p}}\otimes X_p)\Pi_k\ket{\psi}
=(\identity_{\overline{p}}\otimes X_p)\ket{\psi}=\ket{\tau},$
which is a contradiction. Hence, no such $X$ can exist, implying that the center of $\textup{alg}\{\Sigma_p(\Pi_k)\}_{\neigh_k\ni p}$ is equal to $\textup{span}\{\identity,\identity_{\textup{supp}(\rho_p)}\oplus 0\}$. Together with the fact that $\textup{alg}\{\Sigma_p(\Pi_k)\}_{\neigh_k\ni p}$ is contained in $\mathcal{B}(\hilbert_{\textup{supp}(\rho_p)})\oplus\textup{span}\{\identity_{\ker(\rho_p)}\}$, this ensures that 
\[ \textup{alg}\{\Sigma_p(\Pi_k)\}_{\neigh_k\ni p}=\mathcal{B}(\hilbert_{\textup{supp}(\rho_p)})\oplus\textup{span}\{\identity_{\ker(\rho_p)}\}. \]
As described, the matching overlap condition ensures that the intersection of any non-disjoint neighborhoods $\neigh_j$ and $\neigh_k$ is some coarse-grained particle $p$. 
Thus, from $[\Pi_j,\Pi_k]=0$, we have $[\Sigma_p(\Pi_j),\Sigma_p(\Pi_k)]=0$, abusing notation. Hence, for any two neighborhoods $\neigh_j$ and $\neigh_k$ containing $p$, we have $\textup{alg}\{\Sigma_p(\Pi_j)\}\leq\textup{alg}\{\Sigma_p(\Pi_k)\}'$. The algebra $\textup{alg}\{\Sigma_p(\Pi_k)\}_{\neigh_k\ni p}$, then, is seen to be generated by a finite number of mutually commuting algebras. Given the form of this algebra in 
the above equation, these generating subalgebras $\textup{alg}\{\Sigma_p(\Pi_k)\}$ can only mutually commute if $\hilbert_{\textup{supp}(\rho_p)}=\bigotimes_{k|\neigh_k\ni p}\hat{\hilbert}^k_p$, whereby 
$$\textup{alg}\{\Sigma_p(\Pi_k)\}=(\mathcal{B}(\hat{\hilbert}^k_p)\otimes\identity_{\hat{\hilbert}^{\overline{k}}_p})\oplus\textup{span}\{\identity_{\ker(\rho_p)}\},$$
for each neighborhood $\neigh_k\ni p$. 

We have obtained a decomposition for each coarse-grained particle Hilbert space $\hilbert_p\simeq(\bigotimes_{k|\neigh_k\ni p}\hat{\hilbert}^k_p)\oplus\hilbert_{\ker(\rho_p)}$. Thus, the global Hilbert space decomposes as
\begin{align*}
\hilbert&\simeq\bigotimes_p \hilbert_p\simeq \bigotimes_p\bigg((\bigotimes_{k|\neigh_k\ni p}\hat{\hilbert}^k_p)\oplus\hilbert_{\ker(\rho_p)}\bigg)\nonumber\\
&\simeq \bigg(\bigotimes_p\bigotimes_{k|\neigh_k\ni p}\hat{\hilbert}^k_p\bigg)\oplus\hilbert_{0} 
\simeq \bigg(\bigotimes_k\bigotimes_{p\in\neigh_k}\hat{\hilbert}^k_p\bigg)\oplus\hilbert_{0}\nonumber\\ 
&\equiv \bigg(\bigotimes_k\hat{\hilbert}^k\bigg)\oplus\hilbert_{0}.
\end{align*}
By the way this decomposition was formed, the $\Pi_k$ act trivially on all but one of the virtual factors, $\Pi_k=0\oplus\hat{\Pi}_k\bigotimes_{j\neq k}\hat{\identity}_j$. Hence, $\ket{\psi}$ satisfying Eq. (\ref{eq:qls}) implies that
\( \Pi_1\Pi_2\ldots\Pi_T=0\oplus\bigotimes_j \hat{\Pi}_j = \ketbra{\psi}.\)
Similar to the proof of Thm. \ref{thm:factorization}, the trace of the left hand side is the product of ranks of projectors $\hat{\Pi}_j$ and is equal to the trace of $\ketbra{\psi}$, which is 1. Hence, each projector satisfies $\hat{\Pi}_j=\ketbra{\hat{\psi}_j}$. Thus,
\( \ket{\psi}=0\oplus\bigotimes_j \ket{\hat{\psi}_j}.\)
With this factorization of $\ket{\psi}$, as well as the fact that $\identity_0\oplus\mathcal{B}(\hat{\hilbert}_j)\otimes\identity_{\overline{j}}\leq\mathcal{B}(\hilbert_{\neigh_j})\otimes\identity_{\overline{\neigh}_j}$ for each $j$, Thm. \ref{thm:lrhilbertdecomp} ensures the $\ket{\psi}$ is RFTS with respect to $\neigh$.
\end{proof}

\vspace*{1mm}

\noindent 
$\bullet$ \textbf{Efficiency of FTS/RFTS:}

\vspace*{1mm}

\noindent 
{{\bf Proposition \ref{thm:latticecomplexity} (\textbf{Lattice circuit size scaling}).}  
{\em Consider an $N$-dimensional subset and neighborhood structure $\neigh^{(N)}$ on 
a $m$-dimensional lattice. If $\ket{\psi^{(N)}}$ 
is RFTS with respect to $\neigh^{(N)}$, then $\ket{\psi}$ can be stabilized by a dissipative circuit of 
size at most $|\neigh^0|(N/ {\sf c})$ and depth at most $|\neigh^0|\textup{diam}(\neigh^0)^m$.}

\begin{proof}
For any RFTS state, the circuit size is equal to the number of neighborhoods. From the unit cell definition, 
the latter is $|\neigh^0|(N/{\sf c})$.
To bound the depth of the circuit, we devise a scheme which parallelizes the circuit to one with constant depth. 
Specifically, we show that there exists a partitioning of the neighborhoods of $\neigh$, and hence $\neigh^{(N)}$, 
into $|\neigh^0|\textup{diam}(\neigh^0)^m$ parts, such that each part consists of a set of mutually disjoint neighborhoods. 
If the union of the unit cell neighborhoods $\neigh^0$ is translated in any direction a distance ${\sf d} \equiv \textup{diam}(\neigh^0)$, 
the resulting set is disjoint from $\neigh^0$. In particular, if we select a single neighborhood $\neigh_k\in\neigh^0$ and construct the set of neighborhoods generated by linear combinations of ${\sf d}\hat{e}_i$ for each $i=1,\ldots,m$, the neighborhoods in this set are ensured to 
be disjoint from one another. Hence, the sequence of the corresponding neighborhood maps act in parallel and constitute a layer of the circuit. This set of neighborhoods is generated by a subgroup $({\sf d}\mathbb{Z})^m={\sf d}\mathbb{Z}\times\ldots\times {\sf d}\mathbb{Z}$ of the discrete translation group $\mathbb{Z}^m\simeq {\mathbb L}$. Therefore, the translated copies of $\neigh_k$ for which this did not account each correspond to a coset of $({\sf d}\mathbb{Z})^m$ in $\mathbb{Z}^m$ with respect to elements 
$\vec{\ell}=(\ell_1,\ldots,\ell_m)\in\mathbb{Z}^m$. This coset group is isomorphic to $\mathbb{Z}_{\sf d}^m=
\mathbb{Z}_{\sf d}\times\ldots\times\mathbb{Z}_{\sf d}$, whose size is $|\mathbb{Z}_{\sf d}^m |={\sf d}^m$. 
Using group-action notation, we denote 
the $\vec{\ell}$-translated version of $\neigh_0$ as $\vec{\ell}\neigh_0$. 
Each layer of neighborhood maps corresponds then to a set of disjoint neighborhoods,
\( \vec{\ell}({\sf d}\mathbb{Z})^m\neigh_k,\,\, k=1,\ldots,|\neigh^0|,\) \(\vec{\ell}\in\mathbb{Z}_{\sf d}^m.\)
\noindent
Each neighborhood is accounted for, and there are $|\neigh^0|{\sf d}^m$ layers. With this scheme, we define
\[ \mathcal{E}_{\vec{\ell},k}\equiv \prod_{\vec{v}\in({\sf d} \mathbb{Z})^m}
\mathcal{E}_{(\vec{v}+\vec{\ell})\neigh_k}, \quad k=1,\ldots,|\neigh^0|,\;\vec{\ell}\in\mathbb{Z}_{\sf d}^m .
\]
The sequence of neighborhood maps that prepares the target state can then be parallelized as 
$\rho\trn{}=\mathcal{E}_N \circ\ldots \circ\mathcal{E}_1=
\prod_{k=1}^{|\neigh^0|}\prod_{\vec{\ell}\in\mathbb{Z}_{\sf d}^m}\mathcal{E}_{\vec{\ell},k}$, 
of which there are $|\neigh^0| \textup{diam}(\neigh^0)^m$ parallelized maps, as claimed.
\end{proof}}


\noindent 
$\bullet$ \textbf{RFTS implies rapid mixing:}

\vspace*{1mm}

\noindent 
{\bf Proposition \ref{thm:systemsizebound}} (\textbf{Commuting Liouvillian contraction bound}).
{\em Let $\{\mathcal{L}_j\}$ be uniformly-bounded Liouvillians, each acting on a neighborhood of uniformly-bounded size. Assume that the spectral gaps obey $\bar\lambda(\mathcal{L}_j)\geq\nu>0$, for all $j$. Then, there exists $R>0$ such that for any subset $\mathcal{S}$ of mutually commuting $\mathcal{L}_j$, we have:}
\begin{equation*}
\eta(e^{\sum_{{\cal L}_j\in\mathcal{S}}\mathcal{L}_jt})\leq |\mathcal{S}|Re^{-\nu t}.
\end{equation*}

\noindent
{\em Proof.}
From Theorem \ref{thm:commutingcontraction}, commutativity of the terms implies $\eta(e^{\mathcal{L} t}) \leq \sum_{\mathcal{S}} \eta(e^{\mathcal{L}_jt})$. With $\nu<\lambda({\cal L}_j)$, Theorem \ref{thm:commutingcontraction} also ensures that, for each $\mathcal{L}_j\in\{\mathcal{L}_j\}$, there exists $R_j>0$ such that $\eta(e^{\mathcal{L}_jt})\leq R_je^{-\nu t}$. In \cite{Reeb2012}, it is shown that, for fixed $\nu$, $R_j$ is upper-bounded by a function of order $d_j^{d_j^2}$, where $d_j$ is the dimension of the system on which $\mathcal{L}_j$ acts. 
{Let $B \geq d_j$ be the uniform subsystem dimension bound. Then, we can find constants $R$ and $c$ such that, for all $j$, 
$R>c B^{B^2}>c d_j^{d_j^2}>R_j$.} Hence,
\begin{align*}
\hspace*{1cm}\eta(e^{\mathcal{L} t}) &\leq \sum_{\mathcal{S}} \eta(e^{\mathcal{L}_jt})\leq\sum_{\mathcal{S}} R_je^{-\nu t}\nonumber\\
&\leq\sum_{\mathcal{S}} Re^{-\nu t}=|\mathcal{S}|Re^{-\nu t}.\hspace*{2.4cm}\Box
\end{align*}

\noindent {\bf Theorem \ref{thm:RFTSrapidmixing}} (\textbf{Rapid mixing for commuting RFTS}).
{\em Consider a scalable family of $ \ket{\psi^{(s)}}$ that is made RFTS  by a set of commuting 
neighborhood maps $\{\mathcal{E}^{(s)}_k\}$. 
Assume that there exists $\nu>0$, such that each $\lambda\in\textup{eig}(\mathcal{E}^{(s)}_k)$ satisfies either 
$\lambda=1$ or $| \lambda |< 1-\nu$. Then, there exists a family of bounded-norm QL Liouvillians ${\cal L}^{(s)}$ 
satisfying rapid mixing with respect to $ \ket{\psi^{(s)} }$.}

\vspace*{1mm}

\begin{proof}
For each $s$, let the neighborhood Liouvillian operators 
$\mathcal{L}^{(s)}_k\equiv\mathcal{E}^{(s)}_k-\mathcal{I}^{(s)}$. 
These Liouvillians have bounded norm and 
the spectral gap $\bar{\lambda}_k$ of each $\mathcal{L}^{(s)}_k$  satisfies $\bar{\lambda}_k^{(s)} > \nu>0$. 
Take $\{\mathcal{L}^{(s)}_k\}_{k,{(s)}}$ as a set of Liouvillians, and define the sequence of subsets 
$\mathcal{S}^{(s)}=\{\mathcal{L}^{(s)}_k\}_k$, indexed by ${s}$. Then, for each $s$, the global Liouvillian is
\( \mathcal{L}^{(s)}=\sum_k \mathcal{L}^{(s)}_k. \)
For each ${s}$, this Liouvillian is a sum of commuting terms with $B$ and $\nu$ satisfying the conditions in Proposition 
\ref{thm:systemsizebound} for some finite prefactor $R$. Thus,
\begin{align*}
\eta( e^{ \mathcal{L}^{(s)} t }  )\leq |\mathcal{S}^{(s)} |R \,e^{-\nu t} = & \\ |\mathcal{\neigh}^{(s)}| R\, e^{-\nu t}
&\leq R\, b \, N^{(s)}
e^{-\nu t}.
\end{align*}
Identifying $c=Rb$, $\gamma= \nu$, and $\delta=1$ in Definition \ref{def:rm} verifies rapid mixing, as desired.
\end{proof}

\end{appendix}


\begin{thebibliography}{79}
\bibitem[{rem({\natexlab{a}})}]{remarkOnestep}
\bibinfo{note}{If $\rho$ is the desired quantum state, it suffice to consider a
  CPTP map of the form $\mathcal{E}\equiv \rho \text{Tr}(\cdot)$, where Tr
  denotes the trace functional.}
\bibitem[{rem({\natexlab{c}})}]{remarkExample}
\bibinfo{note}{For instance, in the scheme presented for the three spin-1/2
  AKLT target state in Fig. \ref{fig:FTSAKLT}, the cooling rate $\log_d
  r=\log_3 (4)$, as $\mathcal{W}$ maps four such copies of $\Sigma^0$
  (including itself) into $\Sigma^0.$ Each global unitary stabilizer map is
  decomposable into no more than $2\cdot 26^2$ neighborhood stabilizer
  unitaries, giving a circuit size upper bounded by $1+2\cdot 26^2+1+ 2\cdot
  26^2+1=2,707$.}
\end{thebibliography}

\end{document}